\documentclass[10pt,conference,final]{IEEEtranx}
\pdfoutput=1
\ifCLASSOPTIONdraft
 \def\baselinestretch{1}
\fi

\usepackage[utf8]{inputenc}
\usepackage[T1]{fontenc}

\usepackage{mathptmx}
\usepackage{xspace}
\usepackage[usenames,dvipsnames,svgnames,table]{xcolor}
\usepackage[final,stretch=10,shrink=10]{microtype}
\usepackage{ifdraft}
\usepackage{tikz}
\usetikzlibrary{backgrounds,snakes,arrows,decorations.markings,calc}
\usepackage[noend]{algpseudocode}
\usepackage{algorithm}
\usepackage{algorithmicx}
\usepackage{varwidth} %
\usepackage[caption=false]{subfig} %
\usepackage{todonotes}
\usepackage[compress]{cite} %
\usepackage{placeins} %

\usepackage{calc}

\usepackage{captcont}
\usepackage{caption}
\DeclareCaptionLabelSeparator{dot}{.\enspace}
\clearcaptionsetup{figure}
\usepackage[thmmarks]{ntheorem}
\usepackage{macros}
\usepackage{amsmath}
\usepackage{amssymb}

\usepackage[english]{babel}
\usepackage[final,breaklinks]{hyperref}
\usepackage{color}
\definecolor{linkcolor}{rgb}{0,0,0.5}
\hypersetup{
 colorlinks=true,
  linkcolor=linkcolor,
  anchorcolor=linkcolor,
  citecolor=linkcolor,
  filecolor=linkcolor,
  menucolor=linkcolor,
  runcolor=linkcolor,
  urlcolor=linkcolor,
}

\usepackage{linegoal} %

\ifdraft{
  \usepackage[scrtime]{prelim2e}
  \usepackage{showlabels}
  \pagestyle{plain}
  \ifdefined\fastmode
    \usepackage{comment}
    \excludecomment{figure}
    
  \fi
}{}

\begin{document}
\hyphenation{Brow-serID}
\hyphenation{in-fra-struc-ture}
\hyphenation{brow-ser}
\hyphenation{doc-u-ment}
\hyphenation{Chro-mi-um}
\hyphenation{meth-od}
\hyphenation{sec-ond-ary}
\hyphenation{Java-Script}
\hyphenation{Mo-zil-la}
\hyphenation{post-Mes-sage}

\title{The Web SSO Standard \emph{OpenID Connect}:\\In-Depth Formal Security Analysis\\and Security Guidelines}
 \author{\IEEEauthorblockN{Daniel Fett, Ralf K\"usters, and Guido Schmitz}
   \IEEEauthorblockA{
     University of Stuttgart,
     Germany\\
     Email: \texttt{\{daniel.fett,ralf.kuesters,guido.schmitz\}@sec.uni-stuttgart.de}
   }
 }

\maketitle

\ifdraft{
\listoftodos

}{ }

\begin{abstract}
  Web-based single sign-on (SSO) services such as \emph{Google Sign-In} and \emph{Log In with Paypal} are based on the \emph{OpenID Connect} protocol. This protocol enables so-called relying parties to delegate user authentication to so-called identity providers. OpenID Connect is one of the newest and most widely deployed single sign-on protocols on the web. Despite its importance, it has not received much attention from security researchers so far, and in particular, has not undergone any rigorous security analysis.

In this paper, we carry out the first in-depth security analysis of OpenID Connect. To this end, we use a comprehensive generic model of the web to develop a detailed formal model of OpenID Connect. Based on this model, we then precisely formalize and prove central security properties for OpenID Connect, including authentication, authorization, and session integrity properties.

In our modeling of OpenID Connect, we employ security measures in order to avoid attacks on OpenID Connect that have been discovered previously and new attack variants that we document for the first time in this paper. Based on these security measures, we propose security guidelines for implementors of OpenID Connect. Our formal analysis demonstrates that these guidelines are in fact effective and sufficient.

\end{abstract}

\section{Introduction}
\label{sec:introduction}

OpenID Connect is a protocol for delegated authentication in the
web: A user can log into a relying party (RP) by
authenticating herself at a so-called identity provider (IdP). For
example, a user may sign into the website \nolinkurl{tripadvisor.com}
using her Google account.

Although the names might suggest otherwise, OpenID Connect (or
\emph{OIDC} for short) is not based on the older OpenID
protocol. Instead, it builds upon the OAuth~2.0 framework, which
defines a protocol for delegated \emph{authorization} (e.g., a user
may grant a third party website access to her resources at Facebook).
While OAuth~2.0 was not designed to provide \emph{authentication},
it has often been used for this purpose as well, leading to several
severe security flaws in the past \cite{Bradley-OAuth-Authentication-2012,Wangetal-USENIX-Explicating-SDKs-2013}.

OIDC was created not only to retrofit authentication into OAuth~2.0 by
using cryptographically secured tokens and a precisely defined method
for user authentication, but also to enable additional important features. For
example, using the \emph{Discovery} extension, RPs can automatically
identify the IdP that is responsible for a given identity. With the
\emph{Dynamic Client Registration} extension, RPs do not need a manual set-up
process to work with a specific IdP, but can instead register
themselves at the IdP on the fly.

Created by the OpenID Foundation and standardized only in
November 2014, OIDC is already very widely used. Among others, it is used and
supported by Google, Amazon, Paypal, Salesforce, Oracle, Microsoft,
Symantec, Verizon, Deutsche Telekom, PingIdentity, RSA Security,
VMWare, and IBM. Many corporate and end-user single sign-on solutions
are based on OIDC, for example, well-known services such as
\emph{Google Sign-In} and \emph{Log In with Paypal}.

Despite its wide use, OpenID Connect has not received much attention
from security researchers so far (in contrast to OpenID and OAuth~2.0). In
particular, there have been no formal analysis efforts for OpenID
Connect until now. In fact, the only previous works on the security of
OpenID Connect are a large-scale study of deployments of Google's implementation of
OIDC performed by Li and Mitchell~\cite{LiMitchell-DIMVA-2016} and an 
informal evaluation by Mainka et
al.~\cite{MainkaMladenovSchwenkWich-EuroSP-2017}.

In this work, we aim to fill the gap and formally verify the security
of OpenID Connect.

\mpar{Contributions of this Paper.} We provide the first in-depth
formal security analysis of OpenID Connect. Based on a comprehensive
formal web model and strong attacker models, we analyze the security
of all flows available in the OIDC standard, including many of the
optional features of OIDC and the important Discovery and Dynamic
Client Registration extensions. More specifically, our contributions
are as follows.

\paragraph{Attacks on OIDC and Security Guidelines}
We first compile an overview of attacks on OIDC, common pitfalls, and
their respective mitigations. Most of these attacks were documented
before, but we point out new attack variants and aspects.

Starting from these attacks and pitfalls, we then derive security
guidelines for implementors of OIDC. Our guidelines are backed-up by
our formal security analysis, showing that the mitigations that
we propose are in fact effective and sufficient.

\paragraph{Formal model of OIDC} Our formal analysis of OIDC is based
on the expressive Dolev-Yao style model of the web infrastructure (FKS
model) proposed by Fett, K{\"u}sters, and
Schmitz~\cite{FettKuestersSchmitz-SP-2014}. This web model is designed
independently of a specific web application and closely mimics
published (de-facto) standards and specifications for the web, for
instance, the HTTP/1.1 and HTML5 standards and associated (proposed)
standards. It is the most comprehensive web model to date. Among
others, HTTP(S) requests and responses, including several headers,
such as cookie, location, referer, authorization, strict transport
security (STS), and origin headers, are modeled. The model of web
browsers captures the concepts of windows, documents, and iframes,
including the complex navigation rules, as well as modern
technologies, such as web storage, web messaging (via postMessage),
and referrer policies. JavaScript is modeled in an abstract way by
so-called \emph{scripts} which can be sent around and, among others,
can create iframes, access other windows, and initiate
XMLHttpRequests. Browsers may be corrupted dynamically by the
adversary.

The FKS model has already been used to analyze the security of the
BrowserID single sign-on
system~\cite{FettKuestersSchmitz-SP-2014,FettKuestersSchmitz-ESORICS-BrowserID-Primary-2015},
the security and privacy of the SPRESSO SSO
system~\cite{FettKuestersSchmitz-CCS-2015}, and the security of
OAuth~2.0~\cite{FettKuestersSchmitz-CCS-2016}, each time uncovering
new and severe attacks that have been missed by previous analysis
attempts.

Using the generic FKS model, we build a formal model of OIDC, closely
following the standard. We employ the defenses and mitigations
discussed earlier in order to create a model with state-of-the-art
security features in place. Our model includes RPs and IdPs that
(simultaneously) support all modes of OIDC and can be dynamically
corrupted by the adversary.

\paragraph{Formalization of security properties} Based on this model
of OIDC, we formalize four main security properties of OIDC:
authentication, authorization, session integrity for authentication,
and session integrity for authorization. We also formalize further
OIDC specific properties.

\paragraph{Proof of Security for OpenID Connect} Using the model and
the formalized security properties, we then show, by a manual yet
detailed proof, that OIDC in fact satisfies the security
properties. This is the first proof of security of OIDC. Being
based on an expressive and comprehensive formal model of the web,
including a strong attacker model, as well as on a modeling of OpenID
Connect which closely follows the standard, our security analysis
covers a wide range of attacks.

\mpar{Structure of this Paper.} We provide an informal description of
OIDC in Section~\ref{sec:openid-connect}. Attacks and security
guidelines are discussed in Section~\ref{sec:attacks-impl-guid}. In
Section~\ref{sec:fks-web-model}, we briefly recall the FKS model. The
model and analysis of OIDC are then presented in
Section~\ref{sec:analysis}. Related work is discussed in
Section~\ref{sec:related-work}. We conclude in
Section~\ref{sec:conclusion}. All details of our work, including the
proofs, are provided in the appendix.

\section{OpenID Connect}
\label{sec:openid-connect}

The OpenID Connect protocol allows users to authenticate to RPs using
their existing account at an IdP.\footnote{Note that the OIDC standard
  also uses the terms \emph{client} for RP and \emph{OpenID provider
    (OP)} for the IdP. We here use the more common terms RP and IdP.} (Typically, this is an email account at the IdP.)
OIDC was defined by the OpenID Foundation in a \emph{Core} document
\cite{openid-connect-core-1.0} and in extension documents (e.g.,
\cite{openid-connect-discovery-1.0,
  openid-connect-dynamic-client-registration-1.0}). Supporting
technologies were standardized at the IETF, e.g.,
\cite{rfc7519-jwt,rfc7033-webfinger}. (Recall that OpenID Connect is
not to be confused with the older OpenID standards, which are very
different to OpenID Connect.)

Central to OIDC is a cryptographically signed document, the \emph{id
  token}. It is created by the user's IdP and serves as a one-time
proof of the user's identity to the RP.

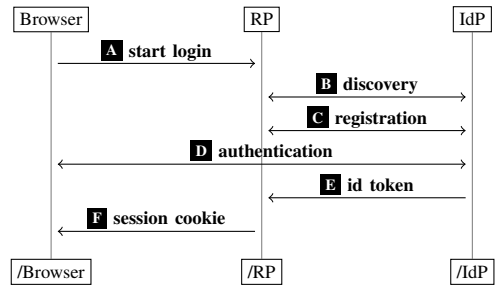
\begin{figure}[t!]
  \centering
   \scriptsize{ \newlength\blockExtraHeightAeHFJDHdfHddAECfGJGBeaaHCcCJfCaEE
\settototalheight\blockExtraHeightAeHFJDHdfHddAECfGJGBeaaHCcCJfCaEE{\parbox{0.4\linewidth}{}}
\setlength\blockExtraHeightAeHFJDHdfHddAECfGJGBeaaHCcCJfCaEE{\dimexpr \blockExtraHeightAeHFJDHdfHddAECfGJGBeaaHCcCJfCaEE - 4ex/4}
\newlength\blockExtraHeightBeHFJDHdfHddAECfGJGBeaaHCcCJfCaEE
\settototalheight\blockExtraHeightBeHFJDHdfHddAECfGJGBeaaHCcCJfCaEE{\parbox{0.4\linewidth}{}}
\setlength\blockExtraHeightBeHFJDHdfHddAECfGJGBeaaHCcCJfCaEE{\dimexpr \blockExtraHeightBeHFJDHdfHddAECfGJGBeaaHCcCJfCaEE - 4ex/4}
\newlength\blockExtraHeightCeHFJDHdfHddAECfGJGBeaaHCcCJfCaEE
\settototalheight\blockExtraHeightCeHFJDHdfHddAECfGJGBeaaHCcCJfCaEE{\parbox{0.4\linewidth}{}}
\setlength\blockExtraHeightCeHFJDHdfHddAECfGJGBeaaHCcCJfCaEE{\dimexpr \blockExtraHeightCeHFJDHdfHddAECfGJGBeaaHCcCJfCaEE - 4ex/4}
\newlength\blockExtraHeightDeHFJDHdfHddAECfGJGBeaaHCcCJfCaEE
\settototalheight\blockExtraHeightDeHFJDHdfHddAECfGJGBeaaHCcCJfCaEE{\parbox{0.4\linewidth}{}}
\setlength\blockExtraHeightDeHFJDHdfHddAECfGJGBeaaHCcCJfCaEE{\dimexpr \blockExtraHeightDeHFJDHdfHddAECfGJGBeaaHCcCJfCaEE - 4ex/4}
\newlength\blockExtraHeightEeHFJDHdfHddAECfGJGBeaaHCcCJfCaEE
\settototalheight\blockExtraHeightEeHFJDHdfHddAECfGJGBeaaHCcCJfCaEE{\parbox{0.4\linewidth}{}}
\setlength\blockExtraHeightEeHFJDHdfHddAECfGJGBeaaHCcCJfCaEE{\dimexpr \blockExtraHeightEeHFJDHdfHddAECfGJGBeaaHCcCJfCaEE - 4ex/4}
\newlength\blockExtraHeightFeHFJDHdfHddAECfGJGBeaaHCcCJfCaEE
\settototalheight\blockExtraHeightFeHFJDHdfHddAECfGJGBeaaHCcCJfCaEE{\parbox{0.4\linewidth}{}}
\setlength\blockExtraHeightFeHFJDHdfHddAECfGJGBeaaHCcCJfCaEE{\dimexpr \blockExtraHeightFeHFJDHdfHddAECfGJGBeaaHCcCJfCaEE - 4ex/4}

 \begin{tikzpicture}
   \tikzstyle{xhrArrow} = [color=blue,decoration={markings, mark=at
    position 1 with {\arrow[color=blue]{triangle 45}}}, preaction
  = {decorate}]

    \matrix [column sep={2.8cm,between origins}, row sep=3ex]
  {

    \node[draw,anchor=base](Browser-start-0){Browser}; & \node[draw,anchor=base](RP-start-0){RP}; & \node[draw,anchor=base](IdP-start-0){IdP};\\
\node(Browser-0){}; & \node(RP-0){}; & \node(IdP-0){};\\[\blockExtraHeightAeHFJDHdfHddAECfGJGBeaaHCcCJfCaEE]
\node(Browser-1){}; & \node(RP-1){}; & \node(IdP-1){};\\[\blockExtraHeightBeHFJDHdfHddAECfGJGBeaaHCcCJfCaEE]
\node(Browser-2){}; & \node(RP-2){}; & \node(IdP-2){};\\[\blockExtraHeightCeHFJDHdfHddAECfGJGBeaaHCcCJfCaEE]
\node(Browser-3){}; & \node(RP-3){}; & \node(IdP-3){};\\[\blockExtraHeightDeHFJDHdfHddAECfGJGBeaaHCcCJfCaEE]
\node(Browser-4){}; & \node(RP-4){}; & \node(IdP-4){};\\[\blockExtraHeightEeHFJDHdfHddAECfGJGBeaaHCcCJfCaEE]
\node(Browser-5){}; & \node(RP-5){}; & \node(IdP-5){};\\[\blockExtraHeightFeHFJDHdfHddAECfGJGBeaaHCcCJfCaEE]
\node[draw,anchor=base](Browser-end-1){/Browser}; & \node[draw,anchor=base](RP-end-1){/RP}; & \node[draw,anchor=base](IdP-end-1){/IdP};\\
};
\draw[->] (Browser-0) to node [above=2.6pt, anchor=base]{\alphprotostep{oichl-start} \textbf{start login}} node [below=-8pt, text width=0.5\linewidth, anchor=base]{\begin{center} \end{center}} (RP-0); 

\draw[<->] (RP-1) to node [above=2.6pt, anchor=base]{\alphprotostep{oichl-discovery} \textbf{discovery}} node [below=-8pt, text width=0.5\linewidth, anchor=base]{\begin{center} \end{center}} (IdP-1); 

\draw[<->] (RP-2) to node [above=2.6pt, anchor=base]{\alphprotostep{oichl-registration} \textbf{registration}} node [below=-8pt, text width=0.5\linewidth, anchor=base]{\begin{center} \end{center}} (IdP-2); 

\draw[<->] (Browser-3) to node [above=2.6pt, anchor=base]{\alphprotostep{oichl-auth} \textbf{authentication}} node [below=-8pt, text width=0.5\linewidth, anchor=base]{\begin{center} \end{center}} (IdP-3); 

\draw[->] (IdP-4) to node [above=2.6pt, anchor=base]{\alphprotostep{oichl-token} \textbf{id token}} node [below=-8pt, text width=0.5\linewidth, anchor=base]{\begin{center} \end{center}} (RP-4); 

\draw[->] (RP-5) to node [above=2.6pt, anchor=base]{\alphprotostep{oichl-set-service-cookie} \textbf{session cookie}} node [below=-8pt, text width=0.5\linewidth, anchor=base]{\begin{center} \end{center}} (Browser-5); 

\begin{pgfonlayer}{background}
\draw [color=gray] (Browser-start-0) -- (Browser-end-1);
\draw [color=gray] (RP-start-0) -- (RP-end-1);
\draw [color=gray] (IdP-start-0) -- (IdP-end-1);
\end{pgfonlayer}
\end{tikzpicture}}
  \caption{OpenID Connect --- high level overview.}
  \label{fig:oidc-high-level}
\end{figure}

A high-level overview of OIDC is given in
Figure~\ref{fig:oidc-high-level}. First, the user requests to be
logged in at some RP and provides her email
address~\refalphprotostep{oichl-start}. RP now retrieves operational
information (e.g., some URLs) for the remaining protocol flow
(\emph{discovery},~\refalphprotostep{oichl-discovery}) and registers
itself at the IdP~\refalphprotostep{oichl-registration}. The user is
then redirected to the IdP, where she authenticates
herself~\refalphprotostep{oichl-auth} (e.g., using a password). The
IdP issues an id token to RP~\refalphprotostep{oichl-token}, which RP
can then verify to ensure itself of the user's identity. (The way of
how the IdP sends the id token to the RP is subject to the different
modes of OIDC, which are described in detail later in this section. In
short, the id token is either relayed via the user's browser or it is
fetched by the RP from the IdP directly.) The id token includes an
identifier for the IdP (the \emph{issuer}),\footnote{The issuer
  identifier of an IdP is an HTTPS URL without any query or fragment
  components.} a user identifier (unique at the
respective IdP), and is signed by the IdP. The RP uses the issuer
identifier and the user identifier to determine the user's identity.
Finally, the RP may set a session cookie in the user's browser which
allows the user to access the services of
RP~\refalphprotostep{oichl-set-service-cookie}.

Before we explain the modes of operation of OIDC, we first present
some basic concepts used in OIDC. At the end of this section, we
discuss the relationship of OIDC to OAuth~2.0.

\subsection{Basic Concepts}\label{sec:oidc-basic-concepts}

We have seen above that id tokens are essential to OIDC. Also, to
allow users to use any IdP to authenticate to any RP, the RP needs to
\emph{discover} some information about the IdP. Additionally, the IdP
and the RP need to establish some sort of relationship between each
other. The process to establish such a relationship is called
\emph{registration}. Both, discovery and registration, can be either a
manual task or a fully automatic process. Further, OIDC allows users
to \emph{authorize} an RP to access user's data at IdP on the user's
behalf. All of these concepts are described in the following.

\subsubsection{Authentication and ID Tokens}\label{sec:oidc-basic-concepts-id-tokens}
The goal of OIDC is to \emph{authenticate} a user to an RP, i.e., the
RP gets assured of the identity of the user interacting with the RP.
This assurance is based on id tokens. As briefly mentioned before, an
id token is a document signed by the IdP. It contains several
\emph{claims}, i.e., information about the user and further meta
data. More precisely, an id token contains a user identifier (unique
at the respective IdP) and the issuer identifier of the IdP. Both
identifiers in combination serve as a global user identifier for
authentication. Also, every id token contains an identifier for the RP
at the IdP, which is assigned during registration (see below). The id
token may also contain a nonce chosen by the RP during the
authentication flow as well as an expiration timestamp and a timestamp
of the user's authentication at the IdP to prevent replay
attacks. Further, an id token may contain information about the
particular method of authentication and other claims, such as data
about the user and a hash of some data sent outside of the id token.

When an RP validates an id token, it checks in particular whether the
signature of the token is correct (we will explain below how RP obtains the public key of the IdP), the issuer identifier is the one of
the currently used IdP, the id token is issued for this RP, the nonce
is the one RP has chosen during this login flow, and the token has not
expired yet. If the id token is valid, the RP trusts the claims contained in the id token and is confident in the user's identity.

\subsubsection{Discovery and Registration}\label{sec:oidc-basic-concepts-discovery-registration}
The OIDC protocol is heavily based on redirection of the user's
browser: An RP redirects the user's browser to some IdP and
vice-versa. Hence, both parties, the RP and the IdP, need some
information about the respective URLs (so-called \emph{endpoints})
pointing to each other. Also, the RP needs a public key of the IdP to
verify the signature of id tokens. Further, an RP can contact the IdP
directly to exchange protocol information. This exchange may 
include authentication of the RP at the IdP.

More specifically, an RP and an IdP need to exchange the following
information: (1) a URL where the user can authenticate to the IdP
(\emph{authorization endpoint}), (2) one or more URLs at RP where the
user's browser can be redirected to by the IdP after authentication
(\emph{redirection endpoint}), (3) a URL where the RP can contact the
IdP in order to retrieve an id token (\emph{token endpoint}), (4) the
issuer identifier of the IdP, (5) the public key of the IdP to verify
the id token's signature, (6) an identifier of the RP at IdP
(\emph{client id}), and optionally (7) a secret used by RP to
authenticate itself to the token endpoint (\emph{client secret}).
(Recall that \emph{client} is another term for RP, and in particular
does not refer to the browser.)

This information can be exchanged manually by the administrator of the
RP and the administrator of the IdP, but OIDC also allows one to completely
automate the discovery of IdPs~\cite{openid-connect-discovery-1.0} and
dynamically register RPs at an
IdP~\cite{openid-connect-dynamic-client-registration-1.0}.

During the automated discovery, the RP first determines which IdP is responsible for
the email address provided by the user who wants to log in using the WebFinger protocol~\cite{rfc7033-webfinger}.
As a result, the RP learns the issuer identifier of the IdP and can retrieve the URLs
of the authorization endpoint and the token endpoint from the
IdP. Furthermore, the RP receives a URL where it can retrieve the public key
to verify the signature of the id token (\emph{JWKS URI}), and a URL
where the RP can register itself at the IdP (\emph{client registration
  endpoint}).

If the RP has not registered itself at this IdP before, it starts the
registration ad-hoc at the client registration endpoint: The
RP sends its redirection endpoint URLs to the IdP and receives a new
client id and (optionally) a client secret in return.

\subsubsection{Authorization and Access Tokens}\label{sec:oidc-basic-concepts-authorization}
OIDC allows users to authorize RPs to access the user's data stored at
IdPs or act on the user's behalf at IdPs. For example, a photo
printing service (the RP) might access or manage the user's photos on
Google Drive (the IdP). For authorization, the RP receives a so-called
\emph{access token} (besides the id token). Access tokens follow the
concept of so-called \emph{bearer tokens}, i.e., they are used as the
only authentication component in requests from an RP to an IdP. In our
example, the photo printing service would have to add the access token
to each HTTP request to Google Drive.

\subsection{Modes}\label{sec:oidc-modes}

OIDC defines three  modes: the
\emph{authorization code mode}, the \emph{implicit mode}, and the
\emph{hybrid mode}. While in the authorization code mode, the id token
is retrieved by an RP from an IdP in direct server-to-server communication
(back channel), in the implicit mode, the id token is relayed from an IdP
to an RP via the user's browser (front channel). The hybrid mode is a
combination of both modes and allows id tokens to be exchanged via the
front and the back channel at the same time.

We now provide a detailed description of all three modes.

\subsubsection{Authorization Code Mode}\label{sec:oidc-modes-auth-code}
\begin{figure}[t!]
  \centering
  \input{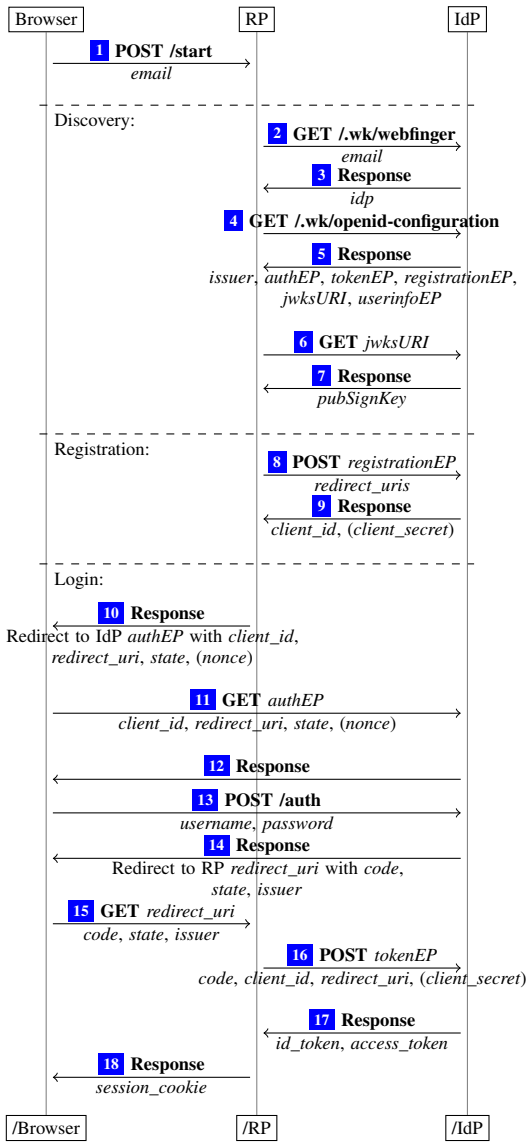}
  \caption{OpenID Connect authorization code mode. Note that data depicted
    below the arrows is either transferred in URI parameters, HTTP headers, or
    POST bodies.}
  \label{fig:oidc-auth-code-flow}
\end{figure}
In this mode, an RP redirects the user's browser to an IdP. At the IdP, the user
authenticates and then the IdP issues a so-called \emph{authorization
  code} to the RP. The RP now uses this code to obtain an id
token from the IdP.

\paragraph{Step-by-Step Protocol Flow}
The protocol flow is depicted in
Figure~\ref{fig:oidc-auth-code-flow}. First, the user starts the login
process by entering her email address\footnote{Note that OIDC also allows other types of user ids, such as
  personal URLs.} in her browser (at some web page of an RP), which
sends the email address to the RP in~\refprotostep{oicacf-start-req}.

Now, the RP uses the OIDC discovery extension
\cite{openid-connect-discovery-1.0} to gather information about the IdP:
As the first step (in this extension), the RP uses the WebFinger mechanism
\cite{rfc7033-webfinger} to discover information about which IdP is
responsible for this email address. For this discovery, the RP contacts the
server of the email domain in~\refprotostep{oicacf-wf-req} (in
the figure, the server of the user's email domain is depicted as the
same party as the IdP). The result of the WebFinger request
in~\refprotostep{oicacf-wf-resp} contains the issuer identifier of the
IdP (which is also a URL). With this information, the RP can continue the
discovery by requesting the OIDC configuration from the IdP
in~\refprotostep{oicacf-conf-req} and~\refprotostep{oicacf-conf-resp}. This configuration contains meta
data about the IdP, including all endpoints at the IdP and a URL where
the RP can retrieve the public key of the IdP (used to later verify the id
token's signature). If the RP does not know this public key yet, the RP
retrieves the key (Steps~\refprotostep{oicacf-jwks-req}
and~\refprotostep{oicacf-jwks-resp}). This concludes the OIDC
discovery in this login flow.

Next, if the RP is not registered at the IdP yet, the RP starts the OIDC dynamic client registration extension~\cite{openid-connect-dynamic-client-registration-1.0}: In
Step~\refprotostep{oicacf-reg-req} the RP contacts the IdP and
provides its redirect URIs. In return, the IdP issues a client id and
(optionally) a client secret to the RP in
Step~\refprotostep{oicacf-reg-resp}. This concludes the registration.

Now, the core part of the OIDC protocol starts: the RP
redirects the user's browser to the IdP
in~\refprotostep{oicacf-start-resp}. This redirect contains the
information that the authorization code mode is used. Also, this
redirect contains the client id of the RP, a redirect URI, and a $\mi{state}$
value, which serves as a Cross-Site Request Forgery (CSRF) token when the browser is later
redirected back to the RP. The redirect may also optionally include a
nonce, which will be included in the id token issued later in this
flow. This data is sent to the
IdP by the browser~\refprotostep{oicacf-idp-auth-req-1}. The user
authenticates to the IdP~\refprotostep{oicacf-idp-auth-resp-1},
\refprotostep{oicacf-idp-auth-req-2}, and the IdP redirects the user's
browser back to the RP in~\refprotostep{oicacf-idp-auth-resp-2}
and~\refprotostep{oicacf-redir-ep-req} (using the redirect URI from
the request in~\refprotostep{oicacf-idp-auth-req-1}). This redirect
contains an authorization code, the $\mi{state}$ value as received
in~\refprotostep{oicacf-start-resp}, and the issuer
identifier.\footnote{The issuer identifier will be included
  in this message in an upcoming revision of OIDC to mitigate the IdP Mix-Up attack, see
  Section~\ref{sec:idp-mix-up}.} If the $\mi{state}$ value and the issuer
identifier are correct, the RP contacts the IdP
in~\refprotostep{oicacf-token-req} at the token endpoint with the received authorization
code, its client id, its client secret (if any), and the redirect URI used to obtain the authorization code. If these values are correct, the IdP
responds with a fresh access token and an id token to the RP
in~\refprotostep{oicacf-token-resp}. %
If the id token is valid, then the RP considers the user to be logged
in (under the identifier composed from the user id in the id token and
the issuer identifier). Hence, the RP may set a session cookie at the
user's browser in~\refprotostep{oicacf-redir-ep-resp}.

\subsubsection{Implicit Mode}\label{sec:oidc-modes-implicit}
The implicit mode (depicted in Figure~\ref{fig:oidc-implicit-flow} in Appendix~\ref{app:idp-mix-up}) is similar to the authorization code mode, but instead of
providing an authorization code, the IdP issues an id token right away
to the RP (via the user's browser) when the user authenticates to the IdP. Hence, the
Steps~\refprotostep{oicacf-start-req}--\refprotostep{oicacf-idp-auth-req-2}
of the authorization code mode (Figure~\ref{fig:oidc-auth-code-flow})
are the same. After these steps, the IdP redirects the
user's browser to the redirection endpoint at the RP, providing an id
token, (optionally) an access token, the $\mi{state}$ value,
and the issuer identifier.
These values are not provided
as a URL parameter but in the URL fragment instead. Hence, the browser
does not send them to the RP at first. Instead, the RP has to provide a
JavaScript that retrieves these values from the fragment and sends them
to the RP. If the id token is valid, the issuer is correct, and the
$\mi{state}$ matches the one previously chosen by the RP, the RP considers the user to be
logged in and issues a session cookie.

\subsubsection{Hybrid Mode}\label{sec:oidc-modes-hybrid}
The hybrid mode (depicted in Figure~\ref{fig:oidc-hybrid-flow} in
Appendix~\ref{app:idp-mix-up}) is a combination of the authorization
code mode and the implicit mode: First, this mode works like the
implicit mode, but when IdP redirects the browser back to RP, the IdP issues an authorization
code, and either an id token or an access token or both.\footnote{The
  choice of the IdP to issue either an id token or an access token or
  both depends on the IdP's configuration and the request in
  Step~\refprotostep{oichf-idp-auth-req-1} in
  Figure~\ref{fig:oidc-hybrid-flow}.} The RP then retrieves these
values as in the implicit mode (as they are sent in the fragment like
in the implicit mode) and uses the authorization code to obtain a
(potentially second) id token and a (potentially second) access token
from IdP.

\subsection{Relationship to OAuth 2.0}\label{sec:oidc-rel-to-oauth}

Technically, OIDC is derived from OAuth~2.0. It goes, however, far
beyond what was specified in OAuth~2.0 and introduces many new
concepts: OIDC defines a method for authentication (while retaining
the option for authorization) using a new type of tokens, the id
token. Some messages and tokens in OIDC can be cryptographically
signed or encrypted while OAuth~2.0 does neither use signing nor
encryption. The new hybrid flow combines features of the implicit mode
and the authorization code mode. Importantly, with ad-hoc discovery and
dynamic registration, OIDC standardizes and automates a process that
is completely out of the scope of OAuth~2.0.
  
These new features and their interplay potentially introduce new
security flaws. It is therefore not sufficient to analyze the security
of OAuth~2.0 to derive any guarantees for OIDC. OIDC rather requires a
new security analysis. (See Section~\ref{sec:analysis} for a more
detailed discussion. In Section~\ref{sec:attacks-impl-guid} we
describe attacks that cannot be applied to OAuth~2.0.)

\section{Attacks and Security Guidelines}
\label{sec:attacks-impl-guid}

In this section, we present a concise overview of known attacks on
OIDC and present additions that have not been documented so far. We
also summarize mitigations and implementation guidelines that have to
be implemented to avoid these attacks.

The main focus of this work is to prove central security properties of
OIDC, by which these mitigations and implementation guidelines are
backed up. Moreover, further (potentially unknown types of) attacks on
OIDC that can be captured by our security analysis are ruled out as
well.

The rest of the section is structured as follows: we first present the
attacks, mitigations and guidelines, then point out differences to
OAuth~2.0, and finally conclude with a brief discussion.

\subsection{Attacks, Mitigations, and Guidelines}
\label{sec:att-mit-guide}

(Mitigations and guidelines are presented along with every class of attack.)

\subsubsection{IdP Mix-Up Attacks}
\label{sec:idp-mix-up}
In two previously reported
attacks~\cite{FettKuestersSchmitz-CCS-2016,MainkaMladenovSchwenkWich-EuroSP-2017},
the aim was to confuse the RP about the identity of the IdP. In both
attacks, the user was tricked into using an \emph{honest} IdP to
authenticate to an \emph{honest} RP, while the RP is made to believe
that the user authenticated to the attacker. The RP therefore, after
successful user authentication, tries to use the authorization code or
access token at the attacker, which then can impersonate the user or
access the user's data at the IdP. We present a detailed description
of an application of the IdP Mix-Up attack to OpenID Connect in
Appendix~\ref{app:idp-mix-up}.

The IETF OAuth Working Group drafted a proposal for a mitigation
technique \cite{rfc-draft-ietf-oauth-mix-up-mitigation-01} that is
based on a proposal in \cite{FettKuestersSchmitz-CCS-2016} and that
also applies to OpenID Connect. The proposal is that the IdP puts its
identity into the response from the authorization endpoint. (This is
already included in our description of OIDC above, see the issuer in
Step~\refprotostep{oicacf-idp-auth-resp-2} in
Figure~\ref{fig:oidc-auth-code-flow}.) The RP can then check that the
user authenticated to the expected IdP.

\subsubsection{Attacks on the State Parameter}
\label{sec:state-leak-attack}
The $\mi{state}$
parameter is used in OIDC to protect against attacks on session
integrity, i.e., attacks in which an attacker forces a user to be
logged in at some RP (under the attacker's account). Such attacks can
arise from session swapping or CSRF vulnerabilities.

OIDC recommends the use of the $\mi{state}$
parameter. It should contain a nonce that is bound to the user's
session. Attacks that can result from omitting or incorrectly using
$\mi{state}$
were described in the context of OAuth~2.0
in~\cite{rfc6819-oauth2-security,BansalBhargavanetal-JCS-2014,SantsaiBeznosov-CCS-2012-OAuth,LiMitchell-ISC-2014}.

The nonce for the $\mi{state}$
value should be chosen freshly for \emph{each login attempt} to
prevent an attack described in \cite{FettKuestersSchmitz-CCS-2016}
(Section~5.1) where the same state value is used first in a
user-initiated login flow with a malicious IdP and then in a login
flow with an honest IdP (forcefully initiated by the attacker with the
attacker's account and the user's browser).

\subsubsection{Code/Token/State Leakage}\label{sec:codet-leak}
Care should be taken that a value of $\mi{state}$
or an authorization code is not inadvertently sent to an untrusted
third party through the \emph{Referer} header. The $\mi{state}$
and the authorization code parameters are part of the redirection
endpoint URI (at the RP), the $\mi{state}$
parameter is also part of the authorization endpoint URI (at the IdP).
If, on either of these two pages, a user clicks on a link to an
external page, or if one of these pages embeds external resources
(images, scripts, etc.), then the third party will receive the full
URI of the endpoint, including these parameters, in the Referer header
that is automatically sent by the browser.

Documents delivered at the respective endpoints should therefore be
vetted carefully for links to external pages and resources. In modern
browsers, \emph{referrer policies} \cite{w3c-draft-referrer-policy}
can be used to suppress the Referer header. As a second line of
defense, both parameters should be made single-use, i.e., $\mi{state}$
should expire after it has been used at the redirection endpoint and
authorization code after it has been redeemed at IdP.

In a related attack, an attacker that has access to logfiles
or browsing histories (e.g., through malicious browser extensions) can
steal authentication codes, access tokens, id tokens, or $\mi{state}$
values and re-use these to impersonate a user or to break session
integrity. A subset of these attacks was dubbed \emph{Cut-and-Paste
  Attacks} by the IETF OAuth working group
\cite{rfc-draft-ietf-oauth-mix-up-mitigation-01}.

There are drafts for RFCs that tackle specific aspects of these
leakage attacks, e.g., \cite{rfc-draft-oauth-jwt-encoded-state} which
discusses binding the $\mi{state}$
parameter to the browser instance, and
\cite{rfc-draft-oauth-token-binding-01} which discusses binding the
access token to a TLS session. Since these mitigations are still very
early IETF drafts, subject to change, and not easy to implement in the
majority of the existing OIDC implementations, we did not model them.

In our analysis, we assume that implementations keep logfiles and
browsing histories (of honest browsers) secret and employ referrer
policies as described above.

\subsubsection{Na\"ive RP Session Integrity Attack}
\label{sec:naive-rp-session}
So far, we have assumed that after
Step~\refprotostep{oicacf-start-resp}
(Figure~\ref{fig:oidc-auth-code-flow}), the RP remembers the user's
choice (which IdP is used) in a session; more precisely, the user's
choice is stored in RP's session data. This way, in
Step~\refprotostep{oicacf-redir-ep-req}, the RP looks up the user's
selected IdP in the session data. In
\cite{FettKuestersSchmitz-CCS-2016}, this is called \emph{explicit
  user intention tracking}.

There is, however, an alternative to storing the IdP in the session.
As pointed out by \cite{FettKuestersSchmitz-CCS-2016}, some
implementations put the identity of the IdP into the
$\mi{redirect\_uri}$
(cf.~Step~\refprotostep{oicacf-start-resp}), e.g., by appending it as
the last part of the path or in a parameter. Then, in
Step~\refprotostep{oicacf-redir-ep-req}, the RP can retrieve this
information from the URI. This is called \emph{na\"ive user intention
  tracking}.

RPs that use na\"ive user intention tracking are susceptible to the
na\"ive RP session integrity attack described
in~\cite{FettKuestersSchmitz-CCS-2016}: An attacker obtains an
authorization code, id token, or access token for his own account at
an honest IdP (HIdP). He then waits for a user that wants to log in at
some RP using the attacker's IdP (AIdP) such that AIdP obtains a valid
$\mi{state}$
for this RP. AIdP then redirects the user to the redirection endpoint
URI of RP using the identity of HIdP plus the obtained $\mi{state}$
value and code or (id) token. Since the RP cannot see that the user
originally wanted to log in using AIdP instead of HIdP, the user will
now be logged in under the attacker's identity.

Therefore, an RP should always use sessions to store the user's chosen
IdP (explicit user intention tracking), which, as mentioned, is also
what we do in our formal OIDC model.

\subsubsection{307 Redirect Attack}
\label{sec:307-redirect-attack}
Although OIDC explicitly allows for any redirection method to be used
for the redirection in Step~\refprotostep{oicacf-idp-auth-resp-2} of
Figure~\ref{fig:oidc-auth-code-flow}, IdPs should not use an HTTP 307
status code for redirection. Otherwise, credentials entered by the
user at an IdP will be repeated by the browser in the request to RP
(Step~\refprotostep{oicacf-redir-ep-req} of
Figure~\ref{fig:oidc-auth-code-flow}), and hence, malicious RPs would
learn these credentials and could impersonate the user at the
IdP. This attack was presented in \cite{FettKuestersSchmitz-CCS-2016}.
In our model, we exclusively use the 303 status code, which prevents
re-sending of form data.

\subsubsection{Injection Attacks}
\label{sec:injection-attacks}
It is well known that Cross-Site Scripting (XSS) and SQL Injection
attacks on RPs or IdPs can lead to theft of access tokens, id tokens,
and authorization codes (see, for example,
\cite{rfc6749-oauth2,rfc6819-oauth2-security,BansalBhargavanetal-JCS-2014,SantsaiBeznosov-CCS-2012-OAuth,MainkaMladenovSchwenkWich-EuroSP-2017}).
XSS attacks can, for example, give an attacker access to session ids.
Besides using proper escaping (and Content Security Policies
\cite{w3c-content-security-policy} as a second line of defense), OIDC
endpoints should therefore be put on domains separate from other,
potentially more vulnerable, web pages on IdPs and
RPs.\footnote{\label{fn:origin-separation}Since scripts on one origin
  can often access documents on the same origin, origins of
  the OIDC endpoints should be free from untrusted scripts.}
(See \emph{Third-Party Resources} below for
another motivation for this separation.)

In OIDC implementations, data that can come from untrusted sources
(e.g., client ids, user attributes, $\mi{state}$
and $\mi{nonce}$
values, redirection URIs) must be treated as such: For example, a
malicious IdP might try to inject user attributes containing malicious
JavaScript to the RP. If the RP displays this data without applying
proper escaping, the JavaScript is executed.

We emphasize that in a similar manner, attackers can try to inject
additional parameters into URIs by appending them to existing
parameter values, e.g., the $\mi{state}$.
Since data is often passed around in OIDC, proper escaping of such
parameters can be overlooked easily.

As a result of such parameter injection attacks or independently, parameter
pollution attacks can be a threat for OIDC implementations. In these
attacks, an attacker introduces duplicate parameters into URLs
(see, e.g., \cite{Balduzzietal-NDSS-2011}). For example, a simple
parameter pollution attack could be launched as follows: A malicious
RP could redirect a user to an honest IdP, using a client id of some
honest RP but appending two redirection URI parameters, one pointing
to the honest RP and one pointing to the attacker's RP. Now, if the
IdP checks the first redirection URI parameter, but afterwards
redirects to the URI in the second parameter, the attacker learns
authentication data that belongs to the honest RP and can impersonate
the user.

Mitigations against all these kinds of injection
attacks are well known: implementations have to vet incoming data
carefully, and properly escape any output data. In our model, we assume that these mitigations are
implemented.

\subsubsection{CSRF Attacks and Third-Party Login Initiation}
\label{sec:csrf-attacks-tpli}
Some endpoints need protection against CSRF in addition to the
protection that the $\mi{state}$
parameter provides, e.g., by checking the origin header. Our analysis
shows that the RP only needs to protect the URI on which the login
flow is started (otherwise, an attacker could start a login flow using
his own identity in a user's browser) and for the IdP to protect the
URI where the user submits her credentials (otherwise, an attacker
could submit his credentials instead). 

In the OIDC Core standard \cite{openid-connect-core-1.0}, a so-called
\emph{login initiation endpoint} is described which allows a third
party to start a login flow by redirecting a user to this endpoint,
passing the identity of an IdP in the request. The RP will then start
a login flow at the given IdP. Members of the OIDC foundation
confirmed to us that this endpoint is essentially an intentional CSRF
protection bypass. We therefore recommend login initiation
endpoints not to be implemented (they are not a mandatory feature),
or to require explicit confirmation by the user.

\subsubsection{Server-Side Request Forgery (SSRF)}\label{sec:ssrf}
SSRF attacks can arise when an attacker can instruct a server to
send HTTP(S) requests to other hosts, causing unwanted side-effects or
revealing information~\cite{Pellegrinoetal-RAID-2016}. For example, if
an attacker can instruct a server behind a firewall to send
requests to other hosts behind this firewall, the attacker might be
able to call services or to scan the internal network (using timing
attacks). He might also instruct the server to retrieve very large
documents from other sources, thereby creating Denial of Service
attacks.

SSRF attacks on OIDC were described for the first time in
\cite{MainkaMladenovSchwenkWich-EuroSP-2017},
in the context of the OIDC Discovery extension: An attacker could set
up a malicious discovery service that, when queried by an RP, answers
with links to arbitrary, network-internal or external servers (in
Step~\refprotostep{oicacf-conf-resp} of
Figure~\ref{fig:oidc-auth-code-flow}).

We here, for the first time, point out that not
only RPs can be vulnerable to SSRF, but also IdPs. OIDC defines a way to
indirectly pass parameters for the authorization~request (cf.
Step~\refprotostep{oicacf-idp-auth-req-1} in
Figure~\ref{fig:oidc-auth-code-flow}). To this end, the IdP accepts a new
parameter, $\mi{request\_uri}$
in the authorization request. This parameter contains a URI from which
the IdP retrieves the additional parameters (e.g., $\mi{redirect\_uri}$).
The attacker can use this feature to easily mount an SSRF attack
against the IdP even without any OIDC extensions: He can put an
arbitrary URI in an authorization request causing the IdP to contact this
URI.

This new attack vector shows that not only RPs but also IdPs have to
protect themselves against SSRF by using appropriate filtering and
limiting mechanisms to restrict unwanted requests that originate from
a web server (cf.~\cite{Pellegrinoetal-RAID-2016}).

SSRF attacks typically depend on an application specific context, such
as the structure of and (vulnerable) services in an internal
network.
In our model, attackers can trigger SSRF requests, but the model does
not contain vulnerable applications/services aside from OIDC. (Our analysis focuses on the security of the OIDC standard itself,
rather than on specific applications and services.) Timing and
performance properties, while sometimes relevant for SSRF attacks, are
also outside of our analysis.

\subsubsection{Third-Party Resources}
\label{sec:third-party-reso}
RPs and IdPs that include third-party resources, e.g., tracking or
advertisement scripts, subject their users to token theft and other
attacks by these third parties. If possible, RPs and IdPs should
therefore avoid including third-party resources on any web resources
delivered from the same origins\textsuperscript{\ref{fn:origin-separation}} as the OIDC endpoints (see also Section~\ref{sec:discussion}). For newer browsers,
\emph{subresource integrity}~\cite{AkhaweBraunMarierWeinberger-w3c-subresource-integrity} can help to reduce the risks associated
with embedding third-party resources. With subresource integrity,
websites can instruct supporting web browsers to reject third-party
content if this content does not match a specific hash.
In our model, we assume that websites do not include untrusted third-party resources.

\subsubsection{Transport Layer Security}
\label{sec:transp-layer-secur}
The security of OIDC depends on the confidentiality and integrity of
the transport layer. In other words, RPs and IdPs should use HTTPS.
Endpoint URIs that are provided for the end user and that are
communicated, e.g., in the discovery phase of the protocol, should
only use the \texttt{https://} scheme. HTTPS Strict Transport Security
and Public Key Pinning can be used to further strengthen the security
of the OIDC endpoints. (In our model, we assume that users enter their
passwords only over HTTPS web sites because otherwise, any
authentication could be broken trivially.) %

\subsubsection{Session Handling}
\label{sec:session-handling}
Sessions are typically identified by a nonce that is stored in the user's browser as a cookie. It is a well known best practice that cookies should make use of the
\emph{secure} attribute (i.e., the cookie is only ever used over HTTPS
connections) and the \emph{HttpOnly} flag (i.e., the cookie is not
accessible by JavaScript). Additionally, after the login, the
RP should replace the session id of the user by a freshly chosen nonce
in order to prevent session fixation attacks: Otherwise, a network
attacker could set a login session cookie that is bound to a known
$\mi{state}$ value into the user's browser (see~\cite{Zhengetal-cookies-usenix-2015}), lure the user into logging in
at the corresponding RP, and then use the session cookie to access the
user's data at the RP (\emph{session fixation}, see
\cite{owasp-session-fixation}). In our model, RPs use two kinds of
sessions: Login sessions (which are valid until just before a user is
authenticated at the RP) and service sessions (which signify that a
user is already signed in to the RP). For both sessions, the secure
and HttpOnly flags are used.

\subsection{Relationship to OAuth~2.0}
\label{sec:discussion-att-mit-guide-oauth}

Many, but not all of the attacks described above can also be applied
to OAuth~2.0. The following attacks in particular are only applicable
to OIDC: (1)~Server-side request forgery attacks are facilitated by
the ad-hoc discovery and dynamic registration features. (2)~The same
features enable new ways to carry out injection attacks. (3)~The new
OIDC feature third-party login initiation enables new CSRF attacks.
(4)~Attacks on the id token only apply to OIDC, since there is no such
token in OAuth~2.0.

It is interesting to note that on the other hand, some attacks on
OAuth~2.0 cannot be applied to OIDC
(see~\cite{rfc6819-oauth2-security,FettKuestersSchmitz-CCS-2016} for
further discussions on these attacks): (1)~OIDC setups are less prone
to open redirector attacks since placeholders are not allowed in
redirection URIs. (2)~TLS is mandatory for some messages in OIDC,
while it is optional in OAuth~2.0. (3)~The nonce value can prevent
some replay attacks when the state value is not used or leaks to an
attacker.

\subsection{Discussion}
\label{sec:discussion-att-mit-guide}

In this section, our focus was to provide a concise overview of known
attacks on OIDC and present some additions, namely SSRF at IdPs and
third-party login initiation, along with mitigations and
implementation guidelines. Our formal analysis of OIDC, which is the
main focus of our work and is presented in the next sections, shows
that the mitigations and implementation guidelines presented above are
effective and that we can exclude other, potentially unknown types of
attacks.

\section{The FKS Web Model}
\label{sec:fks-web-model}

Our formal security analysis of OIDC is based on the FKS model, a
generic Dolev-Yao style web model proposed by Fett et al.
in~\cite{FettKuestersSchmitz-SP-2014}. Here, we only briefly recall
this model following the description in
\cite{FettKuestersSchmitz-CCS-2016} (see Appendices~\ref{app:web-model}\,ff. for a full description, and 
\cite{FettKuestersSchmitz-SP-2014,FettKuestersSchmitz-CCS-2015,FettKuestersSchmitz-ESORICS-BrowserID-Primary-2015}
for comparison with other models and discussion
of its scope and limitations).

The FKS model is designed independently of a specific web application and
closely mimics published (de-facto) standards and specifications for
the web, for example, the HTTP/1.1 and HTML5 standards and associated
(proposed) standards. The FKS model defines a general communication
model, and, based on it, web systems consisting of web browsers, DNS
servers, and web servers as well as web and network attackers. 

\paragraph{Communication Model}
The main entities in the model are \emph{(atomic) processes}, which
are used to model browsers, servers, and attackers. Each process
listens to one or more (IP) addresses.
Processes communicate via \emph{events}, which consist of a message as
well as a receiver and a sender address. In every step of a run, one
event is chosen non-deterministically from a ``pool'' of waiting
events and is delivered to one of the processes that listens to the
event's receiver address. The process can then handle the event and
output new events, which are added to the pool of events, and so on.

As usual in Dolev-Yao models (see, e.g.,
\cite{AbadiFournet-POPL-2001}), messages are expressed as formal terms
over a signature $\Sigma$.
The signature contains constants (for (IP) addresses, strings, nonces)
as well as sequence, projection, and function symbols (e.g., for
encryption/decryption and signatures). For example, in the web model,
an HTTP request is represented as a term $r$
containing a nonce, an HTTP method, a domain name, a path, URI
parameters, request headers, and a message body. For instance, an HTTP
request for the URI \url{http://ex.com/show?p=1} is represented as
$\mi{r} := \langle \cHttpReq, n_1, \mGet, \str{ex.com}, \str{/show},
\an{\an{\str{p},1}}, \an{}, \an{} \rangle$ where the body and the list
of request headers is empty. An HTTPS request for $r$
is of the form $\ehreqWithVariable{r}{k'}{\pub(k_\text{ex.com})}$,
where $k'$
is a fresh symmetric key (a nonce) generated by the sender of the
request (typically a browser); the responder is supposed to use this
key to encrypt the response.

The \emph{equational theory} associated
with $\Sigma$
is defined as usual in Dolev-Yao models. The theory induces a congruence
relation $\equiv$
on terms, capturing the meaning of the function symbols in $\Sigma$.
For instance, the equation in the equational theory which captures
asymmetric decryption is $\dec{\enc x{\pub(y)}}{y}=x$.
With this, we have that, for example, $\dec{\ehreqWithVariable{r}{k'}{\pub(k_\text{ex.com})}}{k_\text{ex.com}}\equiv
  \an{r,k'}\,,$
i.e., these two terms are equivalent w.r.t.~the equational theory.

A \emph{(Dolev-Yao) process} consists of a set of addresses the
process listens to, a set of states (terms), an initial state, and a
relation that takes an event and a state as input and
(non-deterministically) returns a new state and a sequence of events.
The relation models a computation step of the
process.
It is required that the output can be computed (formally, derived
in the usual Dolev-Yao style) from the input event and the state.

The so-called \emph{attacker process} is a Dolev-Yao process which
records all messages it receives and outputs all events it can
possibly derive from its recorded messages. Hence, an attacker process
carries out all attacks any Dolev-Yao process could possibly perform.
Attackers can corrupt other parties.

A \emph{script} models JavaScript running in a browser. Scripts are
defined similarly to Dolev-Yao processes. When triggered by a browser, a
script is provided with state information. The script then outputs a
term representing a new internal state and a command to be interpreted
by the browser (see also the specification of browsers below). We give
an annotated example for a script in Algorithm~\ref{alg:script-rp-index} in the
appendix. Similarly to an attacker process, the so-called
\emph{attacker script} outputs everything that is derivable from
the input.

A \emph{system} is a set of processes. A \emph{configuration} of this
system consists of the states of all processes in the system, the pool
of waiting events, and a sequence of unused nonces. Systems induce
\emph{runs}, i.e., sequences of configurations, where each
configuration is obtained by delivering one of the waiting events of
the preceding configuration to a process, which then performs a
computation step. The transition from one configuration to the next
configuration in a run is called a \emph{processing~step}. We write,
for example, $Q = (S, E, N) \xrightarrow[]{} (S', E', N')$
to denote the transition from the configuration $(S, E, N)$
to the configuration $(S', E', N')$,
where $S$
and $S'$
are the states of the processes in the system, $E$
and $E'$
are~pools of waiting events, and $N$
and $N'$ are sequences of unused~nonces.

A \emph{web system} formalizes the web infrastructure and web
applications. It contains a system consisting of honest and attacker
processes. Honest processes can be web browsers, web servers, or DNS
servers. Attackers can be either \emph{web attackers} (who can listen
to and send messages from their own addresses only) or \emph{network
  attackers} (who may listen to and spoof all addresses and therefore
are the most powerful attackers). A web system further contains a set
of scripts (comprising honest scripts and the attacker script).

In our analysis of OIDC, we consider either one network attacker or a
set of web attackers (see Section~\ref{sec:analysis}). In our OIDC
model, we need to specify only the behavior of servers and scripts.
These are not defined by the FKS model since they depend on the
specific application, unless they become corrupted, in which case they
behave like attacker processes and attacker scripts; browsers are
specified by the FKS model (see below). The modeling of OIDC servers
and scripts is outlined in Section~\ref{sec:analysis}
and with full details provided in Appendices~\ref{app:model-oidc-auth} and~\ref{app:model-oidc-auth-webattackers}.

\paragraph{Web Browsers}
An honest browser is thought to be used by one honest user, who is
modeled as part of the browser. User actions, such as following a link, are
modeled as non-deterministic actions of the web browser. User
credentials are stored in the initial state of the browser and are
given to selected web pages when needed. Besides user credentials, the
state of a web browser contains (among others) a tree of windows and
documents, cookies, and web storage data (localStorage and
sessionStorage).

A \emph{window} inside a browser contains a set of
\emph{documents} (one being active at any time), modeling the
history of documents presented in this window. Each represents one
loaded web page and contains (among others) a script and a list of
subwindows (modeling iframes). The script, when triggered by the
browser, is provided with all data it has access to, such as a
(limited) view on other documents and windows, certain cookies, and
web storage data. Scripts then output a command and a new state. This
way, scripts can navigate or create windows, send XMLHttpRequests and
postMessages, submit forms, set/change cookies and web storage data,
and create iframes. Navigation and security rules ensure that scripts
can manipulate only specific aspects of the browser's state, according
to the relevant web standards.

A browser can output messages on the network of different types,
namely DNS and HTTP(S) (including XMLHttpRequests), and it processes
the responses. Several HTTP(S) headers are modeled, including, for
example, cookie, location, strict transport security (STS), and origin
headers. A browser, at any time, can also receive a so-called trigger
message upon which the browser non-de\-ter\-min\-is\-tically choses an
action, for instance, to trigger a script in some document. The script
now outputs a command, as described above, which is then further
processed by the browser. Browsers can also become corrupted, i.e., be
taken over by web and network attackers. Once corrupted, a browser
behaves like an attacker process.

\section{Analysis}
\label{sec:analysis}

We now present our security analysis of OIDC, including a formal model
of OIDC, the specifications of central security properties, and our
theorem which establishes the security of OIDC in our model. 

More precisely, our formal model of OIDC uses the FKS model as a
foundation and is derived by closely following the OIDC standards
Core, Discovery, and Dynamic Client
Registration~\cite{openid-connect-core-1.0,openid-connect-discovery-1.0,openid-connect-dynamic-client-registration-1.0}.
(As mentioned above, the goal in this work is to analyze OIDC itself
instead of concrete implementations.) We then formalize the main
security properties for OIDC, namely authentication, authorization,
session integrity for authentication, and session integrity for
authorization. We also formalize secondary security properties that
capture important aspects of the security of OIDC, for example,
regarding the outcome of the dynamic client registration. We then
state and prove our main theorem. Finally, we discuss the relationship
of our work to the analysis of OAuth~2.0 presented
in~\cite{FettKuestersSchmitz-CCS-2016} and conclude with a discussion
of the results.

We refer the reader to
Appendices~\ref{app:model-oidc-auth}--\ref{app:proof-oidc} for full
details, including definitions, specifications, and proofs. To
provide an intuition of the abstraction level, syntax, and concepts
that we use for the modeling without reading all details, we
extensively annotated Algorithms~\ref{alg:rp-oidc-http-request},
\ref{alg:rp-check-id-token}, and~\ref{alg:script-rp-index} in
Appendix~\ref{app:model-oidc-auth}.

\subsection{Model}
\label{sec:model}

Our model of OIDC includes all features that are commonly found in
real-world implementations, for example, all three modes, a detailed
model of the Discovery mechanism~\cite{openid-connect-discovery-1.0} (including the WebFinger
protocol~\cite{rfc7033-webfinger}), and Dynamic Client
Registration~\cite{openid-connect-dynamic-client-registration-1.0}
(including dynamic exchange of signing keys). RPs, IdPs, and, as usual
in the FKS model, browsers can be corrupted by the adversary
dynamically.

We do not model less used features, in particular OIDC logout, self-issued
OIDC providers (``personal, self-hosted OPs that issue self-signed ID Tokens'',~\cite{openid-connect-core-1.0}), and
ACR/AMR (Authentication Class/Methods Reference) values that can be
used to indicate the level of trust in the authentication of the user
to the IdP.

Since the FKS model has no notion of time,
we overapproximate by never letting tokens, e.g., id tokens, expire.
Moreover, we subsume user claims (information about the user that can
be retrieved from IdPs) by user identifiers, and hence, in our model
users have identities, but no other properties.

We have two versions of our OIDC model, one with a network attacker
and one with an unbounded number of web attackers, as explained next.
The reason for having two versions is that while the authentication
and authorization properties can be proven assuming a network
attacker, such an attacker could easily break session integrity.
Hence, for session integrity we need to assume web attackers (see the
explanations for session integrity in
Section~\ref{sec:security-properties}).

\subsubsection{OIDC Web System with a Network Attacker}
\label{sec:proc-oidcw}
We model OIDC as a class of web systems (in the sense of
Section~\ref{sec:fks-web-model}) which can contain an unbounded finite number
of RPs, IdPs, browsers, and one network attacker.

More formally, an \emph{OIDC web system with a network attacker}
($\oidcwebsystem^n$) consists of a network attacker, a finite set of
web browsers, a finite set of web servers for the RPs, and a finite
set of web servers for the IdPs. Recall that in $\oidcwebsystem^n$,
since we have a network attacker, we do not need to consider web
attackers (as the network attacker subsumes all web attackers). All
non-attacker parties are initially honest, but can become corrupted
dynamically upon receiving a special message and then behave just like
a web attacker process.

As already mentioned in Section~\ref{sec:fks-web-model},  to model
OIDC based on the FKS model, we have to specify the protocol specific
behavior only, i.e., the servers for RPs and IdPs as well as the scripts
that they use. We start with a description of the servers.

\paragraph{Web Servers} Since RPs and IdPs both are web servers, we
developed a generic model for HTTPS server processes for the FKS
model. We call these processes \emph{HTTPS server base processes}.
Their definition covers decrypting received HTTPS messages and
handling HTTP(S) requests to external webservers, including DNS
resolution.

RPs and IdPs are derived from this HTTPS server base process. Their
models follow the OIDC standard closely and include the mitigations
discussed in Section~\ref{sec:attacks-impl-guid}.

An RP waits for users to start a login flow and then
non-deterministically decides which mode to use. If needed, it starts
the discovery and dynamic registration phase of the protocol, and
finally redirects the user to the IdP for user authentication.
Afterwards, it processes the received tokens and uses them according
to their type (e.g., with an access token, the RP would retrieve an id
token from the IdP). If an id token is received that passes all
checks, the user will be logged in. As mentioned briefly in
Section~\ref{sec:attacks-impl-guid}, RPs manage two kinds of sessions:
The \emph{login sessions}, which are used only during the user login
phase, and \emph{service sessions}.

The IdP provides several endpoints according to its role in the login
process, including its OIDC configuration endpoint and endpoints for
receiving authentication and token requests.

\paragraph{Scripts} Three scripts (altogether 30 lines of code) can be
sent from honest IdPs and RPs to web browsers. The script
\emph{script\_rp\_index} is sent by an RP when the user visits the
RP's web site. It starts the login process. The script
\emph{script\_rp\_get\_fragment} is sent by an RP during an implicit
or hybrid mode flow to retrieve the data from the URI fragment. It
extracts the access token, authorization code, and $\mi{state}$
from the fragment part of its own URI and sends this information in
the body of a POST request back to the RP. IdP sends the script
\emph{script\_idp\_form} for user authentication at the IdP.

\subsubsection{OIDC Web System with Web Attackers}
We also consider a class of web systems where the network attacker is
replaced by an unbounded finite set of web attackers and a DNS server
is introduced.
We denote
such a system by $\oidcwebsystem^w$
and call it an \emph{OIDC web system with web attackers}. Such web
systems are used to analyze session integrity, see below.

\subsection{Main Security Properties}
\label{sec:security-properties}

Our primary security properties capture authentication, authorization
and session integrity for authentication and authorization. We will
present these security properties in the following, with full details
in Appendix~\ref{app:form-secur-prop}.

\paragraph{Authentication Property} 
The most important property for OIDC is the authentication property.
In short, it captures that a network attacker (and therefore also web
attackers) should be unable to log in as an honest user at an honest
RP using an honest IdP. 

Before we define this property in more detail, recall that in our
modeling, an RP uses two kinds of sessions: login sessions, which are
only used for the login flow, and service sessions, which are used
after a user/browser was logged in (see
Section~\ref{sec:session-handling} for details). When a login session
has finished successfully (i.e., the RP received a valid id token),
the RP uses a fresh nonce as the service session id, stores this id in
the session data of the login session, and sends the service session
id as a cookie to the browser. In the same step, the RP also stores
the issuer, say $d$,
that was used in the login flow and the identity (email address) of
the user, say $\mi{id}$,
as a pair $\an{d, \mi{id}}$,
referred to as a global user identifier in
Section~\ref{sec:oidc-basic-concepts}.

Now, our authentication property defines that a network attacker
should be unable to get hold of a service session id by which the
attacker would be considered to be logged in at an honest RP under an
identity governed by an honest IdP for an honest user/browser.

In order to define the authentication property formally, we first need
to define the precise notion of a service session. In the following,
as introduced in Section~\ref{sec:fks-web-model}, $(S, E, N)$
denotes a configuration in the run $\rho$
with its components $S$,
a mapping from processes to states of these processes, $E$,
a set of events in the network that are waiting to be delivered to
some party, and $N$,
a set of nonces that have not been used yet. By $\mathsf{governor}(\mi{id})$
we denote the IdP that is responsible for a given user identity (email
address) $\mi{id}$,
and by $\mathsf{dom}(\mathsf{governor}(\mi{id}))$,
we denote the set of domains that are owned by this IdP. By
$S(r).\str{sessions}[\mi{lsid}]$
we denote a data structure in the state of $r$
that contains information about the login session identified by $\mi{lsid}$.
This data structure contains, for example, the identity for which the
login session with the id $\mi{lsid}$
was started and the service session id that was issued after the login
session. 

We can now define that there is a service session identified by a
nonce $n$
for an identity $\mi{id}$
at some RP $r$
iff there exists a login session (identified by some nonce
$\mi{lsid}$)
such that $n$
is the service session associated with this login session, and $r$
has stored that the service session is logged in for the id $\mi{id}$
using an issuer $d$
(which is some domain of the governor of $\mi{id}$).

\begin{definition}[Service Sessions]
  We say that there is a \emph{service session identified by a nonce
    $n$
    for an identity $\mi{id}$
    at some RP $r$}
  in a configuration $(S, E, N)$
  of a run $\rho$
  of an OIDC web system iff there exists some login session id
  $\mi{lsid}$
  and a domain $d \in \mathsf{dom}(\mathsf{governor}(\mi{id}))$
  such that
  $S(r).\str{sessions}[\mi{lsid}][\str{loggedInAs}] \equiv \an{d,
    \mi{id}}$ and
  $S(r).\str{sessions}[\mi{lsid}][\str{serviceSessionId}] \equiv n$.
\end{definition}

By $d_{\emptyset}(S(\fAP{attacker}))$
we denote all terms that can be computed (derived in the usual
Dolev-Yao style, see Section~\ref{sec:fks-web-model}) from the
attacker's knowledge in the state $S$.
We can now define that an OIDC web system with a network attacker is
secure w.r.t.~authentication iff the attacker can never get hold of a
service session id ($n$)
that was issued by an honest RP $r$
for an identity $\mi{id}$
of an honest user (browser) at some honest IdP (governor of
$\mi{id}$).

\begin{definition}[Authentication Property] Let $\oidcwebsystem^n$ be an OIDC web
  system with a network attacker. We say that \emph{$\oidcwebsystem^n$
    is secure w.r.t.~authentication} iff for every run $\rho$
  of $\oidcwebsystem^n$,
  every configuration $(S, E, N)$
  in $\rho$,
  every $r\in \fAP{RP}$
  that is honest in $S$,
  every browser $b$
  that is honest in $S$,
  every identity $\mi{id} \in \mathsf{ID}$
  with $\mathsf{governor}(\mi{id})$
  being an honest IdP, every service session identified by some nonce
  $n$
  for $\mi{id}$
  at $r$, we have that 
  $n$
  is not derivable from the attackers knowledge in $S$
  (i.e., $n \not\in d_{\emptyset}(S(\fAP{attacker}))$).
\end{definition}

\paragraph{Authorization Property} Intuitively, authorization for OIDC
means that a network attacker should not be able to obtain or use a
protected resource available to some honest RP at an IdP for some user
unless certain parties involved in the authorization process are
corrupted. As the access control for such protected resources relies
only on access tokens, we require that an attacker does not learn
access tokens that would allow him to gain unauthorized access to
these resources.

To define the authorization property formally, we need to reason about
the state of an honest IdP, say $i$.
In this state, $i$
creates \emph{records}  containing data about successful
authentications of users at $i$.
Such records are stored in $S(i).\str{records}$. One such record, say $x$,
contains the authenticated user's identity in $x[\str{subject}]$,
two\footnote{In the hybrid mode, IdPs can issue two access tokens, cf.~Section~\ref{sec:oidc-modes-hybrid}.} access
tokens in $x[\str{access\_tokens}]$,
and the client id of the RP in $x[\str{client\_id}]$.

We can now define the authorization property. It defines that an OIDC
web system with a network attacker is secure w.r.t.~authorization iff
the attacker cannot get hold of an access token that is stored in one
of $i$'s
records for an identity of an honest user/browser $b$
and an honest RP $r$.

\begin{definition}[Authorization Property]
  Let $\oidcwebsystem^n$
  be an OIDC web system with a network attacker. We say that
  \emph{$\oidcwebsystem^n$
    is secure w.r.t.~authorization} iff for every run $\rho$
  of $\oidcwebsystem^n$,
  every configuration $(S, E, N)$
  in $\rho$,
  every $r\in \fAP{RP}$
  that is honest in $S$,
  every $i\in \fAP{IdP}$
  that is honest in $S$,
  every browser $b$
  that is honest in $S$,
  every identity $\mi{id} \in \mathsf{ID}$
  owned by $b$
  and $\mathsf{governor}(\mi{id}) = i$,
  every nonce $n$,
  every term $x \in S(i).\str{records}$
  with $x[\str{subject}] \equiv \mi{id}$,
  $n \in x[\str{access\_tokens}]$,
  and the client id $x[\str{client\_id}]$
  having been issued by $i$
  to $r$,\footnote{See
    Definition~\ref{def:client-id-issued} in
    Appendix~\ref{sec:fprop-authorization}.} we have that $n$
  is not derivable from the attackers knowledge in $S$
  (i.e., $n \not\in d_{\emptyset}(S(\fAP{attacker}))$).
\end{definition}

\paragraph{Session Integrity for Authentication}
The two session integrity properties capture that an attacker should
be unable to forcefully log a user/browser in at some RP. This includes attacks such as CSRF
and session swapping. Note that we define these properties over
$\oidcwebsystem^w$, i.e., we consider web attackers instead of a
network attacker. The reason is that OIDC deployments typically use
cookies to track the login sessions of users. Since a network attacker
can put cookies into browsers over unencrypted connections and these
cookies are then also used for encrypted connections, cookies have no
integrity in the presence of a network attacker (see also
\cite{Zhengetal-cookies-usenix-2015}). In particular, a network
attacker could easily break the session integrity of typical OIDC
deployments.

For session integrity for authentication we say that a user/browser
that is logged in at some RP must have expressed her wish to be logged
in to that RP in the beginning of the login flow. Note that not even a
malicious IdP should be able to forcefully log in its users (more
precisely, its user's browsers) at an honest RP. If the IdP is honest,
then the user must additionally have authenticated herself at the IdP
with the same user account that RP uses for her identification. This
excludes, for example, cases where (1) the user is forcefully logged
in to an RP by an attacker that plays the role of an IdP, and (2)
where an attacker can force an honest user to be logged in at some RP
under a false identity issued by an honest IdP.

In our formal definition of session integrity for authentication
(below), $\mathsf{loggedIn}_\rho^Q(b, r, u, i, \mi{lsid})$
denotes that in the processing step $Q$ (see below),
the browser $b$
was authenticated (logged in) to an RP $r$
using the IdP $i$
and the identity $u$
in an RP login session with the session id $\mi{lsid}$.
(Here, the processing step $Q$
corresponds to Step~\refprotostep{oicacf-redir-ep-resp} in
Figure~\ref{fig:oidc-auth-code-flow}.) The user authentication in the
processing step $Q$
is characterized by the browser $b$
receiving the service session id cookie that results from the login
session $\mi{lsid}$.

By $\mathsf{started}_\rho^{Q'}(b, r, \mi{lsid})$
we denote that the browser $b$,
in the processing step $Q'$
triggered the script $\mi{script\_rp\_index}$
to start a login session which has the session id $\mi{lsid}$
at the RP $r$.
(Compare Section~\ref{sec:fks-web-model} on how browsers handle
scripts.) Here, $Q'$
corresponds to Step~\refprotostep{oicacf-start-req} in
Figure~\ref{fig:oidc-auth-code-flow}.

By $\mathsf{authenticated}_\rho^{Q''}(b, r, u, i, \mi{lsid})$
we denote that in the processing step $Q''$,
the user/browser $b$
authenticated to the IdP $i$.
In this case, authentication means that the user filled out the login
form (in $\mi{script\_idp\_form}$)
at the IdP $i$
and, by this, consented to be logged in at $r$
(as in~Step~\refprotostep{oicacf-idp-auth-req-2} in
Figure~\ref{fig:oidc-auth-code-flow}).

Using these notations, we can now define security w.r.t.~session
integrity for authentication of an OIDC web system with web attackers
in a straightforward way:

\begin{definition} [Session Integr. for
  Authentication]
  Let $\oidcwebsystem^w$
  be an OIDC web system with web attackers. We say that
  \emph{$\oidcwebsystem^w$
    is secure w.r.t.~session integrity for authentication} iff for
  every run $\rho$
  of $\oidcwebsystem^w$, every processing step $Q$ in $\rho$ with
  $Q = (S, E, N) \xrightarrow[]{}
  (S', E', N')$ (for some $S$,
  $S'$,
  $E$,
  $E'$,
  $N$,
  $N'$),
  every browser $b$
  that is honest in $S$,
  every $i\in \fAP{IdP}$,
  every identity $u$
  that is owned by $b$,
  every $r\in \fAP{RP}$
  that is honest in $S$,
  every nonce $\mi{lsid}$,
  with $\mathsf{loggedIn}_\rho^Q(b, r, u, i, \mi{lsid})$,
  we have that (1) there exists a processing step $Q'$
  in $\rho$
  (before $Q$)
  such that $\mathsf{started}_\rho^{Q'}(b, r, \mi{lsid})$,
  and (2) if $i$
  is honest in $S$,
  then there exists a processing step $Q''$
  in $\rho$
  (before $Q$)
  such that
  $\mathsf{authenticated}_\rho^{Q''}(b, r, u, i, \mi{lsid})$.
\end{definition}

\paragraph{Session Integrity for Authorization}
For session integrity for authorization we say that if an RP uses some
access token at some IdP in a session with a user, then that user
expressed her wish to authorize the RP to interact with \emph{some}
IdP. Note that one cannot guarantee that the IdP with which RP
interacts is the one the user authorized the RP to interact with. This
is because the IdP might be malicious. In this case, for example in
the discovery phase, the malicious IdP might just claim (in
Step~\refprotostep{oicacf-wf-resp} in
Figure~\ref{fig:oidc-auth-code-flow}) that some other IdP is
responsible for the authentication of the user. If, however, the IdP
the user is logged in with is honest, then it should be guaranteed that the
user authenticated to that IdP and that the IdP the RP interacts with
on behalf of the user is the one intended by the user.

For the formal definition, we use two additional predicates:
$\mathsf{usedAuthorization}_\rho^Q(b, r, i, \mi{lsid})$ means that the
RP $r$, in a login session (session id $\mi{lsid}$) with the browser
$b$ used some access token to access services at the IdP $i$.  By
$\mathsf{actsOnUsersBehalf}_\rho^Q(b, r, u, i, \mi{lsid})$ we denote
that the RP $r$ not only used \emph{some} access token, but used one
that is bound to the user's identity at the IdP $i$.

Again, starting from our informal definition above, we define security
w.r.t.~session integrity for authorization of an OIDC web system with
web attackers in a straightforward way (and similarly to session
integrity for authentication):

\begin{definition} [Session Integr. for
  Authorization]
  Let $\oidcwebsystem^w$
  be an OIDC web system with web attackers. We say that
  \emph{$\oidcwebsystem^w$
    is secure w.r.t.~session integrity for authentication} iff for
  every run $\rho$
  of $\oidcwebsystem^w$, every processing step $Q$ in $\rho$ with
  $Q = (S, E, N) \xrightarrow[]{}
  (S', E', N')$  (for some $S$, $S'$, $E$, $E'$, $N$, $N'$), every browser $b$
  that is honest in $S$,
  every $i\in \fAP{IdP}$,
  every identity $u$
  that is owned by $b$,
  every $r\in \fAP{RP}$
  that is honest in $S$,
  every nonce $\mi{lsid}$,
  we have that (1) if
  $\mathsf{usedAuthorization}_\rho^Q(b, r, i, \mi{lsid})$,
  then there exists a processing step $Q'$
  in $\rho$
  (before $Q$)
  such that $\mathsf{started}_\rho^{Q'}(b, r, \mi{lsid})$,
  and (2) if $i$
  is honest in $S$
  and $\mathsf{actsOnUsersBehalf}_\rho^Q(b, r, u, i, \mi{lsid})$,
  then there exists a processing step $Q''$
  in $\rho$
  (before $Q$)
  such that
  $\mathsf{authenticated}_\rho^{Q''}(b, r, u, i, \mi{lsid})$.
\end{definition}

\subsection{Secondary Security Properties}
\label{sec:second-secur-prop}

We define the following secondary security properties that capture
specific aspects of OIDC. We use these secondary security properties
during our proof of the above main security properties. Nonetheless,
these secondary security properties are important and interesting in
their own right.

We define and prove the following properties (see the corresponding
lemmas in Appendices~\ref{sec:proof-property-a}
and~\ref{sec:proof-session-integrity-all} for details):

\noindent\textbf{Integrity of Issuer Cache:} 
If a relying party requests the issuer identifier from an identity
provider
(cf.~Steps~\refprotostep{oicacf-wf-req}--\refprotostep{oicacf-wf-resp}
in Figure~\ref{fig:oidc-auth-code-flow}), then the RP will only
receive an origin that belongs to this IdP in the response. In other
words, honest IdPs do not use attacker-controlled domains as issuer
identifiers, and the attacker is unable to alter this information on
the way to the RP or in the \emph{issuer cache} at the RP.

\noindent\textbf{Integrity of OIDC Configuration Cache:} 
(1) Honest IdPs only use endpoints under their control in their OIDC
configuration document
(cf.~Steps~\refprotostep{oicacf-conf-req}--\refprotostep{oicacf-conf-resp}
in Figure~\ref{fig:oidc-auth-code-flow}) and (2) this information
(which is stored at the RP in the so-called \emph{OIDC configuration
  cache}) cannot be altered by an attacker.

\noindent\textbf{Integrity of JWKS Cache:} 
RPs receive only ``correct'' signing keys from honest IdPs, i.e., keys
that belong to the respective IdP
(cf.~Steps~\refprotostep{oicacf-jwks-req}--\refprotostep{oicacf-jwks-resp}
in Figure~\ref{fig:oidc-auth-code-flow}).

\noindent\textbf{Integrity of Client Registration:}
Honest RPs register only redirection URIs that point to themselves and
that these URIs always use HTTPS. Recall that when an RP registers at
an IdP, the IdP issues a freshly chosen client id to the RP and then
stores RP's redirection URIs.

\noindent\textbf{Third Parties Do Not Learn Passwords:}
Attackers cannot learn user passwords. More precisely, we define
that $\mathsf{secretOfID}(\mi{id})$,
which denotes the password for a given identity $\mi{id}$,
is not known to any party except for the browser $b$
owning the id and the identity provider $i$
governing the id (as long as $b$
and $i$
are honest).

\noindent\textbf{Attacker Does Not Learn ID Tokens:}
Attackers cannot learn id tokens that were issued by honest IdPs for
honest RPs and identities of honest browsers.

\noindent\textbf{Third Parties Do Not Learn State:}
If an honest browser logs in at an honest RP using an honest IdP, then
the attacker cannot learn the $\mi{state}$
value used in this login flow.

\subsection{Theorem}
\label{sec:theorem}

The following theorem states that OIDC is secure w.r.t.~authentication
and authorization in presence of the network attacker, and that OIDC is
secure w.r.t.~session integrity for authentication and authorization
in presence of web attackers. For the proof we refer the reader
to Appendix~\ref{app:proof-oidc}.

\begin{theorem}\label{thm:theorem-1}
  Let $\oidcwebsystem^n$ be an OIDC web system with a network
  attacker. Then, $\oidcwebsystem^n$ is secure w.r.t. authentication
  and authorization. Let $\oidcwebsystem^w$ be an OIDC web system with
  web attackers. Then, $\oidcwebsystem^w$ is secure w.r.t. session
  integrity for authentication and authorization.
\end{theorem}

\subsection{Comparison to OAuth~2.0}
\label{sec:comparison-to-oauth}

As described in Section~\ref{sec:oidc-rel-to-oauth}, OIDC is based on
OAuth~2.0. Since a formal proof for the security of OAuth~2.0 was
conducted in~\cite{FettKuestersSchmitz-CCS-2016}, one might be tempted
to think that a proof for the security of OIDC requires little more
than an extension of the proof
in~\cite{FettKuestersSchmitz-CCS-2016}. The specific set of features
of OIDC introduces, however, important differences that call for new
formulations of security properties and require new proofs:

\noindent\textbf{Dynamic Discovery and Registration:} Due to the dynamic discovery and registration, RPs can directly
influence and manipulate the configuration data that is stored in
IdPs. In OAuth, this configuration data is fixed and assumed to be
``correct'', greatly limiting the options of the attacker. See, for example, the variant~\cite{MainkaMladenovSchwenkWich-EuroSP-2017} of the IdP Mix-up attack that only works in OIDC (mentioned in
Section~\ref{sec:att-mit-guide}).

\noindent\textbf{Different set of modes:} Compared to OAuth, OIDC introduces the hybrid
mode, but does not use the resource owner password credentials mode
and the client credentials mode.

\noindent\textbf{New endpoints, messages, and parameters:} With additional endpoints (and associated HTTPS messages), the attack surface of OIDC is, also for this reason,
larger than that of OAuth. The registration endpoints, for example,
could be used in ways that break the security of the protocol, which
is not possible in OAuth where these endpoints do not exist. In a
similar vein, new parameters like nonce, request\_uri, and the id
token, are contained in several messages (some of which are also
present in the original OAuth flow) and potentially change the
security requirements for these messages.

\noindent\textbf{Authentication mechanism:} The authentication
mechanisms employed by OIDC and OAuth are quite different. This shows,
in particular, in the fact that OIDC uses the id token mechanism for
authentication, while OAuth uses a different, non-standardized
mechanism. Additionally, unlike in OAuth, authentication can happen
multiple times during one OIDC flow (see the description of the hybrid mode in Section~\ref{sec:oidc-modes}). This greatly influences (the formulation of)
security properties, and hence, also the security proofs.

\medskip

\noindent In summary, taking all these differences into account, our
security proofs had to be carried out from scratch. At the same time,
our proof is more modular than the one
in~\cite{FettKuestersSchmitz-CCS-2016} due to the secondary security
properties we identified. Moreover, our security properties are
similar to the ones by Fett et
al.~in~\cite{FettKuestersSchmitz-CCS-2016} only on a high level. The
underlying definitions in many aspects differ from the ones used for
OAuth.\footnote{As an example,
  in~\cite{FettKuestersSchmitz-CCS-2016}, the definitions rely on a
  notion of \emph{OAuth sessions} which are defined by \emph{connected
    HTTP(S) messages}, i.e., messages that are created by a browser or
  server in response to another message. In our model, the attacker is
  involved in each flow of the protocol (for providing the client id,
  without receiving any prior message), making it hard to apply the
  notion of OAuth sessions. We instead define the properties using the
  existing session identifiers. (See
  Definitions~\ref{def:client-id-issued}, \ref{def:service-sessions},
  \ref{def:user-logged-in}--\ref{def:rp-acts-on-users-behalf} in
  Appendix~\ref{app:form-secur-prop} for details.)}

\subsection{Discussion}
\label{sec:discussion}

Using our detailed formal model, we have shown that OIDC enjoys a high
level of security regarding authentication, authorization, and session
integrity. To achieve this security, it is essential that implementors
follow the security guidelines that we stated in
Section~\ref{sec:attacks-impl-guid}. Clearly, in practice, this is not
always feasible---for example, many RPs want to include third-party
resources for advertisement or user tracking on their origins. As
pointed out, however, not following the security guidelines we
outline can lead to severe attacks.

We have shown the security of OIDC in the most comprehensive
model of the web infrastructure to date. Being a model, however, some
features of the web are not included in the FKS model, for example
browser plugins. Such technologies can under certain circumstances
also undermine the security of OIDC in a manner that is not reflected
in our model. Also, user-centric attacks such as phishing or
clickjacking attacks are also not covered in the model. 

Nonetheless, our formal analysis and the guidelines (along with the
attacks when these guidelines are not followed) provide a clear
picture of the security provided by OIDC for a large class of
adversaries.

\section{Related Work}
\label{sec:related-work}

As already mentioned in the introduction, the only previous works on
the security of OIDC are
\cite{LiMitchell-DIMVA-2016,MainkaMladenovSchwenkWich-EuroSP-2017}.
None of these works establish security guarantees for the OIDC
standard: In~\cite{LiMitchell-DIMVA-2016}, the authors find
implementation errors in deployments of Google Sign-In (which, as
mentioned before, is based on OIDC).
In~\cite{MainkaMladenovSchwenkWich-EuroSP-2017}, the authors describe
a variant of the IdP Mix-Up attack (see
Section~\ref{sec:attacks-impl-guid}), highlight the possibility of
SSRF attacks at RPs, and show some implementation-specific flaws. In
our work, however, we aim at establishing and proving security
properties for OIDC.

In general, there have been only few formal analysis efforts for web
applications, standards, and browsers so far. Most of the existing
efforts are based on formal representations of (parts of) web browsers
or very limited models of web mechanisms and applications~\cite{Caoetal-RAID-2014,kerschbaum-SP-2007-XSRF-prevention,AkhawBarthLamMitchellSong-CSF-2010,Armandoetal-SAML-CS-2013,Armandoetal-FMSE-2008,HedinBelloSabelfeld-JCS-2016,GuhaFredriksonLivshitsSwamy-SP-2011,BohannonPierce-USENIX-2010,BauerCaiJiaPassaroStrouckenTian-NDSS-2015,BielovaDevrieseMassacciPiessens-NSS-2011,BugliesiCalzavaraFocardiKhanTempesta-CSFW-2014,YoshihamaTateishiTabuchiMatsumoto-IEICET-2009,GrossPfitzmannSadeghi-ESORICS-2005,CalzavaraFocardiGrimmMaffei-CSFW-2016,BugliesiCalzavaraFocardiKhan-JCS-2015}.

Only
\cite{BansalBhargavanetal-JCS-2014,BansalBhargavanetal-POST-2013-WebSpi}
and
\cite{FettKuestersSchmitz-CCS-2016,FettKuestersSchmitz-CCS-2015,FettKuestersSchmitz-SP-2014,FettKuestersSchmitz-ESORICS-BrowserID-Primary-2015}
were based on a generic formal model of the web infrastructure. In
\cite{BansalBhargavanetal-JCS-2014}, Bansal, Bhargavan,
Delignat-Lavaud, and Maffeis analyze the security of OAuth~2.0 with the
tool ProVerif in the applied pi-calculus and the WebSpi library.
They identify previously unknown attacks on the OAuth~2.0
implementations of Facebook, Yahoo, Twitter, and many other websites.
They do not, however, establish security guarantees for OAuth~2.0 and
their model is much less expressive than the FKS model.

The relationship of our work to
\cite{FettKuestersSchmitz-CCS-2015,FettKuestersSchmitz-SP-2014,FettKuestersSchmitz-ESORICS-BrowserID-Primary-2015,FettKuestersSchmitz-CCS-2016}
has been discussed in detail throughout the paper.

\section{Conclusion}
\label{sec:conclusion}

Despite being the foundation for many popular and critical login
services, OpenID Connect had not been subjected to a detailed security
analysis, let alone a formal analysis, before. In this work, we filled
this gap.

We developed a detailed and comprehensive formal model of OIDC based
on the FKS model, a generic and expressive formal model of the web
infrastructure. Using this model, we stated central security
properties of OIDC regarding authentication, authorization, and
session integrity, and were able to show that OIDC fulfills these
properties in our model. By this, we could, for the first time,
provide solid security guarantees for one of the most widely deployed single
sign-on systems.

To avoid previously known and newly described attacks, we analyzed
OIDC with a set of practical and reasonable security measures and best
practices in place. We documented these security measures so that they
can now serve as guidelines for secure implementations of OIDC.

\section{Acknowledgements} 
This work was partially supported by \textit{Deu\-tsche
  Forschungsgemeinschaft} (DFG) through  Grant KU\ 1434/10-1.

\bibliographystyle{abbrv}
\appendices
\onecolumn

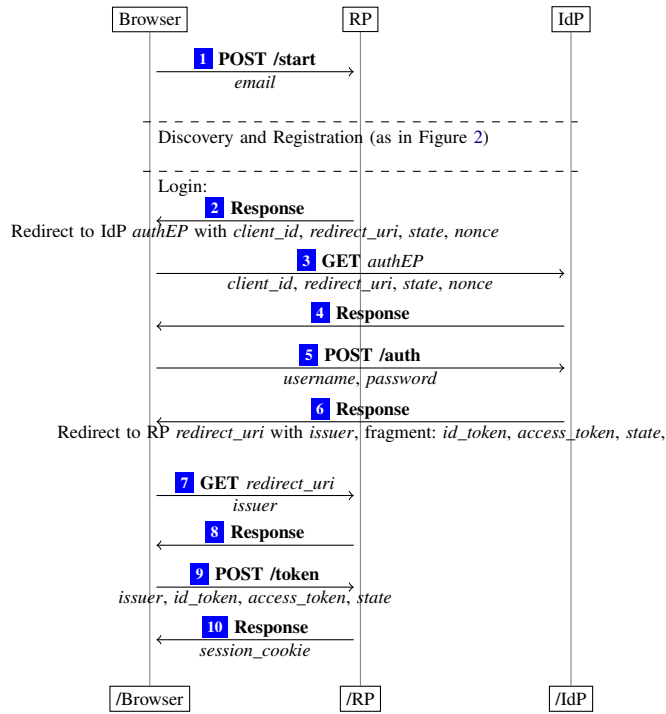
\begin{figure}[h]
  \centering
   \scriptsize{ \newlength\blockExtraHeightAAIeCIJBCBeDEEdabaJcGHdDBGIeCEccE
\settototalheight\blockExtraHeightAAIeCIJBCBeDEEdabaJcGHdDBGIeCEccE{\parbox{0.4\linewidth}{$\mi{email}$}}
\setlength\blockExtraHeightAAIeCIJBCBeDEEdabaJcGHdDBGIeCEccE{\dimexpr \blockExtraHeightAAIeCIJBCBeDEEdabaJcGHdDBGIeCEccE - 4ex/4}
\newlength\blockExtraHeightBAIeCIJBCBeDEEdabaJcGHdDBGIeCEccE
\settototalheight\blockExtraHeightBAIeCIJBCBeDEEdabaJcGHdDBGIeCEccE{\parbox{0.4\linewidth}{None}}
\setlength\blockExtraHeightBAIeCIJBCBeDEEdabaJcGHdDBGIeCEccE{\dimexpr \blockExtraHeightBAIeCIJBCBeDEEdabaJcGHdDBGIeCEccE - 4ex/4}
\newlength\blockExtraHeightCAIeCIJBCBeDEEdabaJcGHdDBGIeCEccE
\settototalheight\blockExtraHeightCAIeCIJBCBeDEEdabaJcGHdDBGIeCEccE{\parbox{0.4\linewidth}{None}}
\setlength\blockExtraHeightCAIeCIJBCBeDEEdabaJcGHdDBGIeCEccE{\dimexpr \blockExtraHeightCAIeCIJBCBeDEEdabaJcGHdDBGIeCEccE - 4ex/4}
\newlength\blockExtraHeightDAIeCIJBCBeDEEdabaJcGHdDBGIeCEccE
\settototalheight\blockExtraHeightDAIeCIJBCBeDEEdabaJcGHdDBGIeCEccE{\parbox{0.4\linewidth}{Redirect to IdP $\mi{authEP}$ with $\mi{client\_id}$, $\mi{redirect\_uri}$, $\mi{state}$, $\mi{nonce}$}}
\setlength\blockExtraHeightDAIeCIJBCBeDEEdabaJcGHdDBGIeCEccE{\dimexpr \blockExtraHeightDAIeCIJBCBeDEEdabaJcGHdDBGIeCEccE - 4ex/4}
\newlength\blockExtraHeightEAIeCIJBCBeDEEdabaJcGHdDBGIeCEccE
\settototalheight\blockExtraHeightEAIeCIJBCBeDEEdabaJcGHdDBGIeCEccE{\parbox{0.4\linewidth}{$\mi{client\_id}$, $\mi{redirect\_uri}$, $\mi{state}$, $\mi{nonce}$}}
\setlength\blockExtraHeightEAIeCIJBCBeDEEdabaJcGHdDBGIeCEccE{\dimexpr \blockExtraHeightEAIeCIJBCBeDEEdabaJcGHdDBGIeCEccE - 4ex/4}
\newlength\blockExtraHeightFAIeCIJBCBeDEEdabaJcGHdDBGIeCEccE
\settototalheight\blockExtraHeightFAIeCIJBCBeDEEdabaJcGHdDBGIeCEccE{\parbox{0.4\linewidth}{}}
\setlength\blockExtraHeightFAIeCIJBCBeDEEdabaJcGHdDBGIeCEccE{\dimexpr \blockExtraHeightFAIeCIJBCBeDEEdabaJcGHdDBGIeCEccE - 4ex/4}
\newlength\blockExtraHeightGAIeCIJBCBeDEEdabaJcGHdDBGIeCEccE
\settototalheight\blockExtraHeightGAIeCIJBCBeDEEdabaJcGHdDBGIeCEccE{\parbox{0.4\linewidth}{$\mi{username}$, $\mi{password}$}}
\setlength\blockExtraHeightGAIeCIJBCBeDEEdabaJcGHdDBGIeCEccE{\dimexpr \blockExtraHeightGAIeCIJBCBeDEEdabaJcGHdDBGIeCEccE - 4ex/4}
\newlength\blockExtraHeightHAIeCIJBCBeDEEdabaJcGHdDBGIeCEccE
\settototalheight\blockExtraHeightHAIeCIJBCBeDEEdabaJcGHdDBGIeCEccE{\parbox{0.4\linewidth}{Redirect to RP $\mi{redirect\_uri}$ with $\mi{issuer}$, fragment: $\mi{id\_token}$, $\mi{access\_token}$, $\mi{state}$,}}
\setlength\blockExtraHeightHAIeCIJBCBeDEEdabaJcGHdDBGIeCEccE{\dimexpr \blockExtraHeightHAIeCIJBCBeDEEdabaJcGHdDBGIeCEccE - 4ex/4}
\newlength\blockExtraHeightIAIeCIJBCBeDEEdabaJcGHdDBGIeCEccE
\settototalheight\blockExtraHeightIAIeCIJBCBeDEEdabaJcGHdDBGIeCEccE{\parbox{0.4\linewidth}{$\mi{issuer}$}}
\setlength\blockExtraHeightIAIeCIJBCBeDEEdabaJcGHdDBGIeCEccE{\dimexpr \blockExtraHeightIAIeCIJBCBeDEEdabaJcGHdDBGIeCEccE - 4ex/4}
\newlength\blockExtraHeightJAIeCIJBCBeDEEdabaJcGHdDBGIeCEccE
\settototalheight\blockExtraHeightJAIeCIJBCBeDEEdabaJcGHdDBGIeCEccE{\parbox{0.4\linewidth}{}}
\setlength\blockExtraHeightJAIeCIJBCBeDEEdabaJcGHdDBGIeCEccE{\dimexpr \blockExtraHeightJAIeCIJBCBeDEEdabaJcGHdDBGIeCEccE - 4ex/4}
\newlength\blockExtraHeightBAAIeCIJBCBeDEEdabaJcGHdDBGIeCEccE
\settototalheight\blockExtraHeightBAAIeCIJBCBeDEEdabaJcGHdDBGIeCEccE{\parbox{0.4\linewidth}{$\mi{issuer}$, $\mi{id\_token}$, $\mi{access\_token}$, $\mi{state}$}}
\setlength\blockExtraHeightBAAIeCIJBCBeDEEdabaJcGHdDBGIeCEccE{\dimexpr \blockExtraHeightBAAIeCIJBCBeDEEdabaJcGHdDBGIeCEccE - 4ex/4}
\newlength\blockExtraHeightBBAIeCIJBCBeDEEdabaJcGHdDBGIeCEccE
\settototalheight\blockExtraHeightBBAIeCIJBCBeDEEdabaJcGHdDBGIeCEccE{\parbox{0.4\linewidth}{$\mi{session\_cookie}$}}
\setlength\blockExtraHeightBBAIeCIJBCBeDEEdabaJcGHdDBGIeCEccE{\dimexpr \blockExtraHeightBBAIeCIJBCBeDEEdabaJcGHdDBGIeCEccE - 4ex/4}

 \begin{tikzpicture}
   \tikzstyle{xhrArrow} = [color=blue,decoration={markings, mark=at
    position 1 with {\arrow[color=blue]{triangle 45}}}, preaction
  = {decorate}]

    \matrix [column sep={2.8cm,between origins}, row sep=4ex]
  {

    \node[draw,anchor=base](Browser-start-0){Browser}; & \node[draw,anchor=base](RP-start-0){RP}; & \node[draw,anchor=base](IdP-start-0){IdP};\\
\node(Browser-0){}; & \node(RP-0){}; & \node(IdP-0){};\\[\blockExtraHeightAAIeCIJBCBeDEEdabaJcGHdDBGIeCEccE]
\node(Browser-1){}; & \node(RP-1){}; & \node(IdP-1){};\\[\blockExtraHeightBAIeCIJBCBeDEEdabaJcGHdDBGIeCEccE]
\node(Browser-2){}; & \node(RP-2){}; & \node(IdP-2){};\\[\blockExtraHeightCAIeCIJBCBeDEEdabaJcGHdDBGIeCEccE]
\node(Browser-3){}; & \node(RP-3){}; & \node(IdP-3){};\\[\blockExtraHeightDAIeCIJBCBeDEEdabaJcGHdDBGIeCEccE]
\node(Browser-4){}; & \node(RP-4){}; & \node(IdP-4){};\\[\blockExtraHeightEAIeCIJBCBeDEEdabaJcGHdDBGIeCEccE]
\node(Browser-5){}; & \node(RP-5){}; & \node(IdP-5){};\\[\blockExtraHeightFAIeCIJBCBeDEEdabaJcGHdDBGIeCEccE]
\node(Browser-6){}; & \node(RP-6){}; & \node(IdP-6){};\\[\blockExtraHeightGAIeCIJBCBeDEEdabaJcGHdDBGIeCEccE]
\node(Browser-7){}; & \node(RP-7){}; & \node(IdP-7){};\\[\blockExtraHeightHAIeCIJBCBeDEEdabaJcGHdDBGIeCEccE]
\node(Browser-8){}; & \node(RP-8){}; & \node(IdP-8){};\\[\blockExtraHeightIAIeCIJBCBeDEEdabaJcGHdDBGIeCEccE]
\node(Browser-9){}; & \node(RP-9){}; & \node(IdP-9){};\\[\blockExtraHeightJAIeCIJBCBeDEEdabaJcGHdDBGIeCEccE]
\node(Browser-10){}; & \node(RP-10){}; & \node(IdP-10){};\\[\blockExtraHeightBAAIeCIJBCBeDEEdabaJcGHdDBGIeCEccE]
\node(Browser-11){}; & \node(RP-11){}; & \node(IdP-11){};\\[\blockExtraHeightBBAIeCIJBCBeDEEdabaJcGHdDBGIeCEccE]
\node[draw,anchor=base](Browser-end-1){/Browser}; & \node[draw,anchor=base](RP-end-1){/RP}; & \node[draw,anchor=base](IdP-end-1){/IdP};\\
};
\draw[->] (Browser-0) to node [above=2.6pt, anchor=base]{\protostep{oicif-start-req} \textbf{POST /start}} node [below=-8pt, text width=0.5\linewidth, anchor=base]{\begin{center} $\mi{email}$\end{center}} (RP-0); 

\draw [dashed] (Browser-1.west) -- (IdP-1.east);
\node[draw=none,anchor=northwest,below=2ex,right=1ex] at (Browser-1.west) {Discovery and Registration (as in Figure~\ref{fig:oidc-auth-code-flow})};

\draw [dashed] (Browser-2.west) -- (IdP-2.east);
\node[draw=none,anchor=northwest,below=2ex,right=1ex] at (Browser-2.west) {Login:};

\draw[->] (RP-3) to node [above=2.6pt, anchor=base]{\protostep{oicif-start-resp} \textbf{Response}} node [below=-8pt, text width=0.5\linewidth, anchor=base]{\begin{center} Redirect to IdP $\mi{authEP}$ with $\mi{client\_id}$, $\mi{redirect\_uri}$, $\mi{state}$, $\mi{nonce}$\end{center}} (Browser-3); 

\draw[->] (Browser-4) to node [above=2.6pt, anchor=base]{\protostep{oicif-idp-auth-req-1} \textbf{GET $\mi{authEP}$}} node [below=-8pt, text width=0.5\linewidth, anchor=base]{\begin{center} $\mi{client\_id}$, $\mi{redirect\_uri}$, $\mi{state}$, $\mi{nonce}$\end{center}} (IdP-4); 

\draw[->] (IdP-5) to node [above=2.6pt, anchor=base]{\protostep{oicif-idp-auth-resp-1} \textbf{Response}} node [below=-8pt, text width=0.5\linewidth, anchor=base]{\begin{center} \end{center}} (Browser-5); 

\draw[->] (Browser-6) to node [above=2.6pt, anchor=base]{\protostep{oicif-idp-auth-req-2} \textbf{POST /auth}} node [below=-8pt, text width=0.5\linewidth, anchor=base]{\begin{center} $\mi{username}$, $\mi{password}$\end{center}} (IdP-6); 

\draw[->] (IdP-7) to node [above=2.6pt, anchor=base]{\protostep{oicif-idp-auth-resp-2} \textbf{Response}} node [below=-8pt, text width=0.5\linewidth, anchor=base]{\begin{center} Redirect to RP $\mi{redirect\_uri}$ with $\mi{issuer}$, fragment: $\mi{id\_token}$, $\mi{access\_token}$, $\mi{state}$,\end{center}} (Browser-7); 

\draw[->] (Browser-8) to node [above=2.6pt, anchor=base]{\protostep{oicif-redir-ep-req} \textbf{GET $\mi{redirect\_uri}$}} node [below=-8pt, text width=0.5\linewidth, anchor=base]{\begin{center} $\mi{issuer}$\end{center}} (RP-8); 

\draw[->] (RP-9) to node [above=2.6pt, anchor=base]{\protostep{oicif-redir-ep-resp} \textbf{Response}} node [below=-8pt, text width=0.5\linewidth, anchor=base]{\begin{center} \end{center}} (Browser-9); 

\draw[->] (Browser-10) to node [above=2.6pt, anchor=base]{\protostep{oifif-redir-ep-token-req} \textbf{POST /token}} node [below=-8pt, text width=0.5\linewidth, anchor=base]{\begin{center} $\mi{issuer}$, $\mi{id\_token}$, $\mi{access\_token}$, $\mi{state}$\end{center}} (RP-10); 

\draw[->] (RP-11) to node [above=2.6pt, anchor=base]{\protostep{oicif-redir-ep-token-resp} \textbf{Response}} node [below=-8pt, text width=0.5\linewidth, anchor=base]{\begin{center} $\mi{session\_cookie}$\end{center}} (Browser-11); 

\begin{pgfonlayer}{background}
\draw [color=gray] (Browser-start-0) -- (Browser-end-1);
\draw [color=gray] (RP-start-0) -- (RP-end-1);
\draw [color=gray] (IdP-start-0) -- (IdP-end-1);
\end{pgfonlayer}
\end{tikzpicture}}
  \caption{OpenID Connect implicit mode.}
  \label{fig:oidc-implicit-flow}
\end{figure}

\begin{figure}[h]
  \centering
  \input{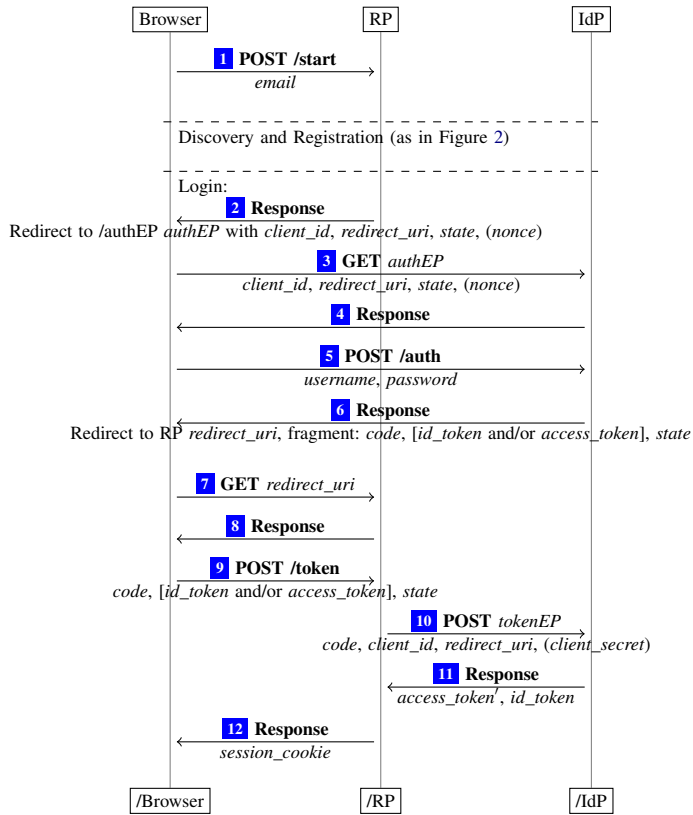}
  \caption{OpenID Connect hybrid mode.}
  \label{fig:oidc-hybrid-flow}
\end{figure}

\section{The IdP Mix-Up Attack}
\label{app:idp-mix-up}

\begin{figure}[tbp]
  \input{figure-oidc-hybrid-flow-attack}
  \caption{Attack on OpenID Connect hybrid mode.}
  \label{fig:oic-hf-att}
\end{figure}

As described in Section~\ref{sec:idp-mix-up} above, in the IdP Mix-Up
attack, an honest RP gets confused about which IdP is used in a login
flow. The honest RP assumes that the login uses the attacker's IdP and
interacts with this IdP, while the user's browser interacts with an
honest IdP and relays the data acquired at this IdP to the RP. As a
result, the attacker learns information such as authorization codes
and access tokens he is not supposed to know and that allow him to
break the authentication and authorization properties. 

There exist several variants of this
attack~\cite{FettKuestersSchmitz-CCS-2016,MainkaMladenovSchwenkWich-EuroSP-2017}.
Here, we describe two variants of this attack using the hybrid mode of
OIDC. The normal flow of the hybrid mode is depicted in
Figure~\ref{fig:oidc-hybrid-flow}, the attack is depicted in
Figures~\ref{fig:oic-hf-att} and~\ref{fig:oic-hf-att-contd} (without
the mitigation against the Mix-Up attack presented in
Section~\ref{sec:idp-mix-up}).

To start the login flow, the user selects an IdP at RP (by entering
her email address) in Step~\refprotostep{oichf-att-start-req}. This
step is the only difference between the two variants that we describe:
In Variant~1, the user selects a malicious IdP, say AIdP. In
Variant~2, the user selects an honest IdP, but the request is
intercepted by the attacker and altered such that the attacker
replaces the honest IdP by AIdP ($\mi{email}$
is replaced by $\mi{email}'$
in Steps~\refprotostep{oichf-att-start-req}
and~\refprotostep{oichf-att-start-req-manipulated} in
Figure~\ref{fig:oic-hf-att}).\footnote{This initial request is often
  unencrypted in practice, see~\cite{FettKuestersSchmitz-CCS-2016}.}

Now, RP starts with the discovery phase of the protocol. As RP thinks
that the user wants to login with AIdP, it retrieves the OIDC
configuration from AIdP (Steps~\refprotostep{oichf-att-conf-req}
and~\refprotostep{oichf-att-conf-resp}). In this configuration, the
attacker does not let all endpoint URLs point to himself, as would be
usual for OIDC, but instead sets the authorization endpoint to be the
one of HIdP. Next, the RP registers itself at AIdP
(Steps~\refprotostep{oichf-att-reg-req}
and~\refprotostep{oichf-att-reg-resp}). In this step, AIdP issues the
same client id to RP which RP is registered with at HIdP (client id
are always public). This is important as HIdP will later redirect the
user's browser back to RP and checks the redirect URI based on the
client id.

Next, RP redirects the user's browser to HIdP (Variant~1) or AIdP
(Variant~2) in Step~\refprotostep{oichf-att-start-resp}. In Variant~1
of the attack, a vigilant user might now be able to detect that she
tried to log in using AIdP but instead is redirected to HIdP. This
does not happen in Variant~2, but here the attacker needs to replace
the redirection to AIdP by a redirection to HIdP (which should not be
any problem if he succeeded in altering the first step of the
protocol).

The user then authenticates at HIdP and is redirected back to RP along
with an authorization code and an access token (depending on the
sub-mode of the hybrid flow, IdPs do not send id tokens in this step).
Now, RP retrieves the authorization code and the access token from the
user's browser and continues the login flow. As RP still assumes that
AIdP is used in this case, it tries to redeem the authorization code
for an id token (and a second access token) at AIdP in
Step~\refprotostep{oichf-att-token-req}.

\begin{figure}[btp]
  \input{figure-oidc-hybrid-flow-attack-contd}
  \caption{Attack on OpenID Connect hybrid mode (continued).}
  \label{fig:oic-hf-att-contd}
\end{figure}

As the authorization code has not been redeemed at HIdP yet, the code
is still valid and the attacker may start a second login flow
(pretending to be the user) at RP
(Steps~\refprotostep{oichf-att-start-req-att}ff.). The attacker skips
the authentication at HIdP and returns to RP with the authorization
code he has learned before. RP now redeems this code at HIdP and
receives an id token issued for the honest user and consequently
assumes that the attacker has the identity of the user and logs the
attacker in.

In another variation of the attack, if HIdP does not issue client
secrets to RPs, the attacker can also redeem the authorization code by
himself (Steps~\refprotostep{oichf-att-token-req-att2}f.). In this
case, the attacker receives an access token valid for the user's
account. With this access token, he can retrieve data of the user or
act on the user's behalf at HIdP. (As he redeems the authorization
code, he cannot use it to log himself into the RP in this case.)

In any case, the attacker can also respond to the authorization code
sent to his token endpoint in
Step~\refprotostep{oichf-att-token-resp-att} with a mock access token
and a mock id token (which will not be used in the following). In the
next step, the RP might then use the access token learned \emph{from
  the honest IdP} in Step~\refprotostep{oichf-att-redir-ep-token-req}
to retrieve data of the user from AIdP
(Steps~\refprotostep{oichf-att-token-resp-att}ff.).\footnote{Depending
  on the RP implementation, the RP might choose to use the mock access
  token or the one learned from the honest IdP in this step. In the
  real-world implementation mod\_auth\_openidc, the access token from
  the honest IdP was used.} Then the attacker learns also this access
token, which (as described in the paragraph above) grants him
unauthorized access to the user's account at HIdP.

This shows that, using the IdP Mix-Up attack, an attacker can
successfully impersonate users at RPs and access their data at honest
IdPs. The mitigation presented in Section~\ref{sec:idp-mix-up} would
have prevented the attack in
Step~\refprotostep{oichf-att-redir-ep-req}ff.

\clearpage
\section{The FKS Web Model}\label{app:web-model}

In this and the following two sections, we present the FKS model for
the web infrastructure as proposed
in~\cite{FettKuestersSchmitz-SP-2014},
\cite{FettKuestersSchmitz-TR-BrowserID-Primary-2015},
and~\cite{FettKuestersSchmitz-TR-OAuth-2015}, along with the addition
of a generic model for HTTPS web servers that harmonizes the behavior
of such servers and facilitates easier proofs.

\subsection{Communication Model}\label{app:communication-model}

We here present details and definitions on the basic concepts of the
communication model.

\subsubsection{Terms, Messages and Events} 
The signature $\Sigma$ for the terms and
messages considered in this work is the union of the following
pairwise disjoint sets of function symbols:
\begin{itemize}
\item constants $C = \addresses\,\cup\, \mathbb{S}\cup
  \{\True,\bot,\notdef\}$ where the three sets are pairwise disjoint,
  $\mathbb{S}$ is interpreted to be the set of ASCII strings
  (including the empty string $\varepsilon$), and $\addresses$ is
  interpreted to be a set of (IP) addresses,
\item function symbols for public keys, (a)symmetric
  en\-cryp\-tion/de\-cryp\-tion, and signatures: $\mathsf{pub}(\cdot)$,
  $\enc{\cdot}{\cdot}$, $\dec{\cdot}{\cdot}$, $\encs{\cdot}{\cdot}$,
  $\decs{\cdot}{\cdot}$, $\sig{\cdot}{\cdot}$,
  $\checksigThree{\cdot}{\cdot}{\cdot}$, and $\unsig{\cdot}$,
\item $n$-ary sequences $\an{}, \an{\cdot}, \an{\cdot,\cdot},
  \an{\cdot,\cdot,\cdot},$ etc., and
\item projection symbols $\pi_i(\cdot)$ for all $i \in \mathbb{N}$.
\end{itemize}
For strings (elements in $\mathbb{S}$), we use a
specific font. For example, $\cHttpReq$ and $\cHttpResp$
are strings. We denote by $\dns\subseteq \mathbb{S}$ the
set of domains, e.g., $\str{example.com}\in \dns$.  We
denote by $\methods\subseteq \mathbb{S}$ the set of methods
used in HTTP requests, e.g., $\mGet$, $\mPost\in \methods$.

The equational theory associated with the signature
$\Sigma$ is given in Figure~\ref{fig:equational-theory}.

\begin{figure}
\begin{align}
\dec{\enc x{\pub(y)}}{y} &= x\\
\decs{\encs x{y}}{y} &= x\\
\checksig{\sig{x}{y}}{\pub(y)} &= \True\\
\unsig{\sig{x}{y}} &= x\\
\pi_i(\an{x_1,\dots,x_n}) &= x_i \text{\;\;if\ } 1 \leq i \leq n \\
\proj{j}{\an{x_1,\dots,x_n}} &= \notdef \text{\;\;if\ } j
\not\in \{1,\dots,n\}
\end{align}
\caption{Equational theory for $\Sigma$.}\label{fig:equational-theory}
\end{figure}

\begin{definition}[Nonces and Terms]\label{def:terms}
  By $X=\{x_0,x_1,\dots\}$ we denote a set of variables and by
  $\nonces$ we denote an infinite set of constants (\emph{nonces})
  such that $\Sigma$, $X$, and $\nonces$ are pairwise disjoint. For
  $N\subseteq\nonces$, we define the set $\gterms_N(X)$ of
  \emph{terms} over $\Sigma\cup N\cup X$ inductively as usual: (1) If
  $t\in N\cup X$, then $t$ is a term. (2) If $f\in \Sigma$ is an
  $n$-ary function symbol in $\Sigma$ for some $n\ge 0$ and
  $t_1,\ldots,t_n$ are terms, then $f(t_1,\ldots,t_n)$ is a term.
\end{definition}

By $\equiv$ we denote the congruence relation on $\terms(X)$ induced
by the theory associated with $\Sigma$. For example, we have that
$\pi_1(\dec{\enc{\an{\str{a},\str{b}}}{\pub(k)}}{k})\equiv \str{a}$.

\begin{definition}[Ground Terms, Messages, Placeholders, Protomessages]\label{def:groundterms-messages-placeholders-protomessages}
  By $\gterms_N=\gterms_N(\emptyset)$, we denote the set of all terms
  over $\Sigma\cup N$ without variables, called \emph{ground terms}.
  The set $\messages$ of messages (over $\nonces$) is defined to be
  the set of ground terms $\gterms_{\nonces}$. 
  
  We define the set $V_{\text{process}} = \{\nu_1, \nu_2, \dots\}$ of
  variables (called placeholders). The set $\messages^\nu :=
  \gterms_{\nonces}(V_{\text{process}})$ is called the set of \emph{protomessages},
  i.e., messages that can contain placeholders.
\end{definition}

\begin{example}
  For example, $k\in \nonces$ and $\pub(k)$ are messages, where $k$
  typically models a private key and $\pub(k)$ the corresponding
  public key. For constants $a$, $b$, $c$ and the nonce $k\in
  \nonces$, the message $\enc{\an{a,b,c}}{\pub(k)}$ is interpreted to
  be the message $\an{a,b,c}$ (the sequence of constants $a$, $b$,
  $c$) encrypted by the public key $\pub(k)$.
\end{example}

\begin{definition}[Normal Form]
  Let $t$ be a term. The \emph{normal form} of $t$ is acquired by
  reducing the function symbols from left to right as far as possible
  using the equational theory shown in
  Figure~\ref{fig:equational-theory}. For a term $t$, we denote its
  normal form as $t\nf$.
\end{definition}

\begin{definition}[Pattern Matching]\label{def:pattern-matching}
  Let $\mi{pattern} \in \terms(\{*\})$ be a term containing the
  wildcard (variable $*$). We say that a term $t$ \emph{matches}
  $\mi{pattern}$ iff $t$ can be acquired from $\mi{pattern}$ by
  replacing each occurrence of the wildcard with an arbitrary term
  (which may be different for each instance of the wildcard). We write
  $t \sim \mi{pattern}$. For a sequence of patterns $\mi{patterns}$ we
  write $t \dot{\sim} \mi{patterns}$ to denote that $t$ matches at
  least one pattern in $\mi{patterns}$.

  For a term $t'$ we write $t'|\, \mi{pattern}$ to denote the term
  that is acquired from $t'$ by removing all immediate subterms of
  $t'$ that do not match $\mi{pattern}$. 
\end{definition}

\begin{example}
  For example, for a pattern $p = \an{\top,*}$ we have that $\an{\top,42} \sim p$, $\an{\bot,42} \not\sim p$, and \[\an{\an{\bot,\top},\an{\top,23},\an{\str{a},\str{b}},\an{\top,\bot}} |\, p = \an{\an{\top,23},\an{\top,\bot}}\ .\]
\end{example}

\begin{definition}[Variable Replacement]
  Let $N\subseteq \nonces$, $\tau \in \gterms_N(\{x_1,\ldots,x_n\})$,
  and $t_1,\ldots,t_n\in \gterms_N$. 

  By
  $\tau[t_1\!/\!x_1,\ldots,t_n\!/\!x_n]$ we denote the (ground) term obtained
  from $\tau$ by replacing all occurrences of $x_i$ in $\tau$ by
  $t_i$, for all $i\in \{1,\ldots,n\}$.
\end{definition}

\begin{definition}[Events and Protoevents]
  An \emph{event (over $\addresses$ and $\messages$)} is a term of the
  form $\an{a, f, m}$, for $a$, $f\in \addresses$ and $m \in
  \messages$, where $a$ is interpreted to be the receiver address and
  $f$ is the sender address. We denote by $\events$ the set of all
  events. Events over $\addresses$ and $\messages^\nu$ are called
  \emph{protoevents} and are denoted $\events^\nu$. By
  $2^{\events\an{}}$ (or $2^{\events^\nu\an{}}$, respectively) we
  denote the set of all sequences of (proto)events, including the
  empty sequence (e.g., $\an{}$, $\an{\an{a, f, m}, \an{a', f', m'},
    \dots}$, etc.). 
\end{definition}

\subsubsection{Notations}\label{app:notation}

\begin{definition}[Sequence Notations]
  For a sequence $t = \an{t_1,\dots,t_n}$ and a set $s$ we
  use $t \subsetPairing s$ to say that $t_1,\dots,t_n \in
  s$.  We define $\left. x \inPairing t\right. \iff \exists
  i: \left. t_i = x\right.$.
  For a term $y$ we write $t \plusPairing y$ to denote the sequence
  $\an{t_1,\dots,t_n,y}$. For a sequence $r = \an{r_1, \dots, r_m}$ we write $t \cup r$ to denote the sequence $\an{t_1, \dots, t_n, r_1, \dots, r_m}$.
  For a finite set $M$ with $M = \{m_1, \dots,m_n\}$ we use
  $\an{M}$ to denote the term of the form
  $\an{m_1,\dots,m_n}$. (The order of the elements does not
  matter; one is chosen arbitrarily.) 
\end{definition}

\begin{definition}\label{def:dictionaries}
  A \emph{dictionary over $X$ and $Y$} is a term of the
  form \[\an{\an{k_1, v_1}, \dots, \an{k_n,v_n}}\] where
  $k_1, \dots,k_n \in X$, $v_1,\dots,v_n \in Y$.
 We call every term $\an{k_i,v_i}$, $i\in
  \{1,\ldots,n\}$, an \emph{element} of the dictionary with
  key $k_i$ and value $v_i$.  We often write $\left[k_1:
    v_1, \dots, k_i:v_i,\dots,k_n:v_n\right]$ instead of
  $\an{\an{k_1, v_1}, \dots, \an{k_n,v_n}}$. We denote the
  set of all dictionaries over $X$ and $Y$ by $\left[X
    \times Y\right]$.
\end{definition}
We note that the empty dictionary is equivalent to the
empty sequence, i.e.,  $[] = \an{}$.  Figure
\ref{fig:dictionaries} shows the short notation for
dictionary operations. For a dictionary $z = \left[k_1:
  v_1, k_2: v_2,\dots, k_n:v_n\right]$ we write $k \in z$ to
say that there exists $i$ such that $k=k_i$. We write
$z[k_j]$ to refer to the value $v_j$. (Note that if a dictionary contains two elements $\an{k, v}$ and $\an{k, v'}$, then the notations and operations for dictionaries apply non-deterministically to one of both elements.) If $k \not\in z$, we
set $z[k] := \an{}$.

\begin{figure}[htb!]\centering
  \begin{align}
    \left[k_1: v_1, \dots, k_i:v_i,\dots,k_n:v_n\right][k_i] = v_i%
  \end{align}\vspace{-2.5em}
  \begin{align}
    \nonumber \left[k_1: v_1, \dots, k_{i-1}:v_{i-1},k_i: v_i, k_{i+1}:v_{i+1}\dots,k_n: v_n\right]-k_i =\\
         \left[k_1: v_1, \dots, k_{i-1}:v_{i-1},k_{i+1}:v_{i+1}\dots,k_n: v_n\right]
  \end{align}
  \caption{Dictionary operators with $1\le i\le n$.}\label{fig:dictionaries}
\end{figure}

Given a term $t = \an{t_1,\dots,t_n}$, we can refer to any
subterm using a sequence of integers. The subterm is
determined by repeated application of the projection
$\pi_i$ for the integers $i$ in the sequence. We call such
a sequence a \emph{pointer}:

\begin{definition}\label{def:pointer}
  A \emph{pointer} is a sequence of non-negative
  integers. We write $\tau.\ptr{p}$ for the application of
  the pointer $\ptr{p}$ to the term $\tau$. This operator
  is applied from left to right. For pointers consisting of
  a single integer, we may omit the sequence braces for
  brevity.
\end{definition}

\begin{example}
  For the term $\tau = \an{a,b,\an{c,d,\an{e,f}}}$ and the
  pointer $\ptr{p} = \an{3,1}$, the subterm of $\tau$ at
  the position $\ptr{p}$ is $c =
  \proj{1}{\proj{3}{\tau}}$. Also, $\tau.3.\an{3,1} =
  \tau.3.\ptr{p} = \tau.3.3.1 = e$.
\end{example}

To improve readability, we try to avoid writing, e.g.,
$\compn{o}{2}$ or $\proj{2}{o}$ in this document. Instead,
we will use the names of the components of a sequence that
is of a defined form as pointers that point to the
corresponding subterms. E.g., if an \emph{Origin} term is
defined as $\an{\mi{host}, \mi{protocol}}$ and $o$ is an
Origin term, then we can write $\comp{o}{protocol}$ instead
of $\proj{2}{o}$ or $\compn{o}{2}$. See also
Example~\ref{ex:url-pointers}.

\subsubsection{Atomic Processes, Systems and Runs} 
An atomic process takes its current state and an
event as input, and then (non-deterministi\-cally) outputs a new state
and a set of events.
\begin{definition}[Generic Atomic Processes and Systems]\label{def:atomic-process-and-process}
  A \emph{(generic) \ap} is a tuple $$p = (I^p, Z^p, R^p, s^p_0)$$ where
  $I^p \subseteq \addresses$, $Z^p \in \terms$ is a set of states,
  $R^p\subseteq (\events \times Z^p) \times (2^{\events^\nu\an{}}
  \times \terms(V_{\text{process}}))$ (input event and old state map to sequence of
  output events and new state), and $s^p_0\in Z^p$ is the initial
  state of $p$. For any new state $s$ and any sequence of nonces
  $(\eta_1, \eta_2, \dots)$ we demand that $s[\eta_1/\nu_1,
  \eta_2/\nu_2, \dots] \in Z^p$. A \emph{system} $\process$ is a
  (possibly infinite) set of \aps.
\end{definition}

\begin{definition}[Configurations]
  A \emph{configuration of a system $\process$} is a tuple $(S, E, N)$
  where the state of the system $S$ maps every atomic process
  $p\in \process$ to its current state $S(p)\in Z^p$, the sequence of
  waiting events $E$ is an infinite sequence\footnote{Here: Not in the
    sense of terms as defined earlier.} $(e_1, e_2, \dots)$ of events
  waiting to be delivered, and $N$ is an infinite sequence of nonces
  $(n_1, n_2, \dots)$.
\end{definition}

\begin{definition}[Concatenating terms and sequences]
  For a term $a = \an{a_1, \dots, a_i}$ and a sequence $b = (b_1, b_2,
  \dots)$, we define the \emph{concatenation} as $a \cdot b := (a_1,
  \dots, a_i, b_1, b_2, \dots)$.
  
\end{definition}

\begin{definition}[Subtracting from Sequences]
  For a sequence $X$ and a set or sequence $Y$ we define $X \setminus
  Y$ to be the sequence $X$ where for each element in $Y$, a
  non-deterministically chosen occurence of that element in $X$ is
  removed.
\end{definition}

\begin{definition}[Processing Steps]\label{def:processing-step}
  A \emph{processing step of the system $\process$} is of the form
  \[(S,E,N) \xrightarrow[p \rightarrow E_{\text{out}}]{e_\text{in}
    \rightarrow p} (S', E', N')\]
  where
  \begin{enumerate}
  \item $(S,E,N)$ and $(S',E',N')$ are configurations of $\process$,
  \item $e_\text{in} = \an{a, f, m} \in E$ is an event,
  \item $p \in \process$ is a process,
  \item $E_{\text{out}}$ is a sequence (term) of events
  \end{enumerate}
  such that there exists 
  \begin{enumerate}
  \item a sequence (term)
    $E^\nu_{\text{out}} \subseteq 2^{\events^\nu\an{}}$ of protoevents,
  \item a term $s^\nu \in \gterms_{\nonces}(V_{\text{process}})$, 
  \item a sequence $(v_1, v_2, \dots, v_i)$ of all placeholders appearing in $E^\nu_{\text{out}}$ (ordered lexicographically),
  \item a sequence $N^\nu = (\eta_1, \eta_2, \dots, \eta_i)$ of the first $i$ elements in $N$ 
  \end{enumerate}
  with
  \begin{enumerate}
  \item $((e_{\text{in}}, S(p)), (E^\nu_{\text{out}}, s^\nu)) \in R^p$
    and $a \in I^p$,
  \item $E_{\text{out}} = E^\nu_{\text{out}}[m_1/v_1, \dots, m_i/v_i]$
  \item $S'(p) = s^\nu[m_1/v_1, \dots, m_i/v_i]$ and $S'(p') = S(p')$ for all $p' \neq p$
  \item $E' = E_{\text{out}} \cdot (E \setminus \{e_{\text{in}}\})$ 
  \item $N' = N \setminus N^\nu$ 
  \end{enumerate}
  We may omit the superscript and/or subscript of the arrow.
\end{definition} 
Intuitively, for a processing step, we select one of the processes in
$\process$, and call it with one of the events in the list of waiting
events $E$. In its output (new state and output events), we replace
any occurences of placeholders $\nu_x$ by ``fresh'' nonces from $N$
(which we then remove from $N$). The output events are then prepended
to the list of waiting events, and the state of the process is
reflected in the new configuration.

\begin{definition}[Runs]
  Let $\process$ be a system, $E^0$ be sequence of events, and $N^0$ be
  a sequence of nonces. A \emph{run $\rho$ of a system $\process$
    initiated by $E^0$ with nonces $N^0$} is a finite sequence of
  configurations $((S^0, E^0, N^0),\dots,$ $(S^n, E^n, N^n))$ or an infinite sequence
  of configurations $((S^0, E^0, N^0),\dots)$ such that $S^0(p) = s_0^p$ for
  all $p \in \process$ and $(S^i, E^i, N^i) \xrightarrow{} (S^{i+1},
  E^{i+1}, N^{i+1})$ for all $0 \leq i < n$ (finite run) or for all $i \geq 0$  
  (infinite run).

  We denote the state $S^n(p)$ of a process $p$ at the end of a run $\rho$ by $\rho(p)$.
\end{definition}

Usually, we will initiate runs with a set $E^0$ containing infinite
trigger events of the form $\an{a, a, \str{TRIGGER}}$ for each $a \in
\addresses$, interleaved by address.

\subsubsection{Atomic Dolev-Yao Processes}  We next define
atomic Dolev-Yao processes, for which we require that the
messages and states that they output can be computed (more
formally, derived) from the current input event and
state. For this purpose, we first define what it means to
derive a message from given messages.

\begin{definition}[Deriving Terms]
  Let $M$ be a set of ground terms. We say that \emph{a
    term $m$ can be derived from $M$ with placeholders $V$} if there
  exist $n\ge 0$, $m_1,\ldots,m_n\in M$, and $\tau\in
  \gterms_{\emptyset}(\{x_1,\ldots,x_n\} \cup V)$ such that $m\equiv
  \tau[m_1/x_1,\ldots,m_n/x_n]$. We denote by $d_V(M)$ the set of all
  messages that can be derived from $M$ with variables $V$.
\end{definition}
For example, $a\in d_{\{\}}(\{\enc{\an{a,b,c}}{\pub(k)}, k\})$.

\begin{definition}[Atomic Dolev-Yao Process] \label{def:adyp} An \emph{atomic Dolev-Yao process
    (or simply, a DY process)} is a tuple $p = (I^p, Z^p,$ $R^p,
  s^p_0)$ such that $(I^p, Z^p, R^p, s^p_0)$ is an atomic process and
  (1) $Z^p \subseteq \gterms_{\nonces}$ (and hence, $s^p_0\in
  \gterms_{\nonces}$), and (2) for all events $e \in \events$,
  sequences of protoevents $E$, $s\in \gterms_{\nonces}$, $s'\in
  \gterms_{\nonces}(V_{\text{process}})$, with $((e, s), (E, s')) \in R^p$ it holds
  true that $E$, $s' \in d_{V_{\text{process}}}(\{e,s\})$.
\end{definition}

\begin{definition}[Atomic Attacker Process]\label{def:atomicattacker}
  An \emph{(atomic) attacker process for a set of sender addresses
    $A\subseteq \addresses$} is an atomic DY process $p = (I, Z, R,
  s_0)$ such that for all events $e$, and $s\in \gterms_{\nonces}$ we
  have that $((e, s), (E,s')) \in R$ iff $s'=\an{e, E, s}$ and
  $E=\an{\an{a_1, f_1, m_1}, \dots, \an{a_n, f_n, m_n}}$ with $n \in
  \mathbb{N}$, $a_1,\dots,a_n\in \addresses$, $f_0,\dots,f_n\in A$,
  $m_1,\dots,m_n\in d_{V_{\text{process}}}(\{e,s\})$.
\end{definition}

\subsection{Scripts}
We define scripts, which model client-side scripting
technologies, such as JavaScript. Scripts are defined
similarly to DY processes.
\begin{definition}[Placeholders for Scripts]\label{def:placeholder-sp}
  By $V_{\text{script}} = \{\lambda_1, \dots\}$ we denote an infinite set of variables
  used in scripts.
\end{definition}

\begin{definition}[Scripts]\label{def:sp} A \emph{script} is a relation $R\subseteq \terms \times
  \terms(V_{\text{script}})$ such that for all $s \in \terms$, $s' \in \terms(V_{\text{script}})$ with
  $(s, s') \in R$ it follows that $s'\in d_{V_{\text{script}}}(s)$.
\end{definition}
A script is called by the browser which provides it with state
information (such as the script's last state and limited information
about the browser's state) $s$. The script then outputs a term $s'$,
which represents the new internal state and some command which is
interpreted by the browser. The term $s'$ may contain variables
$\lambda_1, \dots$ which the browser will replace by (otherwise
unused) placeholders $\nu_1,\dots$ which will be replaced by nonces
once the browser DY process finishes (effectively providing the script
with a way to get ``fresh'' nonces).

Similarly to an attacker process, we define the
\emph{attacker script} $\Rasp$: 
\begin{definition}[Attacker Script]
  The attacker script $\Rasp$ outputs everything that is derivable
  from the input, i.e., $\Rasp=\{(s, s')\mid s\in \terms, s'\in
  d_{V_{\text{script}}}(s)\}$.
\end{definition}

\subsection{Web System}\label{app:websystem}

The web infrastructure and web applications are formalized by what is
called a web system. A web system contains, among others, a (possibly
infinite) set of DY processes, modeling web browsers, web servers, DNS
servers, and attackers (which may corrupt other entities, such as
browsers).

\begin{definition}\label{def:websystem}
  A \emph{web system $\completewebsystem=(\websystem,
    \scriptset,\mathsf{script}, E^0)$} is a tuple with its
  components defined as follows:

  The first component, $\websystem$, denotes a system
  (a set of DY processes) and is partitioned into the
  sets $\mathsf{Hon}$, $\mathsf{Web}$, and $\mathsf{Net}$
  of honest, web attacker, and network attacker processes,
  respectively.  

  Every $p\in \mathsf{Web} \cup \mathsf{Net}$ is an
  attacker process for some set of sender addresses
  $A\subseteq \addresses$. For a web attacker $p\in
  \mathsf{Web}$, we require its set of addresses $I^p$ to
  be disjoint from the set of addresses of all other web
  attackers and honest processes, i.e., $I^p\cap I^{p'} =
  \emptyset$ for all $p' \in \mathsf{Hon} \cup
  \mathsf{Web}$. Hence, a web attacker cannot listen to
  traffic intended for other processes. Also, we require
  that $A=I^p$, i.e., a web attacker can only use sender
  addresses it owns. Conversely, a network attacker may
  listen to all addresses (i.e., no restrictions on $I^p$)
  and may spoof all addresses (i.e., the set $A$ may be
  $\addresses$).

  Every $p \in \mathsf{Hon}$ is a DY process which
  models either a \emph{web server}, a \emph{web browser},
  or a \emph{DNS server}, as further described in the
  following subsections. Just as for web attackers, we
  require that $p$ does not spoof sender addresses and that
  its set of addresses $I^p$ is disjoint from those of
  other honest processes and the web attackers. 

  The second component, $\scriptset$, is a finite set of
  scripts such that $\Rasp\in \scriptset$. The third
  component, $\mathsf{script}$, is an injective mapping
  from $\scriptset$ to $\mathbb{S}$, i.e., by
  $\mathsf{script}$ every $s\in \scriptset$ is assigned its
  string representation $\mathsf{script}(s)$. 

  Finally, $E^0$ is an  (infinite) sequence of events, containing an
  infinite number of events of the form $\an{a,a,\trigger}$
  for every $a \in \bigcup_{p\in \websystem} I^p$.

  A \emph{run} of $\completewebsystem$ is a run of
  $\websystem$ initiated by $E^0$.
\end{definition}

\clearpage
\section{Message and Data
  Formats}\label{app:message-data-formats}

We now provide some more details about data and message
formats that are needed for the formal treatment of the web
model and the analysis  presented in the following.

\subsection{URLs}\label{app:urls}

\begin{definition}\label{def:url}
  A \emph{URL} is a term of the form
  $$\an{\tUrl, \mi{protocol}, \mi{host}, \mi{path},
    \mi{parameters}, \mi{fragment}}$$ with $\mi{protocol}$
  $\in \{\http, \https\}$
  (for \textbf{p}lain (HTTP) and \textbf{s}ecure (HTTPS)),
  $\mi{host} \in \dns$,
  $\mi{path} \in \mathbb{S}$,
  $\mi{parameters} \in \dict{\mathbb{S}}{\terms}$,
  and $\mi{fragment} \in \terms$.
  The set of all valid URLs is $\urls$.
\end{definition}

The $\mi{fragment}$ part of a URL can be omitted when
writing the URL. Its value is then defined to be $\bot$.

\begin{example} \label{ex:url-pointers}
  For the URL $u = \an{\tUrl, a, b, c, d}$, $\comp{u}{protocol} =
  a$. If, in the algorithm described later, we say $\comp{u}{path} :=
  e$ then $u = \an{\tUrl, a, b, c, e}$ afterwards. 
\end{example}

\subsection{Origins}\label{app:origins}
\begin{definition} An \emph{origin} is a term of the form
  $\an{\mi{host}, \mi{protocol}}$ with $\mi{host} \in
  \dns$ and $\mi{protocol} \in \{\http, \https\}$. We write
  $\origins$ for the set of all origins.  
\end{definition}

\begin{example}
  For example, $\an{\str{FOO}, \https}$ is the HTTPS origin
  for the domain $\str{FOO}$, while $\an{\str{BAR}, \http}$
  is the HTTP origin for the domain $\str{BAR}$.
\end{example}
\subsection{Cookies}\label{app:cookies}

\begin{definition} A \emph{cookie} is a term of the form
  $\an{\mi{name}, \mi{content}}$ where $\mi{name} \in
  \terms$, and $\mi{content}$ is a term of the form
  $\an{\mi{value}, \mi{secure}, \mi{session},
    \mi{httpOnly}}$ where $\mi{value} \in \terms$,
  $\mi{secure}$, $\mi{session}$, $\mi{httpOnly} \in
  \{\True, \bot\}$. We write $\cookies$ for the set of all
  cookies and $\cookies^\nu$ for the set of all cookies
  where names and values are defined over $\terms(V)$.
\end{definition}

If the $\mi{secure}$ attribute of a cookie is set, the
browser will not transfer this cookie over unencrypted HTTP
connections. If the $\mi{session}$ flag is set, this cookie
will be deleted as soon as the browser is closed. The
$\mi{httpOnly}$ attribute controls whether JavaScript has
access to this cookie.

Note that cookies of the form described here are only
contained in HTTP(S) requests. In responses, only the
components $\mi{name}$ and $\mi{value}$ are transferred as
a pairing of the form $\an{\mi{name}, \mi{value}}$.

\subsection{HTTP Messages}\label{app:http-messages-full}
\begin{definition}
  An \emph{HTTP request} is a term of the form shown in
  (\ref{eq:default-http-request}). An \emph{HTTP response}
  is a term of the form shown in
  (\ref{eq:default-http-response}).
  \begin{align}
    \label{eq:default-http-request}
    & \hreq{ nonce=\mi{nonce}, method=\mi{method},
      xhost=\mi{host}, xpath=\mi{path},
      parameters=\mi{parameters}, headers=\mi{headers},
      xbody=\mi{body}
    } \\
    \label{eq:default-http-response}
    & \hresp{ nonce=\mi{nonce}, status=\mi{status},
      headers=\mi{headers}, xbody=\mi{body} }
  \end{align}
  The components are defined as follows:
  \begin{itemize}
  \item $\mi{nonce} \in \nonces$ serves to map each
    response to the corresponding request 
  \item $\mi{method} \in \methods$ is one of the HTTP
    methods.
  \item $\mi{host} \in \dns$ is the host name in the HOST
    header of HTTP/1.1.
  \item $\mi{path} \in \mathbb{S}$ is a string indicating
    the requested resource at the server side
  \item $\mi{status} \in \mathbb{S}$ is the HTTP status
    code (i.e., a number between 100 and 505, as defined by
    the HTTP standard)
  \item $\mi{parameters} \in
    \dict{\mathbb{S}}{\terms}$ contains URL parameters
  \item $\mi{headers} \in \dict{\mathbb{S}}{\terms}$,
    containing request/response headers. The dictionary
    elements are terms of one of the following forms: 
    \begin{itemize}
    \item $\an{\str{Origin}, o}$ where $o$ is an origin,
    \item $\an{\str{Set{\mhyphen}Cookie}, c}$ where $c$ is
      a sequence of cookies,
    \item $\an{\str{Cookie}, c}$ where $c \in,
      \dict{\mathbb{S}}{\terms}$ (note that in this header,
      only names and values of cookies are transferred),
    \item $\an{\str{Location}, l}$ where $l \in \urls$,
    \item $\an{\str{Referer}, r}$ where $r \in \urls$,
    \item $\an{\str{Strict{\mhyphen}Transport{\mhyphen}Security},\True}$,
    \item $\an{\str{Authorization}, \an{\mi{username},\mi{password}}}$ where $\mi{username}$, $\mi{password} \in \mathbb{S}$,
    \item $\an{\str{ReferrerPolicy}, p}$ where $p \in \{\str{noreferrer}, \str{origin}\}$
    \end{itemize}
  \item $\mi{body} \in \terms$ in requests and responses. 
  \end{itemize}
  We write $\httprequests$/$\httpresponses$ for the set of
  all HTTP requests or responses, respectively.
\end{definition}

\begin{example}[HTTP Request and Response]
  \begin{align}
    \label{eq:ex-request}
    \nonumber \mi{r} := & \langle
                   \cHttpReq,
                   n_1,
                   \mPost,
                   \str{example.com},
                   \str{/show},
                   \an{\an{\str{index,1}}},\\ & \quad
                   [\str{Origin}: \an{\str{example.com, \https}}],
                   \an{\str{foo}, \str{bar}}
                \rangle \\
    \label{eq:ex-response} \mi{s} := & \hresp{ nonce=n_1,
      status=200,
      headers=\an{\an{\str{Set{\mhyphen}Cookie},\an{\an{\str{SID},\an{n_2,\bot,\bot,\True}}}}},
      xbody=\an{\str{somescript},x}}
  \end{align}
  \noindent
  An HTTP $\mGet$ request for the URL
  \url{http://example.com/show?index=1} is shown in
  (\ref{eq:ex-request}), with an Origin header and a body
  that contains $\an{\str{foo},\str{bar}}$. A possible
  response is shown in (\ref{eq:ex-response}), which
  contains an httpOnly cookie with name $\str{SID}$ and
  value $n_2$ as well as the string representation
  $\str{somescript}$ of the script
  $\mathsf{script}^{-1}(\str{somescript})$ (which should be
  an element of $\scriptset$) and its initial state
  $x$.
\end{example}

\subsubsection{Encrypted HTTP
  Messages} \label{app:http-messages-encrypted-full}
For HTTPS, requests are encrypted using the public key of
the server.  Such a request contains an (ephemeral)
symmetric key chosen by the client that issued the
request. The server is supported to encrypt the response
using the symmetric key.

\begin{definition} An \emph{encrypted HTTP request} is of
  the form $\enc{\an{m, k'}}{k}$, where $k \in terms$, $k' \in
  \nonces$, and $m \in \httprequests$. The corresponding
  \emph{encrypted HTTP response} would be of the form
  $\encs{m'}{k'}$, where $m' \in \httpresponses$. We call
  the sets of all encrypted HTTP requests and responses
  $\httpsrequests$ or $\httpsresponses$, respectively.
\end{definition}

We say that an HTTP(S) response matches or corresponds to
an HTTP(S) request if both terms contain the same nonce.

\begin{example}
  \begin{align}
    \label{eq:ex-enc-request} \ehreqWithVariable{r}{k'}{\pub(k_\text{example.com})} \\
    \label{eq:ex-enc-response} \ehrespWithVariable{s}{k'}
  \end{align} The term (\ref{eq:ex-enc-request}) shows an
  encrypted request (with $r$ as in
  (\ref{eq:ex-request})). It is encrypted using the public
  key $\pub(k_\text{example.com})$.  The term
  (\ref{eq:ex-enc-response}) is a response (with $s$ as in
  (\ref{eq:ex-response})). It is encrypted symmetrically
  using the (symmetric) key $k'$ that was sent in the
  request (\ref{eq:ex-enc-request}).
\end{example}

\subsection{DNS Messages}\label{app:dns-messages}
\begin{definition} A \emph{DNS request} is a term of the form
$\an{\cDNSresolve, \mi{domain}, \mi{n}}$ where $\mi{domain}$ $\in
\dns$, $\mi{n} \in \nonces$. We call the set of all DNS requests
$\dnsrequests$.
\end{definition}

\begin{definition} A \emph{DNS response} is a term of the form
$\an{\cDNSresolved, \mi{domain}, \mi{result}, \mi{n}}$ with $\mi{domain}$ $\in
\dns$, $\mi{result} \in
\addresses$, $\mi{n} \in \nonces$. We call the set of all DNS
responses $\dnsresponses$.
\end{definition}

DNS servers are supposed to include the nonce they received
in a DNS request in the DNS response that they send back so
that the party which issued the request can match it with
the request.

\subsection{DNS Servers}\label{app:DNSservers}

Here, we consider a flat DNS model in which DNS queries are
answered directly by one DNS server and always with the
same address for a domain. A full (hierarchical) DNS system
with recursive DNS resolution, DNS caches, etc.~could also
be modeled to cover certain attacks on the DNS system
itself.

\begin{definition}\label{def:dns-server}
  A \emph{DNS server} $d$ (in a flat DNS model) is modeled
  in a straightforward way as an atomic DY process
  $(I^d, \{s^d_0\}, R^d, s^d_0)$. It has a finite set of
  addresses $I^d$ and its initial (and only) state $s^d_0$
  encodes a mapping from domain names to addresses of the
  form
$$s^d_0=\langle\an{\str{domain}_1,a_1},\an{\str{domain}_2, a_2}, \ldots\rangle \ .$$ DNS
queries are answered according to this table (otherwise
ignored).
\end{definition}

\clearpage
\section{Detailed Description of the Browser Model}
\label{app:deta-descr-brows}
Following the informal description of the browser model in
Section~\ref{sec:fks-web-model}, we now present a formal
model. We start by introducing some notation and
terminology. 

\subsection{Notation and Terminology (Web Browser State)}

Before we can define the state of a web browser, we first
have to define windows and documents. 

\begin{sloppypar}
  \begin{definition} A \emph{window} is a term of the form
    $w = \an{\mi{nonce}, \mi{documents}, \mi{opener}}$ with
    $\mi{nonce} \in \nonces$,
    $\mi{documents} \subsetPairing \documents$ (defined
    below), $\mi{opener} \in \nonces \cup \{\bot\}$ where
    $\comp{d}{active} = \True$ for exactly one
    $d \inPairing \mi{documents}$ if $\mi{documents}$ is
    not empty (we then call $d$ the \emph{active document
      of $w$}). We write $\windows$ for the set of all
    windows. We write $\comp{w}{activedocument}$ to denote
    the active document inside window $w$ if it exists and
    $\an{}$ else.
  \end{definition}
\end{sloppypar}
\noindent We will refer to the window nonce as \emph{(window)
  reference}.

The documents contained in a window term to the left of the
active document are the previously viewed documents
(available to the user via the ``back'' button) and the
documents in the window term to the right of the currently
active document are documents available via the ``forward''
button.

A window $a$ may have opened a top-level window $b$ (i.e.,
a window term which is not a subterm of a document
term). In this case, the \emph{opener} part of the term $b$
is the nonce of $a$, i.e., $\comp{b}{opener} =
\comp{a}{nonce}$.

\begin{sloppypar}
  \begin{definition} A \emph{document} $d$ is a term of the
    form
    \begin{align*}
      \an{\mi{nonce}, \mi{location}, \mi{headers}, \mi{referrer}, \mi{script},
      \mi{scriptstate},\mi{scriptinputs}, \mi{subwindows},
      \mi{active}}  
    \end{align*}
    where $\mi{nonce} \in \nonces$,
    $\mi{location} \in \urls$,
    $\mi{headers} \in \dict{\mathbb{S}}{\terms}$,
    $\mi{referrer} \in \urls \cup \{\bot\}$,
    $\mi{script} \in \terms$,
    $\mi{scriptstate} \in \terms$,
    $\mi{scriptinputs} \in \terms$,
    $\mi{subwindows} \subsetPairing \windows$,
    $\mi{active} \in \{\True, \bot\}$. A \emph{limited
      document} is a term of the form
    $\an{\mi{nonce}, \mi{subwindows}}$ with $\mi{nonce}$,
    $\mi{subwindows}$ as above. A window
    $w \inPairing \mi{subwindows}$ is called a
    \emph{subwindow} (of $d$). We write $\documents$ for
    the set of all documents. For a document term $d$ we
    write $d.\str{origin}$ to denote the origin of the
    document, i.e., the term
    $\an{d.\str{location}.\str{host},
      d.\str{location}.\str{protocol}} \in \origins$.
  \end{definition}%
\end{sloppypar}%

\noindent We will refer to the document nonce as \emph{(document)
  reference}.

\begin{definition} For two window terms $w$
  and $w'$
  we write $w \windowChildOf w'$
  if
  $w \inPairing
  \comp{\comp{w'}{activedocument}}{subwindows}$. We write
  $\windowChildOfX$ for the transitive closure.
\end{definition}

\subsection{Web Browser State}
\label{sec:web-browser-state}

We can now define the set of states of web browsers. Note
that we use the dictionary notation that we introduced in
Definition~\ref{def:dictionaries}.

\begin{definition} The \emph{set of states
    $Z_\text{webbrowser}$
    of a web browser atomic Dolev-Yao process} consists of
  the terms of the form
  \begin{align*} \langle\mi{windows}, \mi{ids},
    \mi{secrets}, \mi{cookies}, \mi{localStorage},
    \mi{sessionStorage}, \mi{keyMapping},& \\\mi{sts},
    \mi{DNSaddress}, \mi{pendingDNS},
    \mi{pendingRequests}, \mi{isCorrupted}&\rangle
  \end{align*} where
  \begin{itemize}
  \item $\mi{windows} \subsetPairing \windows$,
  \item $\mi{ids} \subsetPairing \terms$,
  \item $\mi{secrets} \in \dict{\origins}{\terms}$,
  \item $\mi{cookies}$ is a dictionary over $\dns$ and
    sequences of $\cookies$, 
  \item $\mi{localStorage} \in \dict{\origins}{\terms}$,
  \item $\mi{sessionStorage} \in \dict{\mi{OR}}{\terms}$ for $\mi{OR} := \left\{\an{o,r}
    \middle|\, o \in \origins,\, r \in \nonces\right\}$,
  \item $\mi{keyMapping} \in \dict{\dns}{\terms}$,
  \item $\mi{sts} \subsetPairing \dns$,
  \item $\mi{DNSaddress} \in \addresses$,
  \item $\mi{pendingDNS} \in \dict{\nonces}{\terms}$,
  \item $\mi{pendingRequests} \in$ $\terms$,
  \item and $\mi{isCorrupted} \in \{\bot, \fullcorrupt,$ $
    \closecorrupt\}$.
  \end{itemize} 
\end{definition}

\subsection{Description of the Web Browser
  Relation}\label{app:descr-web-brows}

We will now define the relation $R_{\text{webbrowser}}$
of a standard HTTP browser. We first introduce some
notations and then describe the functions that are used for
defining the browser main algorithm. We then define the
browser relation.

\subsubsection{Helper Functions}
In the following description of the web browser relation
$R_{\text{webbrowser}}$
we use the helper functions $\mathsf{Subwindows}$,
$\mathsf{Docs}$,
$\mathsf{Clean}$,
$\mathsf{CookieMerge}$ and $\mathsf{AddCookie}$.

\paragraph{Subwindows}
Given a browser state $s$,
$\mathsf{Subwindows}(s)$
denotes the set of all pointers\footnote{Recall the
  definition of a pointer in Definition~\ref{def:pointer}.}
to windows in the window list $\comp{s}{windows}$,
their active documents, and (recursively) the subwindows of
these documents. We exclude subwindows of inactive
documents and their subwindows. With $\mathsf{Docs}(s)$
we denote the set of pointers to all active documents in
the set of windows referenced by $\mathsf{Subwindows}(s)$.
\begin{definition} 
  For a browser state $s$ we denote by
  $\mathsf{Subwindows}(s)$ the minimal set of
  pointers that satisfies the
  following conditions: (1) For all windows $w \inPairing
  \comp{s}{windows}$ there is a $\ptr{p} \in
  \mathsf{Subwindows}(s)$ such that $\compn{s}{\ptr{p}} =
  w$. (2) For all $\ptr{p} \in \mathsf{Subwindows}(s)$, the
  active document $d$ of the window $\compn{s}{\ptr{p}}$
  and every subwindow $w$ of $d$ there is a pointer
  $\ptr{p'} \in \mathsf{Subwindows}(s)$ such that
  $\compn{s}{\ptr{p'}} = w$.

  Given a browser state $s$, the set $\mathsf{Docs}(s)$ of
  pointers to active documents is the minimal set such that
  for every $\ptr{p} \in \mathsf{Subwindows}(s)$, there is
  a pointer $\ptr{p'} \in \mathsf{Docs}(s)$ with
  $\compn{s}{\ptr{p'}} =
  \comp{\compn{s}{\ptr{p}}}{activedocument}$.
\end{definition}

By $\mathsf{Subwindows}^+(s)$ and $\mathsf{Docs}^+(s)$ we
denote the respective sets that also include the inactive
documents and their subwindows.

\paragraph{Clean}
The function $\mathsf{Clean}$ will be used to determine
which information about windows and documents the script
running in the document $d$ has access to.
\begin{definition} Let $s$
  be a browser state and $d$
  a document. By $\mathsf{Clean}(s, d)$
  we denote the term that equals $\comp{s}{windows}$
  but with (1) all inactive documents removed (including
  their subwindows etc.), (2) all subterms that represent
  non-same-origin documents w.r.t.~$d$
  replaced by a limited document $d'$
  with the same nonce and the same subwindow list, and (3)
  the values of the subterms $\str{headers}$
  for all documents set to $\an{}$.
  (Note that non-same-origin documents on all levels are
  replaced by their corresponding limited document.)
\end{definition}

\paragraph{CookieMerge}
The function $\mathsf{CookieMerge}$
merges two sequences of cookies together: When used in the
browser, $\mi{oldcookies}$
is the sequence of existing cookies for some origin,
$\mi{newcookies}$
is a sequence of new cookies that was output by some
script. The sequences are merged into a set of cookies
using an algorithm that is based on the \emph{Storage
  Mechanism} algorithm described in RFC6265.
\begin{definition} \label{def:cookiemerge} For a sequence
  of cookies (with pairwise different names)
  $\mi{oldcookies}$
  and a sequence of cookies $\mi{newcookies}$,
  the set
  $\mathsf{CookieMerge}(\mi{oldcookies}, \mi{newcookies})$
  is defined by the following algorithm: From
  $\mi{newcookies}$
  remove all cookies $c$
  that have $c.\str{content}.\str{httpOnly} \equiv \True$.
  For any $c$,
  $c' \inPairing \mi{newcookies}$,
  $\comp{c}{name} \equiv \comp{c'}{name}$,
  remove the cookie that appears left of the other in
  $\mi{newcookies}$.
  Let $m$
  be the set of cookies that have a name that either
  appears in $\mi{oldcookies}$
  or in $\mi{newcookies}$,
  but not in both. For all pairs of cookies
  $(c_\text{old}, c_\text{new})$
  with $c_\text{old} \inPairing \mi{oldcookies}$,
  $c_\text{new} \inPairing \mi{newcookies}$,
  $\comp{c_\text{old}}{name} \equiv
  \comp{c_\text{new}}{name}$, add $c_\text{new}$
  to $m$
  if
  $\comp{\comp{c_\text{old}}{content}}{httpOnly} \equiv
  \bot$ and add $c_\text{old}$
  to $m$
  otherwise. The result of
  $\mathsf{CookieMerge}(\mi{oldcookies}, \mi{newcookies})$
  is $m$.
\end{definition}

\paragraph{AddCookie}
The function $\mathsf{AddCookie}$
adds a cookie $c$
received in an HTTP response to the sequence of cookies
contained in the sequence $\mi{oldcookies}$.
It is again based on the algorithm described in RFC6265 but
simplified for the use in the browser model.
\begin{definition} \label{def:addcookie} For a sequence of cookies (with pairwise different
  names) $\mi{oldcookies}$ and a cookie $c$, the sequence
  $\mathsf{AddCookie}(\mi{oldcookies}, c)$ is defined by the
  following algorithm: Let $m := \mi{oldcookies}$. Remove
  any $c'$ from $m$ that has $\comp{c}{name} \equiv
  \comp{c'}{name}$. Append $c$ to $m$ and return $m$.
\end{definition}

\paragraph{NavigableWindows}
The function $\mathsf{NavigableWindows}$ returns a set of
windows that a document is allowed to navigate. We closely
follow \cite{html5}, Section~5.1.4 for this definition.
\begin{definition} The set $\mathsf{NavigableWindows}(\ptr{w}, s')$
  is the set $\ptr{W} \subseteq
  \mathsf{Subwindows}(s')$ of pointers to windows that the
  active document in $\ptr{w}$ is allowed to navigate. The
  set $\ptr{W}$ is defined to be the minimal set such that
  for every $\ptr{w'}
  \in \mathsf{Subwindows}(s')$ the following is true: %
\begin{itemize}
\item If
  $\comp{\comp{\compn{s'}{\ptr{w}'}}{activedocument}}{origin}
  \equiv
  \comp{\comp{\compn{s'}{\ptr{w}}}{activedocument}}{origin}$
  (i.e., the active documents in $\ptr{w}$ and $\ptr{w'}$ are
  same-origin), then $\ptr{w'} \in \ptr{W}$, and
\item If ${\compn{s'}{\ptr{w}} \childof
    \compn{s'}{\ptr{w'}}}$ $\wedge$ $\nexists\, \ptr{w}''
  \in \mathsf{Subwindows}(s')$ with $\compn{s'}{\ptr{w}'}
  \childof \compn{s'}{\ptr{w}''}$ ($\ptr{w'}$ is a
  top-level window and $\ptr{w}$ is an ancestor window of
  $\ptr{w'}$), then $\ptr{w'} \in \ptr{W}$, and
\item If $\exists\, \ptr{p} \in \mathsf{Subwindows}(s')$
  such that $\compn{s'}{\ptr{w}'} \windowChildOfX
  \compn{s'}{\ptr{p}}$ \\$\wedge$
  $\comp{\comp{\compn{s'}{\ptr{p}}}{activedocument}}{origin}
  =
  \comp{\comp{\compn{s'}{\ptr{w}}}{activedocument}}{origin}$
  ($\ptr{w'}$ is not a top-level window but there is an
  ancestor window $\ptr{p}$ of $\ptr{w'}$ with an active
  document that has the same origin as the active document
  in $\ptr{w}$), then $\ptr{w'} \in \ptr{W}$, and
\item If $\exists\, \ptr{p} \in \mathsf{Subwindows}(s')$ such
  that $\comp{\compn{s'}{\ptr{w'}}}{opener} =
  \comp{\compn{s'}{\ptr{p}}}{nonce}$ $\wedge$ $\ptr{p} \in
  \ptr{W}$ ($\ptr{w'}$ is a top-level window---it has an
  opener---and $\ptr{w}$ is allowed to navigate the opener
  window of $\ptr{w'}$, $\ptr{p}$), then $\ptr{w'} \in
  \ptr{W}$. 
\end{itemize}
\end{definition}

\begin{algorithm}[p]
\caption{\label{alg:getnavigablewindow} Web Browser Model: Determine window for navigation.}
\begin{algorithmic}[1]
  \Function{$\mathsf{GETNAVIGABLEWINDOW}$}{$\ptr{w}$, $\mi{window}$, $\mi{noreferrer}$, $s'$}
    \If{$\mi{window} \equiv \wBlank$} \Comment{Open a new window when $\wBlank$ is used}
      \If{$\mi{noreferrer} \equiv \True$}
        \Let{$w'$}{$\an{\nu_9, \an{}, \bot}$}
      \Else
        \Let{$w'$}{$\an{\nu_9, \an{}, \comp{\compn{s'}{\ptr{w}}}{nonce} }$}
      \EndIf
      \Append{$w'$}{$\comp{s'}{windows}$}  \breakalgohook{2}\textbf{and} let
      $\ptr{w}'$ be a pointer to this new element in $s'$
      \State \Return{$\ptr{w}'$}
    \EndIf
    \LetNDST{$\ptr{w}'$}{$\mathsf{NavigableWindows}(\ptr{w},
      s')$}{$\comp{\compn{s'}{\ptr{w}'}}{nonce} \equiv
      \mi{window}$\breakalgohook{1}}{\textbf{return} $\ptr{w}$} %
    \State \Return{$\ptr{w'}$}
  \EndFunction
\end{algorithmic} %
\end{algorithm}
\begin{algorithm}[tbp]
\caption{\label{alg:getwindow} Web Browser Model: Determine same-origin window.}
\begin{algorithmic}[1]
  \Function{$\mathsf{GETWINDOW}$}{$\ptr{w}$, $\mi{window}$, $s'$}
    \LetNDST{$\ptr{w}'$}{$\mathsf{Subwindows}(s')$}{$\comp{\compn{s'}{\ptr{w}'}}{nonce} \equiv \mi{window}$\breakalgohook{1}}{\textbf{return} $\ptr{w}$} %
    \If{
      $\comp{\comp{\compn{s'}{\ptr{w}'}}{activedocument}}{origin}
      \equiv
      \comp{\comp{\compn{s'}{\ptr{w}}}{activedocument}}{origin}$
    }
      \State \Return{$\ptr{w}'$}
    \EndIf
    \State \Return{$\ptr{w}$}
  \EndFunction
\end{algorithmic} %
\end{algorithm}
\begin{algorithm}[tbp]
\caption{\label{alg:cancelnav} Web Browser Model: Cancel pending requests for given window.}
\begin{algorithmic}[1]
  \Function{$\mathsf{CANCELNAV}$}{$n$, $s'$}
    \State \textbf{remove all} $\an{n, \mi{req}, \mi{key}, \mi{f}}$ \textbf{ from } $\comp{s'}{pendingRequests}$ \textbf{for any} $\mi{req}$, $\mi{key}$, $\mi{f}$
    \State \textbf{remove all} $\an{x, \an{n, \mi{message}, \mi{url}}}$ \textbf{ from } $\comp{s'}{pendingDNS}$\breakalgohook{1} \textbf{for any} $\mi{x}$, $\mi{message}$, $\mi{url}$
    \State \Return{$s'$}
  \EndFunction
\end{algorithmic} %
\end{algorithm}
\begin{algorithm}[tbp]
\caption{\label{alg:send} Web Browser Model: Prepare headers, do DNS resolution, save message. }
\begin{algorithmic}[1]
  \Function{$\mathsf{HTTP\_SEND}$}{$\mi{reference}$, $\mi{message}$, $\mi{url}$, $\mi{origin}$, $\mi{referrer}$, $\mi{referrerPolicy}$, $s'$}
    \If{$\comp{\mi{message}}{host} \inPairing \comp{s'}{sts}$}
      \Let{$\mi{url}.\str{protocol}$}{$\https$}
    \EndIf
    \Let{ $\mi{cookies}$}{$\langle\{\an{\comp{c}{name}, \comp{\comp{c}{content}}{value}} | c\inPairing \comp{s'}{cookies}\left[\comp{\mi{message}}{host}\right]$} \label{line:assemble-cookies-for-request} \breakalgohook{1} $\wedge \left(\comp{\comp{c}{content}}{secure} \implies \left(\mi{url}.\str{protocol} = \https\right)\right) \}\rangle$ \label{line:cookie-rules-http}
    \Let{$\comp{\mi{message}}{headers}[\str{Cookie}]$}{$\mi{cookies}$}
    \If{$\mi{origin} \not\equiv \bot$}
      \Let{$\comp{\mi{message}}{headers}[\str{Origin}]$}{$\mi{origin}$}
    \EndIf
    \If{$\mi{referrerPolicy} \equiv \str{noreferrer}$} 
      \Let{$\mi{referrer}$}{$\bot$}
    \EndIf
    \If{$\mi{referrer} \not\equiv \bot$}
      \If{$\mi{referrerPolicy} \equiv \str{origin}$} 
        \Let{$\mi{referrer}$}{$\an{\cUrl, \mi{referrer}.\str{protocol}, \mi{referrer}.\str{host}, \str{/}, \an{}, \bot}$} \Comment{Referrer stripped down to origin.}
      \EndIf
      \Let{$\mi{referrer}.\str{fragment}$}{$\bot$} \Comment{Browsers do not send fragment identifiers in the Referer header.}
      \Let{$\comp{\mi{message}}{headers}[\str{Referer}]$}{$\mi{referrer}$}
    \EndIf
    \Let{$\comp{s'}{pendingDNS}[\nu_8]$}{$\an{\mi{reference},
        \mi{message}, \mi{url}}$} \label{line:add-to-pendingdns}
    \State \textbf{stop} $\an{\an{\comp{s'}{DNSaddress},a,
    \an{\cDNSresolve, \mi{message}.\str{host}, \nu_8}}}$, $s'$
  \EndFunction
\end{algorithmic} %
\end{algorithm}

\subsubsection{Notations for Functions and Algorithms} We
use the following notations to describe the browser
algorithms:

\paragraph{Non-deterministic chosing and iteration} The notation
$\textbf{let}\ n \leftarrow N$
is used to describe that $n$
is chosen non-de\-ter\-mi\-nis\-tic\-ally from the set $N$.
We write $\textbf{for each}\ s \in M\ \textbf{do}$
to denote that the following commands (until \textbf{end
  for}) are repeated for every element in $M$,
where the variable $s$
is the current element. The order in which the elements are
processed is chosen non-deterministically. %
We write, for example,
\begin{algorithmic}
  \LetST{$x,y$}{$\an{\str{Constant},x,y} \equiv
    t$}{doSomethingElse}
\end{algorithmic} %
for some variables $x,y$,
a string $\str{Constant}$,
and some term $t$
to express that $x := \proj{2}{t}$,
and $y := \proj{3}{t}$
if $\str{Constant} \equiv \proj{1}{t}$
and if $|\an{\str{Constant},x,y}| = |t|$,
and that otherwise $x$
and $y$
are not set and doSomethingElse is executed.

\paragraph{Stop without output} We write
\textbf{stop} (without further parameters) to denote that
there is no output and no change in the state.

\paragraph{Placeholders} In several places throughout the
algorithms presented next we use placeholders to generate
``fresh'' nonces as described in our communication model
(see Definition~\ref{def:terms}).
Figure~\ref{fig:browser-placeholder-list} shows a list of
all placeholders used.

\begin{figure}[b]
  \centering
  \begin{tabular}{|@{\hspace{1ex}}l@{\hspace{1ex}}|@{\hspace{1ex}}l@{\hspace{1ex}}|}\hline 
    \hfill Placeholder\hfill  &\hfill  Usage\hfill  \\\hline\hline
    $\nu_1$ & Algorithm~\ref{alg:browsermain}, new window nonces  \\\hline
    $\nu_2$ & Algorithm~\ref{alg:browsermain}, new HTTP request nonce   \\\hline
    $\nu_3$ & Algorithm~\ref{alg:browsermain}, lookup key for pending HTTP requests entry  \\\hline
    $\nu_4$ & Algorithm~\ref{alg:runscript}, new HTTP request nonce (multiple lines)  \\\hline
    $\nu_5$ & Algorithm~\ref{alg:runscript}, new subwindow nonce  \\\hline
    $\nu_6$ & Algorithm~\ref{alg:processresponse}, new HTTP request nonce  \\\hline
    $\nu_7$ & Algorithm~\ref{alg:processresponse}, new document nonce   \\\hline
    $\nu_8$ & Algorithm~\ref{alg:send}, lookup key for pending DNS entry  \\\hline
    $\nu_9$ & Algorithm~\ref{alg:getnavigablewindow}, new window nonce  \\\hline
    $\nu_{10}, \dots$ & Algorithm~\ref{alg:runscript}, replacement for placeholders in script output   \\\hline

  \end{tabular}
  
  \caption{List of placeholders used in browser algorithms.}
  \label{fig:browser-placeholder-list}
\end{figure}

\begin{algorithm}[tbp]
\caption{\label{alg:navback} Web Browser Model: Navigate a window backward. }
\begin{algorithmic}[1]
  \Function{$\mathsf{NAVBACK}$}{$\ptr{w}$, $s'$}
      \If{$\exists\, \ptr{j} \in
        \mathbb{N}, \ptr{j} > 1$ \textbf{such that}
        $\comp{\compn{\comp{\compn{s'}{\ptr{w'}}}{documents}}{\ptr{j}}}{active}
        \equiv \True$} %
        \Let{$\comp{\compn{\comp{\compn{s'}{\ptr{w'}}}{documents}}{\ptr{j}}}{active}$}{$\bot$}
        \Let{$\comp{\compn{\comp{\compn{s'}{\ptr{w'}}}{documents}}{(\ptr{j}-1)}}{active}$}{$\True$}
        \Let{$s'$}{$\mathsf{CANCELNAV}(\comp{\compn{s'}{\ptr{w}'}}{nonce},
        s')$}
        \EndIf
  \EndFunction
\end{algorithmic} %
\end{algorithm}
\begin{algorithm}[tbp]
\caption{\label{alg:navforward} Web Browser Model: Navigate a window forward. }
\begin{algorithmic}[1]
  \Function{$\mathsf{NAVFORWARD}$}{$\ptr{w}$, $s'$}
        \If{$\exists\, \ptr{j} \in \mathbb{N} $ \textbf{such that} $\comp{\compn{\comp{\compn{s'}{\ptr{w'}}}{documents}}{\ptr{j}}}{active} \equiv \True$ \breakalgohook{3} $\wedge$  $\compn{\comp{\compn{s'}{\ptr{w'}}}{documents}}{(\ptr{j}+1)} \in \mathsf{Documents}$} %
          \Let{$\comp{\compn{\comp{\compn{s'}{\ptr{w'}}}{documents}}{\ptr{j}}}{active}$}{$\bot$}
          \Let{$\comp{\compn{\comp{\compn{s'}{\ptr{w'}}}{documents}}{(\ptr{j}+1)}}{active}$}{$\True$}
          \Let{$s'$}{$\mathsf{CANCELNAV}(\comp{\compn{s'}{\ptr{w}'}}{nonce}, s')$}
        \EndIf
  \EndFunction
\end{algorithmic} %
\end{algorithm}
\begin{algorithm}[tp]
\caption{\label{alg:runscript} Web Browser Model: Execute a script.}
\begin{algorithmic}[1]
  \Function{$\mathsf{RUNSCRIPT}$}{$\ptr{w}$, $\ptr{d}$, $s'$}
    \Let{$\mi{tree}$}{$\mathsf{Clean}(s', \compn{s'}{\ptr{d}})$} \label{line:clean-tree}

    \Let{$\mi{cookies}$}{$\langle\{\an{\comp{c}{name}, \comp{\comp{c}{content}}{value}} | c \inPairing \comp{s'}{cookies}\left[  \comp{\comp{\compn{s'}{\ptr{d}}}{origin}}{host}  \right]$
     \breakalgohook{1} $\wedge\,\comp{\comp{c}{content}}{httpOnly} = \bot$ \breakalgohook{1} $\wedge\,\left(\comp{\comp{c}{content}}{secure} \implies \left(\comp{\comp{\compn{s'}{\ptr{d}}}{origin}}{protocol} \equiv \https\right)\right) \}\rangle$} \label{line:assemble-cookies-for-script}
    \LetND{$\mi{tlw}$}{$\comp{s'}{windows}$ \textbf{such that} $\mi{tlw}$ is the top-level window containing $\ptr{d}$} 
    \Let{$\mi{sessionStorage}$}{$\comp{s'}{sessionStorage}\left[\an{\comp{\compn{s'}{\ptr{d}}}{origin}, \comp{\mi{tlw}}{nonce}}\right]$} %
    \Let{$\mi{localStorage}$}{$\comp{s'}{localStorage}\left[\comp{\compn{s'}{\ptr{d}}}{origin}\right]$}
    \Let{$\mi{secrets}$}{$\comp{s'}{secrets}\left[\comp{\compn{s'}{\ptr{d}}}{origin}\right]$} \label{line:browser-secrets}
    \LetND{$R$}{$\mathsf{script}^{-1}(\comp{\compn{s'}{\ptr{d}}}{script})$} %
    \Let{$\mi{in}$}{$\langle\mi{tree}$, $\comp{\compn{s'}{\ptr{d}}}{nonce}, \comp{\compn{s'}{\ptr{d}}}{scriptstate}$, $\comp{\compn{s'}{\ptr{d}}}{scriptinputs}$, $\mi{cookies},$ \breakalgohook{1}  $\mi{localStorage}$, $\mi{sessionStorage}$, $\comp{s'}{ids}$, $\mi{secrets}\rangle$}\label{line:browser-scriptinputs}
    \LetND{$\mi{state}'$}{$\terms(V)$, \breakalgohook{1}
      $\mi{cookies}' \gets \mathsf{Cookies}^\nu$, \breakalgohook{1}
      $\mi{localStorage}' \gets \terms(V)$,\breakalgohook{1}
      $\mi{sessionStorage}' \gets \terms(V)$,\breakalgohook{1}
      $\mi{command} \gets \terms(V)$, \breakalgohook{1} 
      $\mi{out}^\lambda := \an{\mi{state}', \mi{cookies}', \mi{localStorage}',$ $\mi{sessionStorage}', \mi{command}}$
      \breakalgohook{1} \textbf{such that} $(\mi{in}, \mi{out}^\lambda) \in R$}  \label{line:trigger-script} 
    \Let{$\mi{out}$}{$\mi{out}^\lambda[\nu_{10}/\lambda_1, \nu_{11}/\lambda_2, \dots]$}

    \Let{$\comp{s'}{cookies}\left[\comp{\comp{\compn{s'}{\ptr{d}}}{origin}}{host}\right]$\breakalgohook{1}}{$\langle\mathsf{CookieMerge}(\comp{s'}{cookies}\left[\comp{\comp{\compn{s'}{\ptr{d}}}{origin}}{host}\right]$, $\mi{cookies}')\rangle$} \label{line:cookiemerge}
    \Let{$\comp{s'}{localStorage}\left[\comp{\compn{s'}{\ptr{d}}}{origin}\right]$}{$\mi{localStorage}'$}
    \Let{$\comp{s'}{sessionStorage}\left[\an{\comp{\compn{s'}{\ptr{d}}}{origin}, \comp{\mi{tlw}}{nonce}}\right]$}{$\mi{sessionStorage}'$}
    \Let{$\comp{\compn{s'}{\ptr{d}}}{scriptstate}$}{$state'$}
    \Switch{$\mi{command}$}
      \Case{$\an{\tHref, \mi{url},
          \mi{hrefwindow}, \mi{noreferrer}}$}
      \Let{$\ptr{w}'$}{$\mathsf{GETNAVIGABLEWINDOW}$($\ptr{w}$,
        $\mi{hrefwindow}$, $\mi{noreferrer}$, $s'$)} 
      \Let{$\mi{req}$}{$\hreq{ nonce=\nu_4, 
          method=\mGet, host=\comp{\mi{url}}{host},
          path=\comp{\mi{url}}{path},
          headers=\an{},
          parameters=\comp{\mi{url}}{parameters}, body=\an{}
        }$}
      \If{$\mi{noreferrer} \equiv \True$}
        \Let{$\mi{referrerPolicy}$}{$\str{noreferrer}$}
      \Else
        \Let{$\mi{referrerPolicy}$}{$\compn{s'}{\ptr{d}}.\str{headers}[\str{ReferrerPolicy}]$}
      \EndIf
      \Let{$s'$}{$\mathsf{CANCELNAV}(\comp{\compn{s'}{\ptr{w}'}}{nonce}, s')$}
      \CallFun{HTTP\_SEND}{$\comp{\compn{s'}{\ptr{w}'}}{nonce}$, $\mi{req}$, $\mi{url}$, $\bot$, $\mi{referrer}$, $\mi{referrerPolicy}$, $s'$} \label{line:send-href}
      \EndCase
      \Case{$\an{\tIframe, \mi{url}, \mi{window}}$}
        \Let{$\ptr{w}'$}{$\mathsf{GETWINDOW}(\ptr{w}, \mi{window}, s')$}
        \Let{$\mi{req}$}{$\hreq{
            nonce=\nu_4,
            method=\mGet,
            host=\comp{\mi{url}}{host},
            path=\comp{\mi{url}}{path},
            headers=\an{},
            parameters=\comp{\mi{url}}{parameters},
            body=\an{}
          }$}
        \Let{$\mi{referrer}$}{$s'.\ptr{w}'.\str{activedocument}.\str{location}$}
        \Let{$\mi{referrerPolicy}$}{$\compn{s'}{\ptr{d}}.\str{headers}[\str{ReferrerPolicy}]$}
        \Let{$w'$}{$\an{\nu_5, \an{}, \bot}$}
        \Let{$\comp{\comp{\compn{s'}{\ptr{w}'}}{activedocument}}{subwindows}$\breakalgohook{3}}{ $\comp{\comp{\compn{s'}{\ptr{w}'}}{activedocument}}{subwindows} \plusPairing w'$}
        \CallFun{HTTP\_SEND}{$\nu_5$, $\mi{req}$, $\mi{url}$, $\bot$, $\mi{referrer}$, $\mi{referrerPolicy}$, $s'$} \label{line:send-iframe}
      \EndCase
      \Case{$\an{\tForm, \mi{url}, \mi{method}, \mi{data}, \mi{hrefwindow}}$}
        \If{$\mi{method} \not\in \{\mGet, \mPost\}$} \footnote{The working draft for HTML5 allowed for DELETE and PUT methods in HTML5 forms. However, these have since been removed. See \url{http://www.w3.org/TR/2010/WD-html5-diff-20101019/\#changes-2010-06-24}.}
          \State \textbf{stop} 
        \EndIf
        \Let{$\ptr{w}'$}{$\mathsf{GETNAVIGABLEWINDOW}$($\ptr{w}$, $\mi{hrefwindow}$, $\bot$, $s'$)}
        \If{$\mi{method} = \mGet$}
          \Let{$\mi{body}$}{$\an{}$}
          \Let{$\mi{parameters}$}{$\mi{data}$}
          \Let{$\mi{origin}$}{$\bot$}
        \Else
          \Let{$\mi{body}$}{$\mi{data}$}
          \Let{$\mi{parameters}$}{$\comp{\mi{url}}{parameters}$}
          \Let{$\mi{origin}$}{$\comp{\compn{s'}{\ptr{d}}}{origin}$}
        \EndIf
        \Let{$\mi{req}$}{$\hreq{
            nonce=\nu_4,
            method=\mi{method},
            host=\comp{\mi{url}}{host},
            path=\comp{\mi{url}}{path},
            headers=\an{},
            parameters=\mi{parameters},
            xbody=\mi{body}
          }$}
        \Let{$\mi{referrer}$}{$\comp{\compn{s'}{\ptr{d}}}{location}$}
        \Let{$\mi{referrerPolicy}$}{$\compn{s'}{\ptr{d}}.\str{headers}[\str{ReferrerPolicy}]$}
        \Let{$s'$}{$\mathsf{CANCELNAV}(\comp{\compn{s'}{\ptr{w}'}}{nonce}, s')$}
        \CallFun{HTTP\_SEND}{$\comp{\compn{s'}{\ptr{w}'}}{nonce}$, $\mi{req}$, $\mi{url}$, $\mi{origin}$, $\mi{referrer}$, $\mi{referrerPolicy}$, $s'$} \label{line:send-form}
\algstore{myalg}
\end{algorithmic}
\end{algorithm}

\begin{algorithm}[tp]                    
\begin{algorithmic} [1]                   %
\algrestore{myalg}

      \EndCase
      \Case{$\an{\tSetScript, \mi{window}, \mi{script}}$}
        \Let{$\ptr{w}'$}{$\mathsf{GETWINDOW}(\ptr{w}, \mi{window}, s')$}
        \Let{$\comp{\comp{\compn{s'}{\ptr{w}'}}{activedocument}}{script}$}{$\mi{script}$}
        \State \textbf{stop} $\an{}$, $s'$
      \EndCase
      \Case{$\an{\tSetScriptState, \mi{window}, \mi{scriptstate}}$}
        \Let{$\ptr{w}'$}{$\mathsf{GETWINDOW}(\ptr{w}, \mi{window}, s')$}
        \Let{$\comp{\comp{\compn{s'}{\ptr{w}'}}{activedocument}}{scriptstate}$}{$\mi{scriptstate}$}
        \State \textbf{stop} $\an{}$, $s'$
      \EndCase
      \Case{$\an{\tXMLHTTPRequest, \mi{url}, \mi{method}, \mi{data}, \mi{xhrreference}}$}
        \If{$\mi{method} \in \{\mConnect, \mTrace, \mTrack\} \wedge \mi{xhrreference} \not\in \{\nonces, \bot\}$} 
          \State \textbf{stop} 
        \EndIf
        \If{$\comp{\mi{url}}{host} \not\equiv \comp{\comp{\compn{s'}{\ptr{d}}}{origin}}{host}$ \breakalgohook{3} $\vee$ $\mi{url} \not\equiv \comp{\comp{\compn{s'}{\ptr{d}}}{origin}}{protocol}$} 
          \State \textbf{stop} 
        \EndIf
        \If{$\mi{method} \in \{\mGet, \mHead\}$}
          \Let{$\mi{data}$}{$\an{}$}
          \Let{$\mi{origin}$}{$\bot$}
        \Else
          \Let{$\mi{origin}$}{$\comp{\compn{s'}{\ptr{d}}}{origin}$}
        \EndIf
        \Let{$\mi{req}$}{$\hreq{
            nonce=\nu_4,
            method=\mi{method},
            host=\comp{\mi{url}}{host},
            path=\comp{\mi{url}}{path},
            headers={},
            parameters=\comp{\mi{url}}{parameters},
            xbody=\mi{data}
          }$}
        \Let{$\mi{referrer}$}{$\comp{\compn{s'}{\ptr{d}}}{location}$}
        \Let{$\mi{referrerPolicy}$}{$\compn{s'}{\ptr{d}}.\str{headers}[\str{ReferrerPolicy}]$}
        \CallFun{HTTP\_SEND}{$\an{\comp{\compn{s'}{\ptr{d}}}{nonce}, \mi{xhrreference}}$, $\mi{req}$, $\mi{url}$, $\mi{origin}$, $\mi{referrer}$, $\mi{referrerPolicy}$, $s'$}\label{line:send-xhr}
      \EndCase
      \Case{$\an{\tBack, \mi{window}}$} \footnote{Note that
        navigating a window using the back/forward buttons
        does not trigger a reload of the affected
        documents. While real world browser may chose to
        refresh a document in this case, we assume that the
        complete state of a previously viewed document is
        restored. A reload can be triggered
        non-deterministically at any point (in the main algorithm).}
      \Let{$\ptr{w}'$}{$\mathsf{GETNAVIGABLEWINDOW}$($\ptr{w}$,
        $\mi{window}$, $\bot$, $s'$)} 
        \State $\mathsf{NAVBACK}$($\ptr{w}$, $s'$)
        \State \textbf{stop} $\an{}$, $s'$
      \EndCase
      \Case{$\an{\tForward, \mi{window}}$}
        \Let{$\ptr{w}'$}{$\mathsf{GETNAVIGABLEWINDOW}$($\ptr{w}$, $\mi{window}$, $\bot$, $s'$)}
        \State $\mathsf{NAVFORWARD}$($\ptr{w}$, $s'$)
        \State \textbf{stop} $\an{}$, $s'$
      \EndCase
      \Case{$\an{\tClose, \mi{window}}$}
        \Let{$\ptr{w}'$}{$\mathsf{GETNAVIGABLEWINDOW}$($\ptr{w}$, $\mi{window}$, $\bot$, $s'$)}
        \State \textbf{remove} $\compn{s'}{\ptr{w'}}$ from the sequence containing it 
        \State \textbf{stop} $\an{}$, $s'$
      \EndCase

      \Case{$\an{\tPostMessage, \mi{window}, \mi{message}, \mi{origin}}$}
        \LetND{$\ptr{w}'$}{$\mathsf{Subwindows}(s')$ \textbf{such that} $\comp{\compn{s'}{\ptr{w}'}}{nonce} \equiv \mi{window}$} %
        \If{$\exists \ptr{j} \in \mathbb{N}$ \textbf{such that} $\comp{\compn{\comp{\compn{s'}{\ptr{w'}}}{documents}}{\ptr{j}}}{active} \equiv \True$ \breakalgohook{3} $\wedge  (\mi{origin} \not\equiv \bot \implies \comp{\compn{\comp{\compn{s'}{\ptr{w'}}}{documents}}{\ptr{j}}}{origin} \equiv \mi{origin})$}    \label{line:append-pm-to-scriptinputs-condition} %
        \Let{$\comp{\compn{\comp{\compn{s'}{\ptr{w'}}}{documents}}{\ptr{j}}}{scriptinputs}$\breakalgohook{4}}{ $\comp{\compn{\comp{\compn{s'}{\ptr{w'}}}{documents}}{\ptr{j}}}{scriptinputs}$ \breakalgohook{4} $\plusPairing$
         $\an{\tPostMessage, \comp{\compn{s'}{\ptr{w}}}{nonce}, \comp{\compn{s'}{\ptr{d}}}{origin}, \mi{message}}$} \label{line:append-pm-to-scriptinputs}
        \EndIf
        \State \textbf{stop} $\an{}$, $s'$
      \EndCase
      \Case{else}
        \State \textbf{stop} 
      \EndCase
    \EndSwitch
  \EndFunction
\end{algorithmic} %
\end{algorithm}

\begin{algorithm}[tp]
\caption{\label{alg:processresponse} Web Browser Model: Process an HTTP response.}
\begin{algorithmic}[1]
\Function{$\mathsf{PROCESSRESPONSE}$}{$\mi{response}$, $\mi{reference}$, $\mi{request}$, $\mi{requestUrl}$, $s'$}
  \If{$\mathtt{Set{\mhyphen}Cookie} \in
    \comp{\mi{response}}{headers}$}
    \For{\textbf{each} $c \inPairing \comp{\mi{response}}{headers}\left[\mathtt{Set{\mhyphen}Cookie}\right]$, $c \in \mathsf{Cookies}$}
      \Let{$\comp{s'}{cookies}\left[\mi{request}.\str{host}\right]$\breakalgohook{3}}{$\mathsf{AddCookie}(\comp{s'}{cookies}\left[\mi{request}.\str{host}\right], c)$} \label{line:set-cookie}
    \EndFor
  \EndIf  
  \If{$\mathtt{Strict{\mhyphen}Transport{\mhyphen}Security} \in \comp{\mi{response}}{headers}$ $\wedge$ $\mi{requestUrl}.\str{protocol} \equiv \https$}
    \Append{$\comp{\mi{request}}{host}$}{$\comp{s'}{sts}$}
  \EndIf
  \If{$\str{Referer} \in \comp{request}{headers}$} 
    \Let{$\mi{referrer}$}{$\comp{request}{headers}[\str{Referer}]$}
  \Else
    \Let{$\mi{referrer}$}{$\bot$}
  \EndIf
  \If{$\mathtt{Location} \in \comp{\mi{response}}{headers} \wedge \comp{\mi{response}}{status} \in \{303, 307\}$} \label{line:location-header} 
    \Let{$\mi{url}$}{$\comp{\mi{response}}{headers}\left[\mathtt{Location}\right]$}
    \If{$\mi{url}.\str{fragment} \equiv \bot$}
      \Let{$\mi{url}.\str{fragment}$}{$\mi{requestUrl}.\str{fragment}$}
    \EndIf
    \Let{$\mi{method}'$}{$\comp{\mi{request}}{method}$} 
    \Let{$\mi{body}'$}{$\comp{\mi{request}}{body}$} 
    \If{$\str{Origin} \in \comp{request}{headers}$}
      \Let{$\mi{origin}$}{$\an{\comp{request}{headers}[\str{Origin}], \an{\comp{request}{host}, \mi{url}.\str{protocol}}}$}
    \Else
      \Let{$\mi{origin}$}{$\bot$}
    \EndIf
    \If{$\comp{\mi{response}}{status} \equiv 303 \wedge \comp{\mi{request}}{method} \not\in \{\mGet, \mHead\}$}
      \Let {$\mi{method}'$}{$\mGet$}
      \Let{$\mi{body}'$}{$\an{}$} \label{browser-remove-body}
    \EndIf
    \If{$\exists\, \ptr{w} \in \mathsf{Subwindows}(s')$ \textbf{such that} $\comp{\compn{s'}{\ptr{w}}}{nonce} \equiv \mi{reference}$} \Comment{Do not redirect XHRs.}
      \Let{$\mi{req}$}{$\hreq{
            nonce=\nu_6,
            method=\mi{method'},
            host=\comp{\mi{url}}{host},
            path=\comp{\mi{url}}{path},
            headers=\an{},
            parameters=\comp{\mi{url}}{parameters},
            xbody=\mi{body}'
          }$}
      \Let{$\mi{referrerPolicy}$}{$\mi{response}.\str{headers}[\str{ReferrerPolicy}]$}
      \CallFun{HTTP\_SEND}{$\mi{reference}$, $\mi{req}$, $\mi{url}$, $\mi{origin}$, $\mi{referrer}$, $\mi{referrerPolicy}$, $s'$}\label{line:send-redirect}
    \EndIf
  \EndIf
  \If{$\exists\, \ptr{w} \in \mathsf{Subwindows}(s')$ \textbf{such that} $\comp{\compn{s'}{\ptr{w}}}{nonce} \equiv \mi{reference}$} \Comment{normal response}
    \If{$\mi{response}.\str{body} \not\sim \an{*,*}$}
      \State \textbf{stop} $\{\}$, $s'$
    \EndIf
    \Let{$\mi{script}$}{$\proj{1}{\comp{\mi{response}}{body}}$}
    \Let{$\mi{scriptstate}$}{$\proj{2}{\comp{\mi{response}}{body}}$}
    \Let{$\mi{referrer}$}{$\mi{request}.\str{headers}[\str{Referer}]$}
    \Let{$d$}{$\an{\nu_7, \mi{requestUrl}, \mi{response}.\str{headers}, \mi{referrer}, \mi{script}, \mi{scriptstate}, \an{}, \an{}, \True}$} \label{line:take-script} \label{line:set-origin-of-document}
    \If{$\comp{\compn{s'}{\ptr{w}}}{documents} \equiv \an{}$}
      \Let{$\comp{\compn{s'}{\ptr{w}}}{documents}$}{$\an{d}$}
    \Else
      \LetND{$\ptr{i}$}{$\mathbb{N}$ \textbf{such that} $\comp{\compn{\comp{\compn{s'}{\ptr{w}}}{documents}}{\ptr{i}}}{active} \equiv \True$} %
      \Let{$\comp{\compn{\comp{\compn{s'}{\ptr{w}}}{documents}}{\ptr{i}}}{active}$}{$\bot$}
      \State \textbf{remove} $\compn{\comp{\compn{s'}{\ptr{w}}}{documents}}{(\ptr{i}+1)}$ and all following documents \breakalgohook{3} from $\comp{\compn{s'}{\ptr{w}}}{documents}$
      \Append{$d$}{$\comp{\compn{s'}{\ptr{w}}}{documents}$}
    \EndIf
    \State \textbf{stop} $\{\}$, $s'$
  \ElsIf{$\exists\, \ptr{w} \in \mathsf{Subwindows}(s')$, $\ptr{d}$ \textbf{such that} $\comp{\compn{s'}{\ptr{d}}}{nonce} \equiv \proj{1}{\mi{reference}} $ \breakalgohook{1}  $\wedge$  $\compn{s'}{\ptr{d}} = \comp{\compn{s'}{\ptr{w}}}{activedocument}$} \label{line:process-xhr-response} \Comment{process XHR response}
    \Let{$\mi{headers}$}{$\mi{response}.\str{headers} - \str{Set\mhyphen{}Cookie}$}
    \Append{\breakalgo{3}$\an{\tXMLHTTPRequest, \mi{headers}, \comp{\mi{response}}{body}, \proj{2}{\mi{reference}}}$}{$\comp{\compn{s'}{\ptr{d}}}{scriptinputs}$}
  \EndIf
\EndFunction
\end{algorithmic} %
\end{algorithm}
\begin{algorithm}[tp]
\caption{\label{alg:browsermain} Web Browser Model: Main Algorithm}
\begin{algorithmic}[1]
\Statex[-1] \textbf{Input:} $\an{a,f,m},s$
  \Let{$s'$}{$s$}

  \If{$\comp{s}{isCorrupted} \not\equiv \bot$}
    \Let{$\comp{s'}{pendingRequests}$}{$\an{m, \comp{s}{pendingRequests}}$} \Comment{Collect incoming messages}
    \LetND{$m'$}{$d_{V}(s')$} %
    \LetND{$a'$}{$\addresses$} %
    \State \textbf{stop} $\an{\an{a',a,m'}}$, $s'$
  \EndIf
  \If{$m \equiv \trigger$} \Comment{A special trigger message. }
    \LetND{$\mi{switch}$}{$\{\str{script},\str{urlbar},\str{reload},\str{forward}, \str{back}\}$} \label{line:browser-switch}  %
    \LetNDST{$\ptr{w}$}{$\mathsf{Subwindows}(s')$}{$\comp{\compn{s'}{\ptr{w}}}{documents} \neq \an{}$\breakalgohook{2}}{\textbf{stop}}%
    \Comment{Pointer to some window.}
    \LetNDST{$\ptr{tlw}$}{$\mathbb{N}$}{$\comp{\compn{s'}{\ptr{tlw}}}{documents} \neq \an{}$\breakalgohook{2}}{\textbf{stop}}%
    \Comment{Pointer to some top-level window.}
    \If{$\mi{switch} \equiv \str{script}$} \Comment{Run some script.}
      \Let{$\ptr{d}$}{$\ptr{w} \plusPairing \str{activedocument}$} \label{line:browser-trigger-document}  
      \CallFun{RUNSCRIPT}{$\ptr{w}$, $\ptr{d}$, $s'$}
    \ElsIf{$\mi{switch} \equiv \str{urlbar}$} \Comment{Create some new request.}
      \LetND{$\mi{newwindow}$}{$\{\True, \bot \}$}
      \If{$\mi{newwindow} \equiv \True$} \Comment{Create a new window.}
        \Let{$\mi{windownonce}$}{$\nu_1$}
        \Let{$w'$}{$\an{\mi{windownonce}, \an{}, \bot}$}
        \Append{$w'$}{$\comp{s'}{windows}$}
      \Else \Comment{Use existing top-level window.}
        \Let{$\mi{windownonce}$}{$s'.\ptr{tlw}.nonce$}
      \EndIf
      \LetND{$\mi{protocol}$}{$\{\http, \https\}$} \label{line:browser-choose-url} %
      \LetND{$\mi{host}$}{$\dns$} %
      \LetND{$\mi{path}$}{$\mathbb{S}$} %
      \LetND{$\mi{fragment}$}{$\mathbb{S}$} %
      \LetND{$\mi{parameters}$}{$\dict{\mathbb{S}}{\mathbb{S}}$} %
      \Let{$\mi{url}$}{$\an{\cUrl, \mi{protocol}, \mi{host}, \mi{path}, \mi{parameters}, \mi{fragment}}$}
      \Let{$\mi{req}$}{$\hreq{
          nonce=\nu_2,
          method=\mGet,
          host=\mi{host},
          path=\mi{path},
          headers=\an{},
          parameters=\mi{parameters},
          body=\an{}
        }$}
      \CallFun{HTTP\_SEND}{$\mi{windownonce}$, $\mi{req}$, $\mi{url}$, $\bot$, $\bot$, $\bot$, $s'$}\label{line:send-random}
    \ElsIf{$\mi{switch} \equiv \str{reload}$} \Comment{Reload some document.}
      \LetNDST{$\ptr{w}$}{$\mathsf{Subwindows}(s')$}{$\comp{\compn{s'}{\ptr{w}}}{documents} \neq \an{}$\breakalgohook{2}}{\textbf{stop}} \label{line:browser-reload-window}%
      \Let{$\mi{url}$}{$s'.\ptr{w}.\str{activedocument}.\str{location}$}
      \Let{$\mi{req}$}{$\hreq{ nonce=\nu_2, 
          method=\mGet, host=\comp{\mi{url}}{host},
          path=\comp{\mi{url}}{path},
          headers=\an{},
          parameters=\comp{\mi{url}}{parameters}, body=\an{}
        }$}
      \Let{$\mi{referrer}$}{$s'.\ptr{w}.\str{activedocument}.\str{referrer}$}
      \Let{$s'$}{$\mathsf{CANCELNAV}(\comp{\compn{s'}{\ptr{w}}}{nonce}, s')$}
      \CallFun{HTTP\_SEND}{$\comp{\compn{s'}{\ptr{w}}}{nonce}$, $\mi{req}$, $\mi{url}$, $\bot$, $\mi{referrer}$, $\bot$, $s'$}
    \ElsIf{$\mi{switch} \equiv \str{forward}$}
      \State $\mathsf{NAVFORWARD}$($\ptr{w}$, $s'$)
    \ElsIf{$\mi{switch} \equiv \str{back}$}
      \State $\mathsf{NAVBACK}$($\ptr{w}$, $s'$)
    \EndIf
  \ElsIf{$m \equiv \fullcorrupt$} \Comment{Request to corrupt browser}
    \Let{$\comp{s'}{isCorrupted}$}{$\fullcorrupt$}
    \State \textbf{stop} $\an{}$, $s'$
  \ElsIf{$m \equiv \closecorrupt$} \Comment{Close the browser}
    \Let{$\comp{s'}{secrets}$}{$\an{}$}  
    \Let{$\comp{s'}{windows}$}{$\an{}$}
    \Let{$\comp{s'}{pendingDNS}$}{$\an{}$}
    \Let{$\comp{s'}{pendingRequests}$}{$\an{}$}
    \Let{$\comp{s'}{sessionStorage}$}{$\an{}$}
    \State \textbf{let} $\comp{s'}{cookies} \subsetPairing \cookies$ \textbf{such that} \breakalgohook{1} $(c \inPairing \comp{s'}{cookies}) {\iff} (c \inPairing \comp{s}{cookies} \wedge \comp{\comp{c}{content}}{session} \equiv \bot$)
    \Let{$\comp{s'}{isCorrupted}$}{$\closecorrupt$}
    \State \textbf{stop} $\an{}$, $s'$

\algstore{myalg}
\end{algorithmic}
\end{algorithm}

\begin{algorithm}                     
\begin{algorithmic} [1]                   %
\algrestore{myalg}

  \ElsIf{$\exists\, \an{\mi{reference}, \mi{request}, \mi{url}, \mi{key}, f}$
      $\inPairing \comp{s'}{pendingRequests}$ \breakalgohook{0}
      \textbf{such that} $\proj{1}{\decs{m}{\mi{key}}} \equiv \cHttpResp$ } %
    \Comment{Encrypted HTTP response}
    \Let{$m'$}{$\decs{m}{\mi{key}}$}
    \If{$\comp{m'}{nonce} \not\equiv \comp{\mi{request}}{nonce}$}
      \State \textbf{stop}
    \EndIf
    \State \textbf{remove} $\an{\mi{reference}, \mi{request}, \mi{url}, \mi{key}, f}$ \textbf{from} $\comp{s'}{pendingRequests}$
    \CallFun{PROCESSRESPONSE}{$m'$, $\mi{reference}$, $\mi{request}$, $\mi{url}$, $s'$}
  \ElsIf{$\proj{1}{m} \equiv \cHttpResp$ $\wedge$ $\exists\, \an{\mi{reference}, \mi{request}, \mi{url}, \bot, f}$ $\inPairing \comp{s'}{pendingRequests}$ \breakalgohook{0}\textbf{such that} $\comp{m'}{nonce} \equiv \comp{\mi{request}}{key}$ } %
    \State \textbf{remove} $\an{\mi{reference}, \mi{request}, \mi{url}, \bot, f}$ \textbf{from} $\comp{s'}{pendingRequests}$
    \CallFun{PROCESSRESPONSE}{$m$, $\mi{reference}$, $\mi{request}$, $\mi{url}$, $s'$}
  \ElsIf{$m \in \dnsresponses$} \Comment{Successful DNS response}
      \If{$\comp{m}{nonce} \not\in \comp{s}{pendingDNS} \vee \comp{m}{result} \not\in \addresses \vee \comp{m}{domain} \not\equiv \comp{\proj{2}{\comp{s}{pendingDNS}}}{host}$}
        \State \textbf{stop} \label{line:browser-dns-response-stop}
      \EndIf
      \Let{$\an{\mi{reference}, \mi{message}, \mi{url}}$}{$\comp{s}{pendingDNS}[\comp{m}{nonce}]$}
      \If{$\mi{url}.\str{protocol} \equiv \https$}
        \AppendBreak{2}{$\langle\mi{reference}$, $\mi{message}$, $\mi{url}$, $\nu_3$, $\comp{m}{result}\rangle$}{$\comp{s'}{pendingRequests}$} \label{line:add-to-pendingrequests-https}
        \Let{$\mi{message}$}{$\enc{\an{\mi{message},\nu_3}}{\comp{s'}{keyMapping}\left[\comp{\mi{message}}{host}\right]}$} \label{line:select-enc-key}
      \Else
        \AppendBreak{2}{$\langle\mi{reference}$, $\mi{message}$, $\mi{url}$, $\bot$, $\comp{m}{result}\rangle$}{$\comp{s'}{pendingRequests}$} \label{line:add-to-pendingrequests}
      \EndIf
      \Let{$\comp{s'}{pendingDNS}$}{$\comp{s'}{pendingDNS} - \comp{m}{nonce}$}
      \State \textbf{stop} $\an{\an{\comp{m}{result}, a, \mi{message}}}$, $s'$
  \Else \Comment{Some other message}
    \CallFun{PROCESS\_OTHER}{$m$, $a$, $f$, $s'$}
  \EndIf
  \State \textbf{stop}

\end{algorithmic} %
\end{algorithm}

\subsubsection{Functions} \label{app:proceduresbrowser} In
the description of the following functions, we use $a$,
$f$,
$m$,
and $s$
as read-only global input variables. All other variables
are local variables or arguments.

\begin{itemize}
\item The function $\mathsf{GETNAVIGABLEWINDOW}$
  (Algorithm~\ref{alg:getnavigablewindow}) is called by the
  browser to determine the window that is \emph{actually}
  navigated when a script in the window $s'.\ptr{w}$
  provides a window reference for navigation (e.g., for
  opening a link). When it is given a window reference
  (nonce) $\mi{window}$,
  this function returns a pointer to a selected window term
  in $s'$:
  \begin{itemize}
  \item If $\mi{window}$
    is the string $\wBlank$,
    a new window is created and a pointer to that window is
    returned.
  \item If $\mi{window}$
    is a nonce (reference) and there is a window term with
    a reference of that value in the windows in $s'$,
    a pointer $\ptr{w'}$
    to that window term is returned, as long as the window
    is navigable by the current window's document (as
    defined by $\mathsf{NavigableWindows}$ above).
  \end{itemize}
  In all other cases, $\ptr{w}$
  is returned instead (the script navigates its own
  window).

\item The function $\mathsf{GETWINDOW}$
  (Algorithm~\ref{alg:getwindow}) takes a window reference
  as input and returns a pointer to a window as above, but
  it checks only that the active documents in both windows
  are same-origin. It creates no new windows.

\item The function $\mathsf{CANCELNAV}$
  (Algorithm~\ref{alg:cancelnav}) is used to stop any
  pending requests for a specific window. From the pending
  requests and pending DNS requests it removes any requests
  with the given window reference $n$.

\item The function $\mathsf{HTTP\_SEND}$
  (Algorithm~\ref{alg:send}) takes an HTTP request
  $\mi{message}$
  as input, adds cookie and origin headers to the message,
  creates a DNS request for the hostname given in the
  request and stores the request in $\comp{s'}{pendingDNS}$
  until the DNS resolution finishes. For normal HTTP
  requests, $\mi{reference}$
  is a window reference. For \xhrs, $\mi{reference}$
  is a value of the form $\an{\mi{document}, \mi{nonce}}$
  where $\mi{document}$
  is a document reference and $\mi{nonce}$
  is some nonce that was chosen by the script that
  initiated the request. $\mi{url}$
  contains the full URL of the request (this is mainly used
  to retrieve the protocol that should be used for this
  message, and to store the fragment identifier for use
  after the document was loaded). $\mi{origin}$
  is the origin header value that is to be added to the
  HTTP request.

\item The functions $\mathsf{NAVBACK}$
  (Algorithm~\ref{alg:navback}) and $\mathsf{NAVFORWARD}$
  (Algorithm~\ref{alg:navforward}), navigate a window
  forward or backward. More precisely, they deactivate one
  document and activate that document's succeeding document
  or preceding document, respectively. If no such
  successor/predecessor exists, the functions do not change
  the state.

\item The function $\mathsf{RUNSCRIPT}$
  (Algorithm~\ref{alg:runscript}) performs a script
  execution step of the script in the document
  $\compn{s'}{\ptr{d}}$
  (which is part of the window $\compn{s'}{\ptr{w}}$).
  A new script and document state is chosen according to
  the relation defined by the script and the new script and
  document state is saved. Afterwards, the $\mi{command}$
  that the script issued is interpreted.

\item The function $\mathsf{PROCESSRESPONSE}$
  (Algorithm~\ref{alg:processresponse}) is responsible for
  processing an HTTP response ($\mi{response}$)
  that was received as the response to a request
  ($\mi{request}$)
  that was sent earlier. In $\mi{reference}$,
  either a window or a document reference is given (see
  explanation for Algorithm~\ref{alg:send} above).
  $\mi{requestUrl}$
  contains the URL used when retrieving the document.

  The function first saves any cookies that were contained
  in the response to the browser state, then checks whether
  a redirection is requested (Location header). If that is
  not the case, the function creates a new document (for
  normal requests) or delivers the contents of the response
  to the respective receiver (for \xhr responses).

\end{itemize}

\subsubsection{Definition}
We can now define the \emph{relation $R_\text{webbrowser}$
  of a web browser atomic process} as follows:
\begin{definition}
  The pair
  $\left(\left(\an{\an{a,f,m}}, s\right), \left(M,
      s'\right)\right)$ belongs to $R_\text{webbrowser}$
  if\/f the non-deterministic
  Algorithm~\ref{alg:browsermain} (or any of the functions
  called therein), when given $\left(\an{a,f,m}, s\right)$
  as input, terminates with \textbf{stop}~$M$,~$s'$,
  i.e., with output $M$ and $s'$.
\end{definition}

Recall that $\an{a,f,m}$
is an (input) event and $s$
is a (browser) state, $M$
is a sequence of (output) protoevents, and $s'$
is a new (browser) state (potentially with placeholders for
nonces).

\subsection{Definition of Web Browsers}
\label{sec:web-browsers-1}

Finally, we define web browser atomic Dolev-Yao processes
as follows:
\begin{definition}[Web Browser atomic Dolev-Yao Process]\label{def:webbrowser}
  A web browser atomic Dolev-Yao process is an atomic
  Dolev-Yao process of the form
  $p = (I^p, Z_{\text{webbrowser}}, R_{\text{webbrowser}},
  s{_o}^p)$ for a set $I^p$
  of addresses, $Z_{\text{webbrowser}}$
  and $R_{\text{webbrowser}}$
  as defined above, and an initial state $s{_0}^p
  \in Z_{\text{webbrowser}}$.
\end{definition}

\FloatBarrier %

\clearpage

\section{Generic HTTPS Server Model}
\label{sec:generic-https-server-model}

This model will be used as the base for all servers in the following.
It makes use of placeholder algorithms that are later superseded by
more detailed algorithms to describe a concrete relation for an HTTPS
server.

\begin{definition}[Base state for an HTTPS server.]\label{def:generic-https-state}
  The state of each HTTPS server that is an instantiation of this
  relation must contain at least the following subterms:
  $\mi{pendingDNS} \in \dict{\nonces}{\terms}$,
  $\mi{pendingRequests} \in \dict{\nonces}{\terms}$ (both containing
  arbitrary terms), $\mi{DNSaddress} \in \addresses$ (containing the
  IP address of a DNS server),
  $\mi{keyMapping} \in \dict{\dns}{\terms}$ (containing a mapping from
  domains to public keys), $\mi{tlskeys} \in \dict{\dns}{\nonces}$
  (containing a mapping from domains to private keys), and
  $\mi{corrupt}\in\terms$ (either $\bot$ if the server is not
  corrupted, or an arbitrary term otherwise).
\end{definition}
We note that in concrete instantiations of the generic HTTPS server model, there is no need to extract information from these subterms or alter these subterms. 

Let $\nu_{n0}$
and $\nu_{n1}$
denote placeholders for nonces that are not used in the concrete
instantiation of the server. We now define the default functions of
the generic web server in
Algorithms~\ref{alg:simple-send}--\ref{alg:default-other-handler}, and
the main relation in Algorithm~\ref{alg:generic-server-main}.

\begin{algorithm}[thp]
\caption{\label{alg:simple-send} Generic HTTPS Server Model: Sending a DNS message (in preparation for sending an HTTPS message). }
\begin{algorithmic}[1]
  \Function{$\mathsf{HTTPS\_SIMPLE\_SEND}$}{$\mi{reference}$, $\mi{message}$, $s'$}
    \Let{$\comp{s'}{pendingDNS}[\nu_{n0}]$}{$\an{\mi{reference},
        \mi{message}}$}
    \State \textbf{stop} $\an{\an{\comp{s'}{DNSaddress},a,
    \an{\cDNSresolve, \mi{message}.\str{host}, \nu_{n0}}}}$, $s'$
  \EndFunction
\end{algorithmic} %
\end{algorithm}

\begin{algorithm}[thp]
\caption{\label{alg:default-https-response-handler} Generic HTTPS Server Model: Default HTTPS response handler.}
\begin{algorithmic}[1]
  \Function{$\mathsf{PROCESS\_HTTPS\_RESPONSE}$}{$m$, $\mi{reference}$, $\mi{request}$, $\mi{key}$, $a$, $f$, $s'$}
    \Stop{\DefStop}
  \EndFunction
\end{algorithmic} %
\end{algorithm}

\begin{algorithm}[thp]
\caption{\label{alg:default-trigger} Generic HTTPS Server Model: Default trigger event handler.}
\begin{algorithmic}[1]
  \Function{$\mathsf{TRIGGER}$}{$s'$}
    \Stop{\DefStop}
  \EndFunction
\end{algorithmic} %
\end{algorithm}

\begin{algorithm}[thp]
\caption{\label{alg:default-https-request-handler} Generic HTTPS Server Model: Default HTTPS request handler.}
\begin{algorithmic}[1]
  \Function{$\mathsf{PROCESS\_HTTPS\_REQUEST}$}{$m$, $k$, $a$, $f$, $s'$}
    \Stop{\DefStop}
  \EndFunction
\end{algorithmic} %
\end{algorithm}

\begin{algorithm}[h!]
\caption{\label{alg:default-other-handler} Generic HTTPS Server Model: Default handler for other messages.}
\begin{algorithmic}[1]
  \Function{$\mathsf{PROCESS\_OTHER}$}{$m$, $a$, $f$, $s'$}
    \Stop{\DefStop}
  \EndFunction
\end{algorithmic} %
\end{algorithm}

\FloatBarrier

\begin{algorithm}[t!]
\caption{\label{alg:generic-server-main} Generic HTTPS Server Model: Main relation of a generic HTTPS server}
\begin{algorithmic}[1]
\Statex[-1] \textbf{Input:} $\an{a,f,m},s$
  \If{$s'.\str{corrupt} \not\equiv \bot \vee m \equiv \corrupt$}
    \Let{$s'.\str{corrupt}$}{$\an{\an{a, f, m}, s'.\str{corrupt}}$}
    \LetND{$m'$}{$d_{V}(s')$}
    \LetND{$a'$}{$\addresses$}
    \Stop{$\an{\an{a',a,m'}}$, $s'$}
  \EndIf

   \If{$\exists\, m_{\text{dec}}$, $k$, $k'$, $\mi{inDomain}$ \textbf{such that} $\an{m_{\text{dec}}, k} \equiv \dec{m}{k'}$ $\wedge$ $\an{inDomain,k'} \in s.\str{tlskeys}$} %
     \LetST{$n$, $\mi{method}$, $\mi{path}$, $\mi{parameters}$, $\mi{headers}$, $\mi{body}$}{\breakalgohook{0}$\an{\cHttpReq, n, \mi{method}, \mi{inDomain}, \mi{path}, \mi{parameters}, \mi{headers}, \mi{body}} \equiv m_{\text{dec}}$\breakalgohook{0}}{\textbf{stop} \DefStop}
     \CallFun{PROCESS\_HTTPS\_REQUEST}{$m_\text{dec}$, $k$, $a$, $f$, $s'$}

  \ElsIf{$m \in \dnsresponses$} \Comment{Successful DNS response}
      \If{$\comp{m}{nonce} \not\in \comp{s}{pendingDNS} \vee \comp{m}{result} \not\in \addresses \vee \comp{m}{domain} \not\equiv s.\str{pendingDNS}[m.\str{nonce}].2.\str{host}$}
      \Stop{\DefStop}
      \EndIf
      \Let{$\an{\mi{reference}, \mi{request}}$}{$\comp{s}{pendingDNS}[\comp{m}{nonce}]$}
      \AppendBreak{2}{$\langle\mi{reference}$, $\mi{request}$, $\nu_{n1}$, $\comp{m}{result}\rangle$}{$\comp{s'}{pendingRequests}$} \label{line:generic-move-reference-to-pending-request} 
      \Let{$\mi{message}$}{$\enc{\an{\mi{request},\nu_{n1}}}{\comp{s'}{keyMapping}\left[\comp{\mi{request}}{host}\right]}$} \label{line:generic-select-enc-key}
      \Let{$\comp{s'}{pendingDNS}$}{$\comp{s'}{pendingDNS} - \comp{m}{nonce}$}\label{line:generic-remove-pendingdns}
      \Stop{$\an{\an{\comp{m}{result}, a, \mi{message}}}$, $s'$} \label{line:generic-send-https-request}

  \ElsIf{$\exists\, \an{\mi{reference}, \mi{request}, \mi{key}, f}$
      $\inPairing \comp{s'}{pendingRequests}$ \breakalgohook{0}
      \textbf{such that} $\proj{1}{\decs{m}{\mi{key}}} \equiv \cHttpResp$ } \label{line:generic-https-response}
    \Comment{Encrypted HTTP response}

    \Let{$m'$}{$\decs{m}{\mi{key}}$}
    \If{$\comp{m'}{nonce} \not\equiv \comp{\mi{request}}{nonce}$}
      \Stop{\DefStop}
    \EndIf
    \State \textbf{remove} $\an{\mi{reference}, \mi{request}, \mi{key}, f}$ \textbf{from} $\comp{s'}{pendingRequests}$\label{line:generic-remove-pending-request}
    \CallFun{PROCESS\_HTTPS\_RESPONSE}{$m'$, $\mi{reference}$, $\mi{request}$, $\mi{key}$, $a$, $f$, $s'$}
    \Stop{\DefStop}

  \ElsIf{$m \equiv \str{TRIGGER}$} \Comment{Process was triggered}
    \CallFun{PROCESS\_TRIGGER}{$s'$}

  \EndIf
  \Stop{\DefStop}
  
\end{algorithmic} %
\end{algorithm}

\FloatBarrier %

\clearpage
\section{Formal Model of OpenID Connect with a Network Attacker}
\label{app:model-oidc-auth}

We here present the full details of our formal model of OIDC which we
use to analyze the authentication and authorization properties. This
model contains a network attacker. We will later derive from this
model a model where the network attacker is replaced by a web
attacker. We use this modified model for the session integrity
properties.

We model OIDC as a web system (in the sense of
Appendix~\ref{app:websystem}). We call a web system
$\oidcwebsystem^n=(\bidsystem, \scriptset, \mathsf{script}, E^0)$
an \emph{OIDC web system with a network attacker} if it is of the
form described in what follows.

\subsection{Outline}\label{app:outlineoidcmodel}
The system $\bidsystem=\mathsf{Hon} \cup \mathsf{Net}$
consists of a network attacker process (in $\mathsf{Net}$),
a finite set $\fAP{B}$
of web browsers, a finite set $\fAP{RP}$
of web servers for the relying parties, a finite set $\fAP{IDP}$
of web servers for the identity providers, with
$\mathsf{Hon} := \fAP{B} \cup \fAP{RP} \cup \fAP{IDP}$.
More details on the processes in $\bidsystem$
are provided below. We do not model DNS servers, as they are subsumed
by the network attacker.
Figure~\ref{fig:scripts-in-w} shows the set of scripts $\scriptset$
and their respective string representations that are defined by the
mapping $\mathsf{script}$.
The set $E^0$ contains only the trigger events as specified in
Appendix~\ref{app:websystem}. 

\begin{figure}[htb]
  \centering
  \begin{tabular}{|@{\hspace{1ex}}l@{\hspace{1ex}}|@{\hspace{1ex}}l@{\hspace{1ex}}|}\hline 
    \hfill $s \in \scriptset$\hfill  &\hfill  $\mathsf{script}(s)$\hfill  \\\hline\hline
    $\Rasp$ & $\str{att\_script}$  \\\hline
    $\mi{script\_rp\_index}$ & $\str{script\_rp\_index}$  \\\hline
    $\mi{script\_rp\_get\_fragment}$ & $\str{script\_get\_fragment}$  \\\hline
    $\mi{script\_idp\_form}$ &  $\str{script\_idp\_form}$  \\\hline
  \end{tabular}
  
  \caption{List of scripts in $\scriptset$ and their respective string
    representations.}
  \label{fig:scripts-in-w}
\end{figure}

This outlines $\oidcwebsystem^n$. We will now define the DY processes in
$\oidcwebsystem^n$ and their addresses, domain names, and secrets in more
detail. 

\subsection{Addresses and Domain Names}\label{app:addresses-and-domain-names}
The set $\addresses$
contains for the network attacker in $\fAP{Net}$,
every relying party in $\fAP{RP}$,
every identity provider in $\fAP{IDP}$,
and every browser in $\fAP{B}$
a finite set of addresses each. By $\mapAddresstoAP$
we denote the corresponding assignment from a process to its address.
The set $\dns$
contains a finite set of domains for every relying party in
$\fAP{RP}$,
every identity provider in $\fAP{IDP}$,
and the network attacker in $\fAP{Net}$.
Browsers (in $\fAP{B})$ do not have a domain.

By $\mapAddresstoAP$ and $\mapDomain$ we denote the assignments from
atomic processes to sets of $\addresses$ and $\dns$, respectively.

\subsection{Keys and Secrets} The set $\nonces$
of nonces is partitioned into five sets, an infinite sequence $N$,
an finite set $K_\text{TLS}$,
an finite set $K_\text{sign}$,
and a finite set $\mathsf{Passwords}$. We therefore have
$
  \def\hereMaxHeightPhantom{\vphantom{K_{\text{p}}^\bidsystem}}
  \nonces = 
  \underbrace{N\hereMaxHeightPhantom}_{\text{infinite sequence}} 
  \dot\cup \underbrace{K_{\text{TLS}}\hereMaxHeightPhantom}_{\text{finite}} 
  \dot\cup \underbrace{K_{\text{sign}}\hereMaxHeightPhantom}_{\text{finite}} 
  \dot\cup \underbrace{\mathsf{Passwords}\hereMaxHeightPhantom}_{\text{finite}}
$.

\noindent
These sets are used as follows:
\begin{itemize}
\item The set $N$ contains the nonces that are available for each DY
  process in $\bidsystem$ (it can be used to create a run of
  $\bidsystem$).

\item The set $K_\text{TLS}$
  contains the keys that will be used for TLS encryption. Let
  $\mapTLSKey\colon \dns \to K_\text{TLS}$
  be an injective mapping that assigns a (different) private key to
  every domain. For an atomic DY process $p$
  we define
  $\mi{tlskeys}^p = \an{\left\{\an{d, \mapTLSKey(d)} \mid d \in
      \mapDomain(p)\right\}}$.

\item The set $K_\text{sign}$
  contains the keys that will be used by IdPs for signing id tokens. Let
  $\mathsf{signkey}\colon \fAP{IDP} \to K_\text{sign}$
  be an injective mapping that assigns a (different) signing key to
  every IdP. 

\item The set $\mathsf{Passwords}$
  is the set of passwords (secrets) the browsers share with the
  identity providers. These are the passwords the users use to log in
  at the IdPs.
\end{itemize}

\subsection{Identities and Passwords}\label{app:oidc-pidp-identities}
Identites consist, similar to email addresses, of a user name and a
domain part. For our model, this is defined as follows:
\begin{definition}
  An \emph{identity} (email address) $i$ is a term of the form
  $\an{\mi{name},\mi{domain}}$ with $\mi{name}\in \mathbb{S}$ and
  $\mi{domain} \in \dns$.

  Let $\IDs$
  be the finite set of identities. We say that an ID is
  \emph{governed} by the DY process to which the domain of the ID
  belongs. Formally, we define the mapping
  $\mapGovernor: \IDs \to \bidsystem$,
  $\an{\mi{name}, \mi{domain}} \mapsto \mapDomain^{-1}(\mi{domain})$.
  By $\IDs^y$
  we denote the set $\mathsf{governor}^{-1}(y)$.
\end{definition}%
The governor of an ID will usually be an IdP, but could also be the
attacker. Besides $\mapGovernor$, we define the following mappings:

\begin{itemize}
\item By $\mapIDtoPLI:\IDs \to \mathsf{Passwords}$ we denote the bijective
  mapping that assigns secrets to all identities.

\item   Let $\mapPLItoOwner: \mathsf{Passwords} \to \fAP{B}$ denote the mapping that
  assigns to each secret a browser that \emph{owns} this secret. Now,
  we define the mapping $\mapIDtoOwner: \IDs \to \fAP{B}$, $i \mapsto
  \mapPLItoOwner(\mapIDtoPLI(i))$, which assigns to each identity the
  browser that owns this identity (we say that the identity belongs to
  the browser).
\end{itemize}

\subsection{Corruption}
RPs and IdPs can become corrupted: If they receive the message
$\corrupt$, they start collecting all incoming messages in their state
and (upon triggering) send out all messages that are derivable from
their state and collected input messages, just like the attacker
process. We say that an RP or an IdP is \emph{honest} if the according
part of their state ($s.\str{corrupt}$) is $\bot$, and that they are
corrupted otherwise.

We are now ready to define the processes in $\websystem$ as well as
the scripts in $\scriptset$ in more detail.

\subsection{Network Attackers}\label{app:networkattackers-oidc} As mentioned, the network attacker
$\mi{na}$
is modeled to be a network attacker as specified in
Appendix~\ref{app:websystem}. We allow it to listen to/spoof all
available IP addresses, and hence, define $I^\mi{na} = \addresses$.
The initial state is
$s_0^\mi{na} = \an{\mi{attdoms}, \mi{tlskeys}, \mi{signkeys}}$,
where $\mi{attdoms}$
is a sequence of all domains along with the corresponding private keys
owned by the attacker $\mi{na}$,
$\mi{tlskeys}$
is a sequence of all domains and the corresponding public keys, and
$\mi{signkeys}$
is a sequence containing all public signing keys for all IdPs.

\subsection{Browsers}\label{app:browsers-oidc} 

Each $b \in \fAP{B}$
is a web browser atomic Dolev-Yao process as defined in
Definition~\ref{def:webbrowser}, with $I^b := \mapAddresstoAP(b)$
being its addresses.

To define the inital state, first let $\IDs_b :=
\mapIDtoOwner^{-1}(b)$ be
 the set of all IDs of $b$. We then define the set of passwords that a browser $b$ gives to an origin $o$: If the origin belongs to an IdP, then the user's passwords of this IdP are contained in the set. To define this mapping in the initial state, we first define for some process~$p$

\[
 \mathsf{Secrets}^{b,p} = \big\{ s \bigm| b = \mathsf{ownerOfSecret}(s)  \wedge  (\exists\, i : s = \mathsf{secretOfID}(i) \wedge i \in \mathsf{ID}^p)   \big\} \,.
\]

Then, the initial state $s_0^b$
is defined as follows: the key mapping maps every domain to its public
(TLS) key, according to the mapping $\mapTLSKey$;
the DNS address is an address of the network attacker; the list of
secrets contains an entry
$\an{\an{d,\https}, \an{\mathsf{Secrets}^{b,p}}}$
for each $p \in \fAP{RP} \cup \fAP{IDP}$
and $d \in \mathsf{dom}(p)$;
$\mi{ids}$ is $\an{\IDs_b}$; $\mi{sts}$ is empty.

\subsection{Relying Parties} \label{app:relying-parties-oidc}

A relying party $r \in \fAP{RP}$ is a web server modeled as an atomic
DY process $(I^r, Z^r, R^r, s^r_0)$ with the addresses
$I^r := \mapAddresstoAP(r)$.  Next, we define the set $Z^r$ of states
of $r$ and the initial state $s^r_0$ of~$r$.

\begin{algorithm}[tbp]
\caption{\label{alg:rp-oidc-http-response} Relation of a Relying
  Party $R^r$ -- Processing HTTPS Responses}
\begin{algorithmic}[1]
\Function{$\mathsf{PROCESS\_HTTPS\_RESPONSE}$}{$m$, $\mi{reference}$, $\mi{request}$, $\mi{key}$, $a$, $f$, $s'$}
  \Let{$\mi{session}$}{$s'.\str{sessions}[\mi{reference}[\str{session}]]$}
  \Let{$\mi{id}$}{$\mi{session}[\str{identity}]$} \label{line:rp-lookup-session}
  \Let{$\mi{issuer}$}{$s'.\str{issuerCache}[\mi{id}]$}
  \If{$\mi{reference}[\str{responseTo}] \equiv \str{WEBFINGER}$}
    \Let{$\mi{wf}$}{$m.\str{body}$}
    \If{$\mi{wf}[\str{subject}] \not\equiv \mi{id}$}
      \Stop
    \EndIf
    \If{$\mi{wf}[\str{links}][\str{rel}] \not\equiv \str{OIDC\_issuer}$}
      \Stop
    \EndIf
    \Let{$s'.\str{issuerCache}[\mi{id}]$}{$\mi{wf}[\str{links}][\str{href}]$} \label{line:rp-set-issuercache}
    \CallFun{START\_LOGIN\_FLOW}{$\mi{reference}[\str{session}]$, $s'$}
  \ElsIf{$\mi{reference}[\str{responseTo}] \equiv \str{CONFIG}$}
    \Let{$\mi{oidcc}$}{$m.\str{body}$}
    \If{$\mi{oidcc}[\str{issuer}] \not\equiv \mi{issuer}$}
      \Stop
    \EndIf
    \Let{$s'.\str{oidcConfigCache}[\mi{issuer}]$}{$\mi{oidcc}$}  \label{line:rp-set-oidc-config-cache}
    \CallFun{START\_LOGIN\_FLOW}{$\mi{reference}[\str{session}]$, $s'$}
  \ElsIf{$\mi{reference}[\str{responseTo}] \equiv \str{JWKS}$}
    \Let{$s'.\str{jwksCache}[\mi{issuer}]$}{$m.\str{body}$}
    \CallFun{START\_LOGIN\_FLOW}{$\mi{reference}[\str{session}]$, $s'$}
  \ElsIf{$\mi{reference}[\str{responseTo}] \equiv \str{REGISTRATION}$}
    \Let{$s'.\str{clientCredentialsCache}[\mi{issuer}]$}{$m.\str{body}$}  \label{line:rp-set-clientid}
    \CallFun{START\_LOGIN\_FLOW}{$\mi{reference}[\str{session}]$, $s'$}
  \ElsIf{$\mi{reference}[\str{responseTo}] \equiv \str{TOKEN}$}
    \If{$\str{token} \inPairing \mi{session}[\str{response\_type}] \wedge \mi{useAccessTokenNow} \equiv \True$}
      \CallFun{USE\_ACCESS\_TOKEN}{$\mi{reference}[\str{session}]$, $m.\str{body}[\str{access\_token}]$, $s'$} \label{line:rp-call-use-access-token-2}
    \EndIf
    \CallFun{CHECK\_ID\_TOKEN}{$\mi{reference}[\str{session}]$, $m.\str{body}[\str{id\_token}]$, $s'$} \label{line:rp-call-check-id-token-after-code}
  \EndIf
  \Stop
\EndFunction
\end{algorithmic} %
\end{algorithm}

\begin{algorithm}[tbp]
\caption{\label{alg:rp-oidc-http-request} Relation of a Relying
  Party $R^r$ -- Processing HTTPS Requests}
\begin{algorithmic}[1]
\Function{$\mathsf{PROCESS\_HTTPS\_REQUEST}$}{$m$, $k$, $a$, $f$, $s'$}
\Comment{\textbf{Process an incoming HTTPS request.} Other message types are handled in separate functions. $m$ is the incoming message, $k$ is the encryption key for the response, $a$ is the receiver, $f$ the sender of the message. $s'$ is the current state of the atomic DY process $r$.}%
     \If{$m.\str{path} \equiv \str{/}$} \label{line:rp-serve-index} \Comment{Serve index page.}
      \Let{$\mi{headers}$}{$[\str{ReferrerPolicy}{:}\str{origin}]$}
       \ParboxComment{Set the Referrer Policy for the index page of the RP (cf. Section~\ref{sec:state-leak-attack}).}
       \Let{$m'$}{$\encs{\an{\cHttpResp, m.\str{nonce}, 200, \mi{headers}, \an{\str{script\_rp\_index}, \an{}}}}{k}$}
       \ParboxComment{Send $\mi{script\_rp\_index}$ in HTTP response.}
       \Stop{$\an{\an{f,a,m'}}$, $s'$} \label{line:rp-send-index}
     \ElsIf{$m.\str{path} \equiv \str{/startLogin} \wedge m.\str{method} \equiv \mPost$} \label{line:rp-start-login-endpoint} \Comment{\textbf{Serve the request to start a new login.}}
       \If{$m.\str{headers}[\str{Origin}]  \not\equiv \an{m.\str{host}, \https}$}  
         \Stop{\DefStop} \Comment{Check the Origin header for CSRF protection.}
       \EndIf
       \Let{$\mi{id}$}{$m.\str{body}$}
       \Let{$\mi{sessionId}$}{$\nu_1$} \label{line:rp-choose-lsid} \ParboxComment{Session id is a freshly chosen nonce.}
       \Let{$s'.\str{sessions}[\mi{sessionId}]$}{$[\str{startRequest}{:}[\str{message}{:}m,\str{key}{:}k,\str{receiver}{:}a,\str{sender}{:}f], \str{identity}:\mi{id}]$}
       \ParboxComment{Create new session record.}
       \CallFun{START\_LOGIN\_FLOW}{$\mi{sessionId}$, $s'$}
\ParboxComment{Call the function that starts or proceeds with a login flow. Runs discovery/registration/etc.}%

     \ElsIf{$m.\str{path} \equiv \str{/redirect\_ep}$} \label{line:rp-redir-endpoint} \ParboxComment{\textbf{User is being redirected after authentication to the IdP.}}
       \Let{$\mi{sessionId}$}{$m.\str{headers}[\str{Cookie}][\str{sessionId}]$}
       \If{$\mi{sessionId} \not\in s'.\str{sessions}$}
         \Stop
       \EndIf
       \Let{$\mi{session}$}{$s'.\str{sessions}[\mi{sessionId}]$}
       \ParboxComment{Retrieve session data.}
       \Let{$\mi{id}$}{$\mi{session}[\str{id}]$}
       \Let{$\mi{issuer}$}{$s'.\str{issuerCache}[\mi{identity}]$}
       \ParboxComment{Issuer cache contains mappings from identites to issuers. Caches are being filled during discovery/registration in the function $\mathsf{START\_LOGIN\_FLOW}$.}
       \If{$m.\str{parameters}[\str{iss}] \not\equiv \mi{issuer}$}\label{line:rp-check-iss-redirect}
         \Stop{\DefStop}\ParboxComment{Check issuer parameter (cf. Section~\ref{sec:idp-mix-up}).}
       \EndIf
       \Let{$\mi{oidcConfig}$}{$s'.\str{oidcConfigCache}[\mi{issuer}]$}
       \ParboxComment{Retrieve OIDC configuration for issuer.}
       \Let{$\mi{responseType}$}{$\mi{session}[\str{response\_type}]$}

       \ParboxComment{Determines the OIDC flow to use, e.g., \texttt{code id\_token token} for a hybrid flow.}

       \If{$\mi{responseType} \equiv \an{\str{code}}$} \ParboxComment{Authorization code mode: Take data from URL parameters.}
         \Let{$\mi{data}$}{$m.\str{parameters}$}
       \Else \ParboxComment{Hybrid or implicit mode: Send the script $\str{script\_rp\_get\_fragment}$ to the browser to retrieve data from URL fragment}
         \If{$m.\str{method} \equiv \mGet$}
           \Let{$\mi{headers}$}{$\an{\an{\str{ReferrerPolicy}, \str{origin}}}$}
           \Let{$m'$}{$\encs{\an{\cHttpResp, m.\str{nonce}, 200, \mi{headers}, \an{\str{script\_rp\_get\_fragment}, \bot}}}{k}$} \label{line:rp-send-script-get-fragment}
           \Stop{$\an{\an{f,a,m'}}$, $s'$}
         \Else
           \ParboxComment{If this is a POST request, the script $\str{script\_rp\_get\_fragment}$ is sending the data from URL fragment.}
           \Let{$\mi{data}$}{$m.\str{body}$}
         \EndIf
       \EndIf
       \If{$\mi{data}[\str{state}] \not\equiv \mi{session}[\str{state}]$}
         \Stop \ParboxComment{Check $\mi{state}$ value.}
       \EndIf
       \Let{$s'.\str{sessions}[\mi{sessionId}][\str{redirectEpRequest}]$}{\breakalgohook{2}$[\str{message}{:}m,\str{key}{:}k,\str{receiver}{:}a,\str{sender}{:}f]$} \label{line:rp-set-redirect-ep-request-record} \ParboxComment{Store incoming request for later use in $\mathsf{CHECK\_ID\_TOKEN}$ (Algorithm~\ref{alg:rp-check-id-token}).}
       \If{$\str{id\_token} \inPairing \mi{responseType}$}
         \ParboxComment{Check if the chosen response type contains $\str{id\_token}$.}
         \If{$\str{code} \inPairing \mi{responseType}$} \ParboxComment{In hybrid mode, only one of the two id tokens must be checked(cf. Appendix~\ref{app:idp-mix-up}).}
           \LetND{$\mi{checkIdTokenNow}$}{$\{\True,\bot\}$}
         \ParboxComment{Non-deterministically decide whether to check first or second id token to capture all potential choices in real-world implementations (cf. Appendix~\ref{app:idp-mix-up}).}
         \Else
           \Let{$\mi{checkIdTokenNow}$}{$\True$}
         \ParboxComment{In non-hybrid modes, the id token is always checked.}
         \EndIf
         \If{$\mi{checkIdTokenNow} \equiv \True$}\footnote{Note that the OIDC standard requires that this id token has to be checked. In our model, however, we overapproximate by non-deterministically omitting this check.}
           \CallFun{CHECK\_ID\_TOKEN}{$\mi{sessionId}$, $\mi{data}[\str{id\_token}]$, $s'$} \label{line:rp-call-check-id-token-immediately} \ParboxComment{Check the id token and (if successful) log the user in. See Algorithm~\ref{alg:rp-check-id-token}.}
         \EndIf
       \EndIf
       \LetND{$\mi{useAccessTokenNow}$}{$\{\True, \bot\}$}
       \ParboxComment{Non-det. decide whether to use access token (authorization) or not.}
       \If{$\str{token} \inPairing \mi{responseType} \wedge \mi{useAccessTokenNow} \equiv \True$}
         \CallFun{USE\_ACCESS\_TOKEN}{$\mi{sessionId}$, $m.\str{body}[\str{access\_token}]$, $s'$} \label{line:rp-call-use-access-token-1} \ParboxComment{Use the access token at the IdP.}
       \EndIf
       \If{$\str{code} \inPairing \mi{responseType}$}
         \CallFun{SEND\_TOKEN\_REQUEST}{$\mi{sessionId}$, $m.\str{body}[\str{code}]$, $s'$}
       \ParboxComment{Retrieve a token from the token endpoint of the IdP.}
       \EndIf
     \EndIf
     \Stop
\EndFunction
\end{algorithmic} %
\end{algorithm}

\begin{algorithm}
\caption{\label{alg:rp-token-request} Relying Party $R^r$: Request to token endpoint.}
\begin{algorithmic}[1]
\Function{$\mathsf{SEND\_TOKEN\_REQUEST}$}{$\mi{sessionId}$, $\mi{code}$, $s'$}
  \Let{$\mi{session}$}{$s'.\str{sessions}[\mi{sessionId}]$}
  \Let{$\mi{identity}$}{$\mi{session}[\str{identity}]$}
  \Let{$\mi{issuer}$}{$s'.\str{issuerCache}[\mi{identity}]$}
  \Let{$\mi{credentials}$}{$s'.\str{clientCredentialsCache}[\mi{issuer}]$}
  \Let{$\mi{headers}$}{$[]$}
  \Let{$\mi{body}$}{$[\str{grant\_type}{:} \str{authorization\_code}, \str{code}{:}\mi{code}, \str{redirect\_uri}{:}\mi{session}[\str{redirect\_uri}]]$} \label{line:rp-prepare-atoken-from-code-req}

  \Let{$\mi{clientId}$}{$\mi{credentials}[\str{client\_id}]$}
  \Let{$\mi{clientSecret}$}{$\mi{credentials}[\str{client\_secret}]$} \label{line:rp-send-client-password-2}
  \If{$\mi{clientSecret} \equiv \an{}$}
    \Let{$\mi{body}[\str{client\_id}]$}{$\mi{clientId}$}
  \Else
    \Let{$\mi{headers}[\str{Authorization}]$}{$\an{\mi{clientId},\mi{clientSecret}}$}
  \EndIf
         
  \Let{$\mi{url}$}{$s'.\str{oidcConfigCache}[\mi{issuer}][\str{token\_ep}]$}
  \Let{$\mi{message}$}{$\hreq{ nonce=\nu_2,
      method=\mPost,
      xhost=\mi{url}.\str{domain},
      path=\mi{url}.\str{path},
      parameters=\mi{url}.\str{parameters},
      headers=\mi{headers},
      xbody=\mi{body}}$}  \label{line:rp-send-something-3}
  \CallFun{HTTPS\_SIMPLE\_SEND}{$[\str{responseTo}{:}\str{TOKEN},\str{session}{:}\mi{sessionId}]$, $\mi{message}$, $s'$} \label{line:rp-start-retrieve-code}
\EndFunction
\end{algorithmic} %
\end{algorithm}

\begin{algorithm}[tbp]
\caption{\label{alg:rp-use-access-token} Relying Party $R^r$: Using the access token (no response expected).}
\begin{algorithmic}[1]
\Function{$\mathsf{USE\_ACCESS\_TOKEN}$}{$\mi{sessionId}$, $\mi{token}$, $s'$}
  \Let{$\mi{session}$}{$s'.\str{sessions}[\mi{sessionId}]$}
  \Let{$\mi{identity}$}{$\mi{session}[\str{identity}]$}
  \Let{$\mi{issuer}$}{$s'.\str{issuerCache}[\mi{identity}]$}
  \Let{$\mi{headers}$}{$[\str{Authorization}:\an{\str{Bearer},\mi{token}}]$} 
         
  \Let{$\mi{url}$}{$s'.\str{oidcConfigCache}[\mi{issuer}][\str{token\_ep}]$}
  \LetND{$\mi{url}.\str{path}$}{$\mathbb{S}$}
  \Let{$\mi{message}$}{$\hreq{ nonce=\nu_3,
      method=\mPost,
      xhost=\mi{url}.\str{domain},
      path=\mi{url}.\str{path},
      parameters=\mi{url}.\str{parameters},
      headers=\mi{headers},
      xbody=\an{}}$}  
  \CallFun{HTTPS\_SIMPLE\_SEND}{$[\str{responseTo}{:}\str{RESOURCE\_USAGE},\str{session}{:}\mi{sessionId}]$, $\mi{message}$, $s'$}
\EndFunction
\end{algorithmic} %
\end{algorithm}

\begin{algorithm}[tbp]
\caption{\label{alg:rp-check-id-token} Relying Party $R^r$: Check id token.}
\begin{algorithmic}[1]
\Function{$\mathsf{CHECK\_ID\_TOKEN}$}{$\mi{sessionId}$, $\mi{id\_token}$, $s'$}
  \Comment{\textbf{Check id token validity and create service session.}}
  \Let{$\mi{session}$}{$s'.\str{sessions}[\mi{sessionId}]$}
  \ParboxComment{Retrieve session data.}
  \Let{$\mi{identity}$}{$\mi{session}[\str{identity}]$}
  \Let{$\mi{issuer}$}{$s'.\str{issuerCache}[\mi{identity}]$}
  \ParboxComment{Retrieve issuer.}
  \Let{$\mi{oidcConfig}$}{$s'.\str{oidcConfigCache}[\mi{issuer}]$}
  \ParboxComment{Retrieve OIDC configuration for that issuer.}
  \Let{$\mi{credentials}$}{$s'.\str{clientCredentialsCache}[\mi{issuer}]$}
  \ParboxComment{Retrieve OIDC credentials for issuer.}
  \Let{$\mi{jwks}$}{$s'.\str{jwksCache}[\mi{issuer}]$}
  \ParboxComment{Retrieve signing keys for issuer.}
  \Let{$\mi{data}$}{$\mathsf{extractmsg}(\mi{id\_token})$}
  \ParboxComment{Extract contents of signed id token.}
  \If{$\mi{data}[\str{iss}] \not\equiv \mi{issuer}$} \label{line:rp-start-id-token-checks}
    \Stop \ParboxComment{Check the issuer.}
  \EndIf
  \If{$\mi{data}[\str{aud}] \not\equiv \mi{credentials}[\str{client\_id}]$}
    \Stop \ParboxComment{Check the audience against own client id.}
  \EndIf
  \If{$\mathsf{checksig}(\mi{id\_token},\mi{jwks}) \not\equiv \top$}
    \Stop \ParboxComment{Check the signature of the id token.}
  \EndIf
  \If{$\mi{nonce} \in \mi{session}
    \wedge \mi{data}[\str{nonce}] \not\equiv \mi{session}[\str{nonce}]$}
      \Stop    \ParboxComment{If a nonce was used, check its value.}
  \EndIf
  \Let{$s'.\str{sessions}[\mi{sessionId}][\str{loggedInAs}]$}{$\an{\mi{issuer}, \mi{data}[\str{sub}]}$}
  \ParboxComment{User is now logged in. Store user identity and issuer.}
  \Let{$s'.\str{sessions}[\mi{sessionId}][\str{serviceSessionId}]$}{$\nu_4$} \label{line:rp-choose-service-session-id} \ParboxComment{Choose a new service session id.}
  \Let{$\mi{request}$}{$\mi{session}[\str{redirectEpRequest}]$}
  \ParboxComment{Retrieve stored meta data of the request from the browser to the redir. endpoint in order to respond to it now. The request's meta data was stored in $\mathsf{PROCESS\_HTTPS\_REQUEST}$ (Algorithm~\ref{alg:rp-oidc-http-request}). }
  \Let{$\mi{headers}$}{$[\str{ReferrerPolicy}{:}\str{origin}]$}
  \Let{$\mi{headers}[\cSetCookie]$}{$[\str{serviceSessionId}{:}\an{\nu_4,\top,\top,\top}]$}
  \ParboxComment{Create a cookie containing the service session id.}
  \Let{$m'$}{$\encs{\an{\cHttpResp, \mi{request}[\str{message}].\str{nonce}, 200, \mi{headers}, \str{ok}}}{\mi{request}[\str{key}]}$}
  \ParboxComment{Respond to browser's request to the redirection endpoint.}
  \Stop{$\an{\an{\mi{request}[\str{sender}],\mi{request}[\str{receiver}],m'}}$, $s'$} \label{line:rp-send-set-service-session}
\EndFunction
\end{algorithmic} %
\end{algorithm}

\begin{algorithm}[tbp]
\caption{\label{alg:rp-oidc-cont-login-flow} Relying Party $R^r$: Continuing in the login flow.}
\begin{algorithmic}[1]
\Function{$\mathsf{START\_LOGIN\_FLOW}$}{$\mi{sessionId}$, $s'$}
  \Let{$\mi{redirectUris}$}{$\{\an{\cUrl, \https, d, \str{/redirect\_ep}, \an{},\an{}} |\, d \in \mathsf{dom}(r)\}$} \label{line:rp-choose-redirect-uris} \Comment{Set of redirect URIs for all domains.}
  \Let{$\mi{session}$}{$s'.\str{sessions}[\mi{sessionId}]$}
  \Let{$\mi{identity}$}{$\mi{session}[\str{identity}]$}
  \If{$\mi{identity} \not\in s'.\str{issuerCache}$}
    \Let{$\mi{host}$}{$\mi{identity}.\str{domain}$} \label{line:rp-webfinger-domain}
    \Let{$\mi{path}$}{$\str{/.wk/webfinger}$}
    \Let{$\mi{parameters}$}{$[\str{resource}:\mi{identity}]$}
    \Let{$\mi{message}$}{$\hreq{ nonce=\nu_{5},
        method=\mGet,
        xhost=\mi{host},
        path=\mi{path},
        parameters=\mi{parameters},
        headers=\an{},
        xbody=\an{}}$}
    \CallFun{HTTPS\_SIMPLE\_SEND}{$[\str{responseTo}{:}\str{WEBFINGER}, \str{session}{:}\mi{sessionId}]$, $\mi{message}$, $s'$} \label{line:rp-create-webfinger-request}
  \EndIf
  \Let{$\mi{issuer}$}{$s'.\str{issuerCache}[\mi{identity}]$}
  \If{$\mi{issuer} \not\in s'.\str{oidcConfigCache}$}
    \Let{$\mi{host}$}{$\mi{issuer}$}
    \Let{$\mi{path}$}{$\str{/.wk/openid\mhyphen{}configuration}$}
    \Let{$\mi{message}$}{$\hreq{ nonce=\nu_{5},
        method=\mGet,
        xhost=\mi{host},
        path=\mi{path},
        parameters=[],
        headers=\an{},
        xbody=\an{}}$}
    \CallFun{HTTPS\_SIMPLE\_SEND}{$[\str{responseTo}{:}\str{CONFIG}, \str{session}{:}\mi{sessionId}]$, $\mi{message}$, $s'$} \label{line:rp-create-oidc-config-request}
  \EndIf
  \Let{$\mi{oidcConfig}$}{$s'.\str{oidcConfigCache}[\mi{issuer}]$}
  \If{$\mi{issuer} \not\in s'.\str{jwksCache}$}
    \Let{$\mi{url}$}{$\mi{oidcConfig}[\str{jwks\_uri}]$}
    \Let{$\mi{message}$}{$\hreq{ nonce=\nu_{5},
        method=\mGet,
        xhost=\mi{url}.\str{host},
        path=\mi{url}.\str{path},
        parameters=[],
        headers=\an{},
        xbody=\an{}}$}
    \CallFun{HTTPS\_SIMPLE\_SEND}{$[\str{responseTo}{:}\str{JWKS}, \str{session}{:}\mi{sessionId}]$, $\mi{message}$, $s'$} \label{line:rp-create-jwks-request}
  \EndIf
  \If{$\mi{issuer} \not\in s'.\str{clientCredentialsCache}$}
    \Let{$\mi{url}$}{$\mi{oidcConfig}[\str{reg\_ep}]$} \label{line:rp-create-registration-request}
    \Let{$\mi{message}$}{$\hreq{ nonce=\nu_{5},
        method=\mPost,
        xhost=\mi{url}.\str{host},
        path=\mi{url}.\str{path},
        parameters=[],
        headers=\an{},
        xbody=[\str{redirect\_uris}:\an{\mi{redirectUris}}]}$}
    \CallFun{HTTPS\_SIMPLE\_SEND}{$[\str{responseTo}{:}\str{REGISTRATION}, \str{session}:\mi{sessionId}]$, $\mi{message}$, $s'$} 
  \EndIf
  \Let{$\mi{credentials}$}{$s'.\str{clientCredentialsCache}[\mi{issuer}]$}

  \LetND{$\mi{responseType}$}{$\{\an{\str{code}}, \an{\str{id\_token}}, \an{\str{id\_token}, \str{token}}, \an{\str{code}, \str{id\_token}},$ \breakalgohook{6}$\an{\str{code}, \str{token}}, \an{\str{code}, \str{id\_token}, \str{token}} \}$}
  \LetND{$\mi{redirectUri}$}{$\mi{redirectUris}$}
  \Let{$\mi{data}$}{$[\str{response\_type}{:}\mi{responseType}, \str{redirect\_uri}{:}\mi{redirectUri},$ 
    \breakalgohook{6} $\str{client\_id}{:}\mi{credentials}[\str{client\_id}], \str{state}{:}\nu_6]$}
  \If{$\str{code} \not\inPairing \mi{responseType}$} \Comment{Implicit flow requires nonce.}
    \Let{$\mi{data}[\str{nonce}]$}{$\nu_7$}
  \EndIf
  \Let{$s'.\str{sessions}[\mi{sessionId}]$}{$s'.\str{sessions}[\mi{sessionId}] \cup \mi{data}$}
  \Let{$\mi{authEndpoint}$}{$\mi{oidcConfig}[\str{auth\_ep}]$}
  \Let{$\mi{authEndpoint}.\str{parameters}$}{$\mi{authEndpoint}.\str{parameters} \cup \mi{data}$}
  \Let{$\mi{headers}$}{$[\str{Location}{:}\mi{authEndpoint}, \str{ReferrerPolicy}{:}\str{origin}]$}
  \Let{$\mi{headers}[\cSetCookie]$}{$[\str{sessionId}{:}\an{\mi{sessionId},\top,\top,\top}]$}
  \Let{$\mi{request}$}{$s'.\str{sessions}[\mi{sessionId}][\str{startRequest}]$}
  \Let{$m'$}{$\encs{\an{\cHttpResp, \mi{request}[\str{message}].\str{nonce}, 303, \mi{headers}, \bot}}{\mi{request}[\str{key}]}$} 
  \Stop{$\an{\an{\mi{request}[\str{sender}],\mi{request}[\str{receiver}],m'}}$, $s'$} \label{line:rp-send-authorization-redir}
\EndFunction
\end{algorithmic} %
\end{algorithm}

\begin{definition}\label{def:relying-parties}
  A \emph{state $s\in Z^r$ of an RP $r$} is a term of the form
  $\langle\mi{DNSAddress}$, $\mi{pendingDNS}$, $\mi{pendingRequests}$,
  $\mi{corrupt}$, $\mi{keyMapping}$,
  $\mi{tlskeys}$, %
  $\mi{sessions}$, $\mi{issuerCache}$, $\mi{oidcConfigCache}$,
  $\mi{jwksCache}$, $\mi{clientCredentialsCache}\rangle$ with
  $\mi{DNSaddress} \in \addresses$,
  $\mi{pendingDNS} \in \dict{\nonces}{\terms}$,
  $\mi{pendingRequests} \in \dict{\nonces}{\terms}$,
  $\mi{corrupt}\in\terms$, $\mi{keyMapping} \in \dict{\dns}{\terms}$ ,
  $\mi{tlskeys} \in \dict{\dns}{K_\text{TLS}}$ (all former components
  as in Definition~\ref{def:generic-https-state}),
  $\mi{sessions} \in \dict{\nonces}{\terms}$,
  $\mi{issuerCache} \in \dict{\terms}{\terms}$,
  $\mi{oidcConfigCache} \in \dict{\terms}{\terms}$, and
  $\mi{jwksCache} \in \dict{\terms}{\terms}$.

  An \emph{initial state $s^r_0$ of $r$} is a state of $r$ with
  $s^r_0.\str{pendingDNS} \equiv \an{}$, $s^r_0.\str{pendingRequests} \equiv \an{}$, $s^r_0.\str{corrupt} \equiv \bot$,  $s^r_0.\str{keyMapping}$ being the
  same as the keymapping for browsers above, $s^r_0.\str{tlskeys} \equiv \mi{tlskeys}^r$, $s^r_0.\str{sessions} \equiv \an{}$,
  $s^r_0.\str{issuerCache} \equiv \an{}$,
  $s^r_0.\str{oidcConfigCache} \equiv \an{}$,
  $s^r_0.\str{jwksCache} \equiv \an{}$, and
  $s^r_0.\str{clientCredentialsCache} \equiv \an{}$.
\end{definition}

We now specify the relation $R^r$:
This relation is based on our model of generic HTTPS servers (see
Appendix~\ref{sec:generic-https-server-model}). Hence we only need to
specify algorithms that differ from or do not exist in the generic
server model. These algorithms are defined in
Algorithms~\ref{alg:rp-oidc-http-response}--\ref{alg:rp-oidc-cont-login-flow}.
(Note that in several places throughout these algorithms we use
placeholders to generate ``fresh'' nonces as described in our
communication model (see Definition~\ref{def:terms}).
Figure~\ref{fig:rp-placeholder-list} shows a list of all placeholders
used.)

The scripts that are used by the RP are described in
Algorithms~\ref{alg:script-rp-index}
and~\ref{alg:script-rp-get-fragment}. In these scripts, to extract the
current URL of a document, the function
$\mathsf{GETURL}(\mi{tree},\mi{docnonce})$
is used. We define this function as follows: It searches for the
document with the identifier $\mi{docnonce}$
in the (cleaned) tree $\mi{tree}$
of the browser's windows and documents. It then returns the URL $u$
of that document. If no document with nonce $\mi{docnonce}$
is found in the tree $\mi{tree}$, $\notdef$ is returned.

\begin{figure}[tbh]
  \centering
  \begin{tabular}{|@{\hspace{1ex}}l@{\hspace{1ex}}|@{\hspace{1ex}}l@{\hspace{1ex}}|}\hline 
    \hfill Placeholder\hfill  &\hfill  Usage\hfill  \\\hline\hline
    $\nu_1$ & new login session id  \\\hline
    $\nu_2$ & new HTTP request nonce  \\\hline
    $\nu_3$ & new HTTP request nonce  \\\hline
    $\nu_4$ & new service session id  \\\hline
    $\nu_5$ & new HTTP request nonce  \\\hline
    $\nu_6$ & new state value  \\\hline
    $\nu_{7}$ & new \emph{nonce} value (for the implicit flow)  \\\hline
  \end{tabular}
  
  \caption{List of placeholders used in the relying party algorithm.}
  \label{fig:rp-placeholder-list}
\end{figure}

\FloatBarrier
\begin{algorithm}[h!]
\caption{\label{alg:script-rp-index}Relation of $\mi{script\_rp\_index}$ }
\begin{algorithmic}[1]
\Statex[-1] \textbf{Input:} $\langle\mi{tree}$, $\mi{docnonce}$, $\mi{scriptstate}$, $\mi{scriptinputs}$, $\mi{cookies}$, $\mi{localStorage}$, $\mi{sessionStorage}$, $\mi{ids}$, $\mi{secrets}\rangle$
\Comment{\textbf{Script that models the index page of a relying party.} Users can initiate the login flow or follow arbitrary links. The script receives various information about the current browser state, filtered according to the access rules (same origin policy and others) in the browser.} 

\LetND{$\mi{switch}$}{$\{\str{auth},\str{link}\}$}
\ParboxComment{Non-deterministically decide whether to start a login flow or to follow some link.}

\If{$\mi{switch} \equiv \str{auth}$}
\ParboxComment{\textbf{Start login flow.}}

\Let{$\mi{url}$}{$\mathsf{GETURL}(\mi{tree},\mi{docnonce})$}
\ParboxComment{Retrieve own URL.}
\LetND{$\mi{id}$}{$\mi{ids}$}\label{line:script-rp-index-select-id} \ParboxComment{Retrieve one of user's identities.}
\Let{$\mi{url'}$}{$\an{\tUrl, \https, \mi{url}.\str{host}, \str{/startLogin},
    \an{}, \an{}}$}
\ParboxComment{Assemble URL.}
\Let{$\mi{command}$}{$\an{\tForm, \mi{url'}, \mi{\mPost}, \mi{id}, \bot}$}

\ParboxComment{Post a form including the identity to the  RP.}

\State \textbf{stop} $\an{s,\mi{cookies},\mi{localStorage},\mi{sessionStorage},\mi{command}}$ \label{line:script-rp-index-start-oidc-session} \ParboxComment{Finish script's run and instruct the browser to follow the command (form post).}

\Else  \ParboxComment{\textbf{Follow link.}}

  \LetND{$\mi{protocol}$}{$\{\http, \https\}$} 
  \ParboxComment{Non-deterministically select protocol (HTTP or HTTPS).}
  \LetND{$\mi{host}$}{$\dns$}
  \ParboxComment{Non-det. select host.}
  \LetND{$\mi{path}$}{$\mathbb{S}$}
  \ParboxComment{Non-det. select path.}
  \LetND{$\mi{fragment}$}{$\mathbb{S}$}
  \ParboxComment{Non-det. select fragment part.}
  \LetND{$\mi{parameters}$}{$\dict{\mathbb{S}}{\mathbb{S}}$} 
  \ParboxComment{Non-det. select parameters.}
  \Let{$\mi{url}$}{$\an{\cUrl, \mi{protocol}, \mi{host}, \mi{path}, \mi{parameters}, \mi{fragment}}$}
  \ParboxComment{Assemble URL.}

  \Let{$\mi{command}$}{$\an{\tHref, \mi{url}, \bot, \bot}$}
  \ParboxComment{Follow link to the selected URL.}
  \State \textbf{stop} $\an{s,\mi{cookies},\mi{localStorage},\mi{sessionStorage},\mi{command}}$
\ParboxComment{Finish script's run and instruct the browser to follow the command (follow link).}
\EndIf

\end{algorithmic} %
\end{algorithm}

\begin{algorithm}[h!]
\caption{Relation of $\mi{script\_rp\_get\_fragment}$ }\label{alg:script-rp-get-fragment}
\begin{algorithmic}[1]
\Statex[-1] \textbf{Input:} $\langle\mi{tree}$, $\mi{docnonce}$, $\mi{scriptstate}$, $\mi{scriptinputs}$, $\mi{cookies}$, $\mi{localStorage}$, $\mi{sessionStorage}$, $\mi{ids}$, $\mi{secrets}\rangle$
\Let{$\mi{url}$}{$\mathsf{GETURL}(\mi{tree},\mi{docnonce})$}
\Let{$\mi{url'}$}{$\an{\tUrl, \https, \mi{url}.\str{host}, \str{/redirect\_ep},
    [\str{iss} : \mi{url}.\str{parameters}[\str{iss}]], \an{}}$}
\Let{$\mi{command}$}{$\an{\tForm, \mi{url'}, \mi{\mPost}, \mi{url}.\str{fragment}, \bot}$}

\State \textbf{stop} $\an{s,\mi{cookies},\mi{localStorage},\mi{sessionStorage},\mi{command}}$

\end{algorithmic} %
\end{algorithm}

\FloatBarrier

\subsection{Identity Providers}  \label{app:idps}

An identity provider $i \in \fAP{IDP}$ is a web server modeled as
an atomic process $(I^i, Z^i, R^i, s_0^i)$ with the addresses
$I^i := \mapAddresstoAP(i)$. Next, we define the set $Z^i$ of states
of $i$ and the initial state $s^i_0$ of~$i$.

\begin{definition}\label{def:initial-state-idp}
  A \emph{state $s\in Z^i$ of an IdP $i$} is a term of the form
  $\langle\mi{DNSAddress}$, $\mi{pendingDNS}$, $\mi{pendingRequests}$,
  $\mi{corrupt}$, $\mi{keyMapping}$,
  $\mi{tlskeys}$, %
  $\mi{registrationRequests}$, (sequence of terms) $\mi{clients}$,
  (dict from nonces to terms) $\mi{records}$, (sequence of terms)
  $\mi{jwk}\rangle$ (signing key (only one)) with
  $\mi{DNSaddress} \in \addresses$,
  $\mi{pendingDNS} \in \dict{\nonces}{\terms}$,
  $\mi{pendingRequests} \in \dict{\nonces}{\terms}$,
  $\mi{corrupt}\in\terms$, $\mi{keyMapping} \in \dict{\dns}{\terms}$ ,
  $\mi{tlskeys} \in \dict{\dns}{K_\text{TLS}}$ (all former components
  as in Definition~\ref{def:generic-https-state}),
  $\mi{registrationRequests}\in\terms$,
  $\mi{clients}\in\dict{\terms}{\terms}$, $\mi{records}\in\terms$, and
  $\mi{jwk}\in K_\text{sign}$.

  An \emph{initial state $s^i_0$ of $i$} is a state of $i$ with
  $s^i_0.\str{pendingDNS} \equiv \an{}$,
  $s^i_0.\str{pendingRequests} \equiv \an{}$,
  $s^i_0.\str{corrupt} \equiv \bot$, $s^i_0.\str{keyMapping}$ being
  the same as the keymapping for browsers above,
  $s^i_0.\str{tlskeys} \equiv \mi{tlskeys}^i$,
  $s^i_0.\str{registrationRequests} \equiv \an{}$,
  $s^i_0.\str{clients} \equiv {an}$,
  $s^i_0.\str{records} \equiv \an{}$, and
  $s^i_0.\str{jwk} \equiv \mathsf{signkey}(i)$.
\end{definition}

We now specify the relation $R^i$:
As for the RPs above, this relation is based on our model of generic
HTTPS servers (see Appendix~\ref{sec:generic-https-server-model}). We
specify algorithms that differ from or do not exist in the generic
server model in Algorithms~\ref{alg:idp-oidc}
and~\ref{alg:idp-oidc-other}. As above,
Figure~\ref{fig:idp-placeholder-list} shows a list of all placeholders
used. Algorithm~\ref{alg:oidc-script-idp-form} shows the script
$\mathit{script\_idp\_form}$ that is used by IdPs.

\begin{figure}[htb]
  \centering
  \begin{tabular}{|@{\hspace{1ex}}l@{\hspace{1ex}}|@{\hspace{1ex}}l@{\hspace{1ex}}|}\hline 
    \hfill Placeholder\hfill  &\hfill  Usage\hfill  \\\hline\hline
    $\nu_1$ & new authorization code  \\\hline
    $\nu_2$, $\nu_3$ & new access tokens \\\hline
    $\nu_4$ & new client secret \\\hline
  \end{tabular}
  
  \caption{List of placeholders used in the identity provider algorithm.}
  \label{fig:idp-placeholder-list}
\end{figure}

\begin{algorithm}
\caption{\label{alg:idp-oidc} Relation of IdP $R^i$ -- Processing HTTPS Requests}
\begin{algorithmic}[1]
\Function{$\mathsf{PROCESS\_HTTPS\_REQUEST}$}{$m$, $k$, $a$, $f$, $s'$}
  \If{$m.\str{path} \equiv \str{/.wk/webfinger}$}
    \LetST{$\mi{user}$, $\mi{domain}$}{$\an{\mi{user},\mi{domain}} \equiv m.\str{parameters}[\str{resource}] \wedge \an{\mi{user}, \mi{domain}} \in \IDs^i$}{\textbf{stop}\ \DefStop} \label{line:idp-response-webfinger}
    \Let{$\mi{descriptor}$}{$\left[ \str{subject}: \an{\mi{user},\mi{domain}},  \str{links}:  \left[  \str{rel}: \str{OIDC\_issuer}, \str{href}: m.\str{host}   \right]   \right]$} \label{line:idp-webfinger-choose-issuer}
    \Let{$m'$}{$\encs{\an{\cHttpResp, m.\str{nonce}, 200, \an{}, \mi{descriptor}}}{k}$}
    \Stop{\StopWithMPrime} \label{line:idp-send-webfinger}
  \ElsIf{$m.\str{path} \equiv \str{/.wk/openid\mhyphen{}configuration}$}
    \Let{$\mi{metaData}$}{$\left[ \str{issuer}:  m.\str{host} \right]$}\label{line:idp-response-configuration}\footnote{Note that the metaData does not include the value $\str{response\_types\_supported}$ and $\str{subject\_types\_supported}$ since our model of RP assumes that all types are always supported.}
    \Let{$\mi{metaData}[\str{auth\_ep}]$}{$\an{\tUrl, \https, m.\str{host}, \str{/auth}, \an{}, \an{}}$}
    \Let{$\mi{metaData}[\str{token\_ep}]$}{$\an{\tUrl, \https, m.\str{host}, \str{/token}, \an{}, \an{}}$}
    \Let{$\mi{metaData}[\str{jwks\_uri}]$}{$\an{\tUrl, \https, m.\str{host}, \str{/jwks}, \an{}, \an{}}$}
    \Let{$\mi{metaData}[\str{reg\_ep}]$}{$\an{\tUrl, \https, m.\str{host}, \str{/reg}, \an{}, \an{}}$}
    \Let{$m'$}{$\encs{\an{\cHttpResp, m.\str{nonce}, 200, \an{}, \mi{metaData}}}{k}$}
    \Stop{\StopWithMPrime} \label{line:idp-send-oidc-config}
  \ElsIf{$m.\str{path} \equiv \str{/jwks}$} \label{line:idp-response-jwks}
    \Let{$m'$}{$\encs{\an{\cHttpResp, m.\str{nonce}, 201, \an{}, \mathsf{pub}(s'.\str{jwk})}}{k}$}
    \Stop{\StopWithMPrime} \label{line:idp-send-jwks}
  \ElsIf{$m.\str{path} \equiv \str{/reg} \wedge m.\str{method} \equiv \mPost$} \label{line:idp-process-reg-request}
    \Append{$\an{m,k,a,f}$}{$s'.\str{registrationRequests}$} 
    \Stop{\DefStop} \Comment{Stop here to let attacker choose the client id.}
  \ElsIf{$m.\str{path} \equiv \str{/auth}$}
    \If{$m.\str{method} \equiv \mGet$}
      \Let{$\mi{data}$}{$m.\str{parameters}$}
    \ElsIf{$m.\str{method} \equiv \mPost$}
      \Let{$\mi{data}$}{$m.\str{body}$}
    \EndIf
    \Let{$m'$}{$\encs{\an{\cHttpResp, m.\str{nonce}, 200, \an{\an{\str{ReferrerPolicy}, \str{origin}}}, \an{\str{script\_idp\_form}, \mi{data}}}}{k}$}
    \Stop{\StopWithMPrime}\label{line:idp-send-form}
  \ElsIf{$m.\str{path} \equiv \str{/auth2} \wedge m.\str{method} \equiv \mPost \wedge m.\str{headers}[\str{Origin}]  \equiv \an{m.\str{host}, \https}$}  \label{line:idp-receive-auth2}
    \Let{$\mi{identity}$}{$m.\str{body}[\str{identity}]$}
    \Let{$\mi{password}$}{$m.\str{body}[\str{password}]$}
    \If{$\mi{identity}.\str{domain} \not\in \mathsf{dom}(i)$}
      \Stop{\DefStop}
    \EndIf
    \If{$\mi{password} \not\equiv \mapIDtoPLI(\mi{identity})$}
      \Stop{\DefStop}
    \EndIf
    \Let{$\mi{responseType}$}{$m.\str{body}[\str{response\_type}]$}
    \Let{$\mi{clientId}$}{$m.\str{body}[\str{client\_id}]$}
    \Let{$\mi{redirectUri}$}{$m.\str{body}[\str{redirect\_uri}]$}
    \Let{$\mi{state}$}{$m.\str{body}[\str{state}]$}
    \Let{$\mi{nonce}$}{$m.\str{body}[\str{nonce}]$}

    \If{$\mi{clientId} \not\in s'.\str{clients}$}
      \Stop{\DefStop}
    \EndIf    
    \Let{$\mi{clientInfo}$}{$s'.\str{clients}[\mi{clientId}]$}
    \If{$\mi{redirectUri} \not\inPairing \mi{clientInfo}[\str{redirect\_uris}]$}
      \Stop{\DefStop}
    \EndIf

    \Let{$\mi{record}$}{$[ \str{client\_id} : \mi{clientId} ] $}
    \Let{$\mi{record}[\str{redirect\_uri}]$}{$\mi{redirectUri}$}
    \Let{$\mi{record}[\str{subject}]$}{$\mi{identity}$}
    \Let{$\mi{record}[\str{issuer}]$}{$m.\str{host}$}
    \Let{$\mi{record}[\str{nonce}]$}{$\mi{nonce}$}

    \Let{$\mi{record}[\str{code}]$}{$\nu_1$}
    \Let{$\mi{record}[\str{access\_tokens}]$}{$\an{\nu_2,\nu_3}$}\footnote{Note that an IdP may issue two different access tokens in hybrid mode: one access token along with the authorization response, and one access token along the token response.}

    \Append{$\mi{record}$}{$s'.\str{records}$}

\algstore{myalg}
\end{algorithmic}
\end{algorithm}

\begin{algorithm}                     
\begin{algorithmic} [1]                   %
\algrestore{myalg}

    \Let{$\mi{responseData}$}{$[]$}
    \If{$\str{code} \inPairing \mi{responseType}$}
      \Let{$\mi{responseData}[\str{code}]$}{$\nu_1$}
    \EndIf

    \If{$\str{token} \inPairing \mi{responseType}$}
      \Let{$\mi{responseData}[\str{access\_token}]$}{$\nu_2$}
      \Let{$\mi{responseData}[\str{token\_type}]$}{$\str{bearer}$}
    \EndIf

    \If{$\str{id\_token} \inPairing \mi{responseType}$}
      \Let{$\mi{idTokenBody}$}{$[ \str{iss}: \mi{record}[\str{issuer}], \str{sub} :\mi{record}[\str{subject}],$ \breakalgohook{9} $\str{aud} : \mi{record}[\str{client\_id}], \str{nonce} : \mi{record}[\str{nonce}]  ]$} \label{line:idp-create-id-token-auth2}

      \Let{$\mi{responseData}[\str{id\_token}]$}{$\sig{\mi{idTokenBody}}{s'.\str{jwk}}$} 

    \EndIf

    \If{$\mi{state} \not\equiv \an{}$}
      \Let{$\mi{responseData}[\str{state}]$}{$\mi{state}$}
    \EndIf
    \If{$\mi{responseType} \equiv \an{\str{code}}$} \Comment{Authorization Code Mode}
      \Let{$\mi{redirectUri}.\str{parameters}$}{$\mi{redirectUri}.\str{parameters} \cup \mi{responseData}$}
    \Else \Comment{Implicit/Hybrid Mode}
      \If{$\str{code} \not\inPairing \mi{responseType} \wedge \str{id\_token} \inPairing \mi{responseType} \wedge \mi{nonce} \equiv \an{}$}
        \Stop{\DefStop}\Comment{Nonce is required in implicit mode.}
      \EndIf
      \Let{$\mi{redirectUri}.\str{fragment}$}{$\mi{redirectUri}.\str{fragment} \cup \mi{responseData}$}
    \EndIf
    \Let{$\mi{redirectUri}.\str{parameters}[\str{iss}]$}{$\mi{record}[\str{issuer}]$} \label{line:idp-add-iss-param-to-redirecturi}
    \Let{$m'$}{$\encs{\an{\cHttpResp, m.\str{nonce}, 303, \an{\an{\str{Location}, \mi{redirectUri}}}, \an{}}}{k}$}
    \Stop{\StopWithMPrime} \label{line:idp-send-auth-resp}

  \ElsIf{$m.\str{path} \equiv \str{/token} \wedge m.\str{method} \equiv \mPost$}\label{line:idp-token-ep}
    \If{$\str{client\_id} \in m.\str{body}$} \Comment{Only client id is provided, no client secret.}
      \Let{$\mi{clientId}$}{$m.\str{body}[\str{client\_id}]$}
      \Let{$\mi{clientSecret}$}{$\an{}$}
    \Else
      \Let{$\mi{clientId}$}{$m.\str{headers}[\str{Authorization}].\str{username}$}
      \Let{$\mi{clientSecret}$}{$m.\str{headers}[\str{Authorization}].\str{password}$}
    \EndIf
    \Let{$\mi{clientInfo}$}{$s'.\str{clients}[\mi{clientId}]$}
    \If{$\mi{clientInfo} \equiv \an{} \vee \mi{clientInfo}[\str{client\_secret}] \not\equiv \mi{clientSecret}$}
      \Stop{\DefStop}
    \EndIf

    \Let{$\mi{code}$}{$m.\str{body}[\str{code}]$}
    \LetST{$\mi{record}$, $\mi{ptr}$}{$\mi{record} \equiv s'.\str{records}.\mi{ptr} \wedge \mi{record}[\str{code}] \equiv \mi{code} \wedge \mi{code} \not\equiv \bot$}{\textbf{stop}\ \DefStop}
    \If{$\mi{record}[\str{client\_id}] \not\equiv \mi{clientId}$} 
      \Stop{\DefStop}
    \EndIf
    \If{\textbf{not} $( \mi{record}[\str{redirect\_uri}] \equiv m.\str{body}[\str{redirect\_uri}] \vee (|\mi{clientInfo}[\str{redirect\_uris}]| = 1 \wedge \str{redirect\_uri} \not\in m.\str{body}))$}
      \Stop{\DefStop} \Comment{If only one redirect URI is registered, it can be omitted.}
    \EndIf

    \Let{$s'.\str{records}.\mi{ptr}[\str{code}]$}{$\bot$} \Comment{Invalidate code}

    \LetND{$\mi{accessTokenChoice}$}{$\{1,2\}$}
    \Let{$\mi{accessToken}$}{$\mi{record}[\str{access\_tokens}].\mi{accessTokenChoice}$}

    \Let{$\mi{idTokenBody}$}{$[ \str{iss}: \mi{record}[\str{issuer}] ]$} \label{line:idp-create-id-token-from-code}
    \Let{$\mi{idTokenBody}[sub]$}{$\mi{record}[\str{subject}]$}
    \Let{$\mi{idTokenBody}[aud]$}{$\mi{record}[\str{client\_id}]$}
    \Let{$\mi{idTokenBody}[nonce]$}{$\mi{record}[\str{nonce}]$}

    \Let{$\mi{id\_token}$}{$\sig{\mi{idTokenBody}}{s'.\str{jwk}}$}

    \Let{$m'$}{$\encs{\an{\cHttpResp, m.\str{nonce}, 200, \an{}, [\str{access\_token}{:} \mi{accessToken}, \str{token\_type}{:} \str{bearer}, \str{id\_token}{:} \mi{id\_token}]}}{k}$}      
    \Stop{\StopWithMPrime} \label{idp-send-token-ep}
  \EndIf
\EndFunction
\end{algorithmic} %
\end{algorithm}

\begin{algorithm}
\caption{\label{alg:idp-oidc-other} Relation of IdP $R^i$ -- Processing other messages.}
\begin{algorithmic}[1]
  \Function{$\mathsf{PROCESS\_OTHER}$}{$m$, $a$, $f$, $s'$} 
    \Let{$\mi{clientId}$}{$m$} \Comment{$m$ is client id chosen by and sent by an attacker process.}
    \If{$\mi{clientId} \in s'.\str{clients}$}
      \Stop{\DefStop}
    \EndIf
    \LetST{$m$, $k$, $a$, $f$}{$\an{m,k,a,f} \inPairing s'.\str{registrationRequests}$}{\textbf{stop}}
    \State \textbf{remove} $\an{m,k,a,f}$ \textbf{from} $s'.\str{registrationRequests}$
    \Let{$\mi{redirectUris}$}{$m.\str{body}[\str{redirect\_uris}]$}
    \Let{$\mi{regResponse}$}{$[ \str{client\_id}: \mi{clientId} ]$}
    \LetND{$\mi{issueSecret}$}{$\{\top,\bot\}$}
    \If{$\mi{issueSecret} \equiv \top$}
      \Let{$\mi{clientSecret}$}{$\nu_4$}
      \Let{$\mi{regResponse}[\str{client\_secret}]$}{$\mi{clientSecret}$}
    \EndIf
    \Let{$\mi{clientInfo}$}{$\mi{regResponse}$}
    \Let{$\mi{clientInfo}[\str{redirect\_uris}]$}{$\mi{redirectUris}$}
    \Let{$s'.\str{clients}[\mi{clientId}]$}{$\mi{clientInfo}$}
    \Let{$m'$}{$\encs{\an{\cHttpResp, m.\str{nonce}, 201, \an{}, \mi{regResponse}}}{k}$} 
    \Stop{\StopWithMPrime} \label{line:idp-send-reg-response}
\EndFunction
\end{algorithmic} %
\end{algorithm}

\begin{algorithm}
\caption{\label{alg:oidc-script-idp-form} Relation of $\mi{script\_idp\_form}$ }
\begin{algorithmic}[1]
\Statex[-1] \textbf{Input:} $\langle\mi{tree}$, $\mi{docnonce}$, $\mi{scriptstate}$, $\mi{scriptinputs}$, $\mi{cookies}$, $\mi{localStorage}$, $\mi{sessionStorage}$, $\mi{ids}$, $\mi{secrets}\rangle$
\Let{$\mi{url}$}{$\mathsf{GETURL}(\mi{tree},\mi{docnonce})$} 
\Let{$\mi{url}'$}{$\an{\tUrl, \https, \mi{url}.\str{host}, \str{/auth2},
    \an{}, \an{}}$}
\Let{$\mi{formData}$}{$\mi{scriptstate}$}
\LetND{$\mi{identity}$}{$\mi{ids}$} \label{line:script-idp-form-select-id}
\LetND{$\mi{secret}$}{$\mi{secrets}$}
\Let{$\mi{formData}[\str{identity}]$}{$\mi{identity}$}
\Let{$\mi{formData}[\str{password}]$}{$\mi{secret}$}
\Let{$\mi{command}$}{$\an{\tForm, \mi{url}', \mi{\mPost}, \mi{formData}, \bot}$}
\State \textbf{stop} $\an{s,\mi{cookies},\mi{localStorage},\mi{sessionStorage},\mi{command}}$

\end{algorithmic} %
\end{algorithm}

\FloatBarrier %

\section{Formal Model of OpenID Connect with Web Attackers}
\label{app:model-oidc-auth-webattackers}

We now derive $\oidcwebsystem^w$
(an \emph{OIDC web system with web attackers}) from $\oidcwebsystem^n$
by replacing the network attacker with a finite set of web attackers.
(Note that we more generally speak of an \emph{OIDC web system} if it
is not important what kind of attacker the web system contains.)

\begin{definition}
  An \emph{OIDC web system with web attackers}, $\oidcwebsystem^w$,
  is an OIDC web system
  $\oidcwebsystem^n=(\bidsystem, \scriptset, \mathsf{script}, E^0)$
  with the following changes:
  \begin{itemize}
  \item We have $\bidsystem = \mathsf{Hon} \cup \mathsf{Web}$,
    in particular, there is no network attacker. The set
    $\mathsf{Web}$
    contains a finite number of web attacker processes. The set
    $\mathsf{Hon}$
    is as described above, and additionally contains a DNS server $d$
    as defined below.
  \item The set of IP addresses $\mathsf{IPs}$ contains no IP addresses for the network
    attacker, but instead a finite set of IP addresses for each web
    attacker.
  \item The set of Domains $\mathsf{Doms}$
    contains no domains for the network attacker, but instead a finite
    set of domains for each web attacker.

  \item All honest parties use the DNS server $d$ as their DNS server.
  \end{itemize}
\end{definition}

\subsection{DNS Server}
\label{sec:dns-server-1}

The DNS server $d$
is a DNS server as defined in Definition~\ref{def:dns-server}. Its
initial state $s_0^d$
contains only pairings $\an{D, i}$
such that $i \in \mathsf{addr}(\mathsf{dom}^{-1}(D))$,
i.e., any domain is resolved to an IP address belonging to the owner
of that domain (as defined in Appendix~\ref{app:addresses-and-domain-names}).

\subsection{Web Attackers}
\label{sec:web-attackers}

Web attackers, as opposed to network attackers, can only use their own
IP addresses for listening to and sending messages. Therefore, for any
web attacker process $w$
we have that $I^w = \mathsf{addr}(w)$.
The inital states of web attackers are defined parallel to those of
network attackers, i.e., the initial state for a web attacker process
$w$
is $s_0^w = \an{\mi{attdoms}^w, \mi{tlskeys}, \mi{signkeys}}$,
where $\mi{attdoms}^w$
is a sequence of all domains along with the corresponding private keys
owned by the attacker $w$,
$\mi{tlskeys}$
is a sequence of all domains and the corresponding public keys, and
$\mi{signkeys}$
is a sequence containing all public signing keys for all IdPs.

\clearpage
\section{Formal Security Properties}\label{app:form-secur-prop}

The security properties for OIDC are formally defined as follows.

\subsection{Authentication}
\label{sec:fprop-authentication}

Intuitively, authentication for $\oidcwebsystem^n$
means that an attacker should not be able to login at an (honest) RP
under the identity of a user unless certain parties involved in the
login process are corrupted. As explained above, being logged in at an
RP under some user identity means to have obtained a service token for
this identity from the RP. 

\begin{definition}[Service Sessions]\label{def:service-sessions}
  We say that there is a \emph{service session identified
  by a nonce $n$
  for an identity $\mi{id}$
  at some RP $r$} in a configuration $(S, E, N)$
  of a run $\rho$
  of an OIDC web system
  iff there exists some session id $x$
  and a domain $d \in \mathsf{dom}(\mathsf{governor}(\mi{id}))$
  such that
  $S(r).\str{sessions}[x][\str{loggedInAs}] \equiv \an{d, \mi{id}}$
  and $S(r).\str{sessions}[x][\str{serviceSessionId}] \equiv n$.
\end{definition}

\begin{definition}[Authentication Property]\label{def:property-authn-a} Let $\oidcwebsystem^n$ be an OIDC web
  system with a network attacker. We say that \emph{$\oidcwebsystem^n$
    is secure w.r.t.~authentication} iff for every run $\rho$
  of $\oidcwebsystem^n$,
  every configuration $(S, E, N)$
  in $\rho$,
  every $r\in \fAP{RP}$
  that is honest in $S$,
  every browser $b$
  that is honest in $S$,
  every identity $\mi{id} \in \mathsf{ID}$
  with $\mathsf{governor}(\mi{id})$
  being an honest IdP, every service session identified by some nonce
  $n$
  for $\mi{id}$
  at $r$,
  $n$
  is not derivable from the attackers knowledge in $S$
  (i.e., $n \not\in d_{\emptyset}(S(\fAP{attacker}))$).
\end{definition}

\subsection{Authorization}
\label{sec:fprop-authorization}

Intuitively, authorization for $\oidcwebsystem^n$
means that an attacker should not be able to obtain or use a protected
resource available to some honest RP at an IdP for some user unless
certain parties involved in the authorization process are corrupted.

\begin{definition}\label{def:client-id-issued}
  We say that a client id $c$
  \emph{has been issued to $r$
    by $i$}
  iff $i$ has
  sent a response to a registration request from $r$
  in Line~\ref{line:idp-send-reg-response} of
  Algorithm~\ref{alg:idp-oidc-other} and this response contains $c$
  in its  body under the dictionary key
  $\str{client\_id}$.
\end{definition}

\begin{definition}[Authorization Property]\label{def:property-authz-a} 

  Let $\oidcwebsystem^n$
  be an OIDC web system with a network attacker. We say that
  \emph{$\oidcwebsystem^n$
    is secure w.r.t.~authorization} iff for every run $\rho$
  of $\oidcwebsystem^n$,
  every configuration $(S, E, N)$
  in $\rho$,
  every $r\in \fAP{RP}$
  that is honest in $S$,
  every $i\in \fAP{IdP}$
  that is honest in $S$,
  every browser $b$ that is honest in $S$,
  every identity $\mi{id} \in \mathsf{ID}^i$ owned by $b$,
  every nonce $n$, every term $x \inPairing S(i).\str{records}$
  with $x[\str{subject}] \equiv \mi{id}$,
  $n \inPairing x[\str{access\_tokens}]$,
  and the client id $x[\str{client\_id}]$ has been issued by $i$ to $r$, we have that $n$
  is not derivable from the attackers knowledge in $S$
  (i.e., $n \not\in d_{\emptyset}(S(\fAP{attacker}))$).
\end{definition}

\subsection{Session Integrity for Authentication and Authorization}
\label{sec:fprop-session-integrity}

The two session integrity properties capture that an attacker should
be unable to forcefully log a user in to some RP. This includes
attacks such as CSRF and session swapping.

\subsubsection{Session Integrity Property for Authentication}
This security property captures that (a) a user should only be logged
in when the user actually expressed the wish to start an OIDC flow
before, and (b) if a user expressed the wish to start an OIDC flow
using some honest identity provider and a specific identity, then user
is not logged in under a different identity.

We first need to define notations for the processing steps that
represent important events during a flow of an OIDC web system.

\begin{definition}[User is logged in]\label{def:user-logged-in}
  For a run $\rho$
  of an OIDC web system with web attacker $\oidcwebsystem^w$
  we say that a browser $b$
  was authenticated to an RP $r$
  using an IdP $i$
  and an identity $u$
  in a login session identified by a nonce $\mi{lsid}$
  in processing step $Q$ in $\rho$ with
  $$Q = (S, E, N) \xrightarrow[r \rightarrow  
  E_{\text{out}}]{} (S', E', N')$$  (for some $S$, $S'$, $E$, $E'$, $N$, $N'$)
  and some event $\an{y,y',m} \in E_{\text{out}}$
  such that $m$
  is an HTTPS response matching an HTTPS request sent by $b$
  to $r$
  and we have that in the headers of $m$
  there is a header of the form
  $\an{\str{Set\mhyphen{}Cookie},
    [\str{serviceSessionId}{:}\an{\mi{ssid},\top,\top,\top}]}$ for
  some nonce $\mi{ssid}$
  and we have that there is a term $g$
  such that $S(r).\str{sessions}[\mi{lsid}] \equiv g$,
  $g[\str{serviceSessionId}] \equiv \mi{ssid}$,
  and $g[\str{loggedInAs}] \equiv \an{d, u}$
  with $d \in \mathsf{dom}(i)$.
  We then write $\mathsf{loggedIn}_\rho^Q(b, r, u, i, \mi{lsid})$.
\end{definition}

\begin{definition}[User started a login flow]
  For a run $\rho$
  of an OIDC web system with web attacker $\oidcwebsystem^w$
  we say that the user of the browser $b$
  started a login session identified by a nonce $\mi{lsid}$
  at the RP $r$
  in a processing step $Q$
  in $\rho$
  if (1) in that processing step, the browser $b$
  was triggered, selected a document loaded from an origin of $r$,
  executed the script $\mi{script\_rp\_index}$
  in that document, and in that script, executed the
  Line~\ref{line:script-rp-index-start-oidc-session} of
  Algorithm~\ref{alg:script-rp-index}, and (2) $r$
  sends an HTTPS response corresponding to the HTTPS request sent by
  $b$
  in $Q$
  and in that response, there is a header of the form
  $\an{\str{Set\mhyphen{}Cookie},
    [\str{sessionId}{:}\an{\mi{lsid},\top,\top,\top}]}$. We then write
  $\mathsf{started}_\rho^Q(b, r, \mi{lsid})$.
\end{definition}

\begin{definition}[User authenticated at an IdP]
  For a run $\rho$
  of an OIDC web system with web attacker $\oidcwebsystem^w$
  we say that the user of the browser $b$
  authenticated to an IdP $i$
  using an identity $u$
  for a login session identified by a nonce $\mi{lsid}$
  at the RP $r$
  if there is a processing step $Q$
  in $\rho$   with
  $$Q = (S, E, N) \xrightarrow[]{}
  (S', E', N')$$ (for some $S$, $S'$, $E$, $E'$, $N$, $N'$) in which the browser $b$
  was triggered, selected a document loaded from an origin of $i$,
  executed the script $\mi{script\_idp\_form}$
  in that document, and in that script, (1) in
  Line~\ref{line:script-idp-form-select-id} of
  Algorithm~\ref{alg:oidc-script-idp-form}, selected the identity $u$,
  and (2) we have that the $\mi{scriptstate}$
  of that document, when triggered, contains a nonce $s$
  such that $\mi{scriptstate}[\str{state}] \equiv s$
  and $S(r).\str{sessions}[\mi{lsid}][\str{state}] \equiv s$.
  We then write
  $\mathsf{authenticated}_\rho^Q(b, r, u, i, \mi{lsid})$.
\end{definition}

\begin{definition}[RP uses an access token]\label{def:rp-uses-access-token}
  For a run $\rho$
  of an OIDC web system with web attacker $\oidcwebsystem^w$
  we say that the RP $r$
  uses some access token $t$
  in a login session identified by the nonce $\mi{lsid}$
  established with the browser $b$
  at an IdP $i$ if there is a processing step $Q$ in $\rho$ with
  $$Q = (S, E, N) \xrightarrow[]{}
  (S', E', N')$$  (for some $S$, $S'$, $E$, $E'$, $N$, $N'$) in which (1) $r$
  calls the function $\mathsf{USE\_ACCESS\_TOKEN}$
  with the first two parameters being $\mi{lsid}$
  and $t$,
  (2)
  $S(r).\str{issuerCache}[S(r).\str{sessions}[\mi{lsid}][\str{identity}]]
  \in \mathsf{dom}(i)$, and (3)
  $\an{\str{sessionid}, \an{\mi{lsid}, y, z, z'}} \inPairing S(b).\str{cookies}[d]$
  for $d \in \mathsf{dom}(r)$,
  $y, z, z' \in \terms$.
  We then write $\mathsf{usedAuthorization}_\rho^Q(b, r, i, \mi{lsid})$.
\end{definition}

\begin{definition}[RP acts on the user's behalf]\label{def:rp-acts-on-users-behalf}
  For a run $\rho$
  of an OIDC web system with web attacker $\oidcwebsystem^w$
  we say that the RP $r$
  acts on behalf of the user with the identity $u$
  at an honest IdP $i$
  in a login session identified by the nonce $\mi{lsid}$
  established with the browser $b$
  if there is a processing step $Q$ in $\rho$ with
  $$Q = (S, E, N) \xrightarrow[]{}
  (S', E', N')$$  (for some $S$, $S'$, $E$, $E'$, $N$, $N'$) in which (1) $r$
  calls the function $\mathsf{USE\_ACCESS\_TOKEN}$
  with the first two parameters being $\mi{lsid}$
  and $t$,
  (2) we have that there is a term $g$
  such that $g \inPairing S(i).\str{records}$
  with $t \inPairing g[\str{access\_tokens}]$
  and $g[\str{subject}] \equiv u$,
  and (3)
  $\an{\str{sessionid}, \an{\mi{lsid}, y, z, z'}} \inPairing
  S(b).\str{cookies}[d]$ for $d \in \mathsf{dom}(r)$,
  $y, z, z' \in \terms$.
  We then write
  $\mathsf{actsOnUsersBehalf}_\rho^Q(b, r, u, i, \mi{lsid})$.
\end{definition}

For session integrity for authentication we say that a user that is
logged in at some RP must have expressed her wish to be logged in to
that RP in the beginning of the login flow. If the IdP is honest, then the user must also have
authenticated herself at the IdP with the same user account that RP
uses for her identification. This excludes, for example, cases where
(1) the user is forcefully logged in to an RP by an attacker that
plays the role of an IdP, and (2) where an attacker can force a user
to be logged in at some RP under a false identity issued by an honest
IdP.

\begin{definition} [Session Integrity for
  Authentication]\label{def:property-si-authn}
  Let $\oidcwebsystem^w$
  be an OIDC web system with web attackers. We say that
  \emph{$\oidcwebsystem^w$
    is secure w.r.t.~session integrity for authentication} iff for
  every run $\rho$
  of $\oidcwebsystem^w$, every processing step $Q$ in $\rho$ with
  $$Q = (S, E, N) \xrightarrow[]{}
  (S', E', N')$$ (for some $S$,
  $S'$,
  $E$,
  $E'$,
  $N$,
  $N'$),
  every browser $b$
  that is honest in $S$,
  every $i\in \fAP{IdP}$,
  every identity $u$
  that is owned by $b$,
  every $r\in \fAP{RP}$
  that is honest in $S$,
  every nonce $\mi{lsid}$,
  and $\mathsf{loggedIn}_\rho^Q(b, r, u, i, \mi{lsid})$
  we have that (1) there exists a processing step $Q'$
  in $\rho$
  (before $Q$)
  such that $\mathsf{started}_\rho^{Q'}(b, r, \mi{lsid})$,
  and (2) if $i$
  is honest in $S$,
  then there exists a processing step $Q''$
  in $\rho$
  (before $Q$)
  such that
  $\mathsf{authenticated}_\rho^{Q''}(b, r, u, i, \mi{lsid})$.
\end{definition}

For session integrity for authorization we say that if an RP uses some
access token at some IdP in a session with a user, then that user
expressed her wish to authorize the RP to interact with some IdP. If
the IdP is honest, and the RP acts on the user's behalf at the IdP
(i.e., the access token is bound to the user's identity), then the
user authenticated to the IdP using that identity.

\begin{definition} [Session Integrity for
  Authorization]\label{def:property-si-authz}
  Let $\oidcwebsystem^w$
  be an OIDC web system with web attackers. We say that
  \emph{$\oidcwebsystem^w$
    is secure w.r.t.~session integrity for authentication} iff for
  every run $\rho$
  of $\oidcwebsystem^w$, every processing step $Q$ in $\rho$ with
  $$Q = (S, E, N) \xrightarrow[]{}
  (S', E', N')$$  (for some $S$, $S'$, $E$, $E'$, $N$, $N'$), every browser $b$
  that is honest in $S$,
  every $i\in \fAP{IdP}$,
  every identity $u$
  that is owned by $b$,
  every $r\in \fAP{RP}$
  that is honest in $S$,
  every nonce $\mi{lsid}$,
  we have that (1) if
  $\mathsf{usedAuthorization}_\rho^Q(b, r, i, \mi{lsid})$
  then there exists a processing step $Q'$
  in $\rho$
  (before $Q$)
  such that $\mathsf{started}_\rho^{Q'}(b, r, \mi{lsid})$,
  and (2) if $i$
  is honest in $S$
  and $\mathsf{actsOnUsersBehalf}_\rho^Q(b, r, u, i, \mi{lsid})$
  then there exists a processing step $Q''$
  in $\rho$
  (before $Q$)
  such that
  $\mathsf{authenticated}_\rho^{Q''}(b, r, u, i, \mi{lsid})$.
\end{definition}

\clearpage
\section{Proof of Theorem~\ref{thm:theorem-1}}
\label{app:proof-oidc}

Before we prove Theorem~\ref{thm:theorem-1}, in order to provide a
quick overview, we first provide a proof sketch. We then show some
general properties of OIDC web systems with a network attacker, and
then proceed to prove the authentication, authorization, and session
integrity properties separately.

\subsection{Proof Sketch}
For authentication and authorization, we first show that the secondary
security properties from Section~\ref{sec:second-secur-prop} hold true
(see
Lemmas~\ref{lemma:issuer-cache-integrity}--\ref{lemma:attacker-does-not-learn-id-tokens}
below). We then assume that the authentication/authorization
properties do not hold, i.e., that there is a run $\rho$
of $\oidcwebsystem^n$
that does not satisfy authentication or authorization, respectively.
Using
Lemmas~\ref{lemma:issuer-cache-integrity}--\ref{lemma:attacker-does-not-learn-id-tokens},
it then only requires a few steps to lead the respective assumption to
a contradication and thereby show that $\oidcwebsystem^n$
enjoys authentication/authorization.

For the session integrity properties, we follow a similar scheme. We
first show Lemma~\ref{lemma:attacker-does-not-learn-state},
which essentially says that a web attacker is unable to get hold of
the $\mi{state}$
value that is used in a session between an honest browser $b$,
an honest RP $r$,
and an honest IdP $i$.
(Recall that the $\mi{state}$
value is essential for session integrity.) We then show session
integrity for authentication/authorization by starting from the latest
``known'' processing steps in the respective flows (e.g., for
authentication, $\mathsf{loggedIn}_\rho^Q(b, r, u, i, \mi{lsid})$)
and tracking through the OIDC flows to show the existence of the
earlier processing steps (e.g.,
$\mathsf{started}_\rho^{Q'}(b, r, \mi{lsid})$)
and their respective properties.

\subsection{Properties of $\oidcwebsystem^n$}

Let $\oidcwebsystem^n = (\bidsystem, \scriptset, \mathsf{script}, E^0)$
be an OIDC web system with a network attacker. Let $\rho$ be a run of $\oidcwebsystem^n$. We write
$s_x = (S^x,E^x,N^x)$ for the states in $\rho$.

\begin{definition}\label{def:contains}
  We say that a term $t$ \emph{is derivably contained in (a term) $t'$
    for (a set of DY processes) $P$ (in a processing step $s_i
    \rightarrow s_{i+1}$ of a run $\rho=(s_0,s_1,\ldots)$)} if $t$ is
  derivable from $t'$ with the knowledge available to $P$, i.e.,
  \begin{align*}
    t \in d_{\emptyset}(\{t'\} \cup \bigcup_{p\in P}S^{i+1}(p))
  \end{align*}

\end{definition}

\begin{definition}\label{def:leak}
  We say that \emph{a set of processes $P$ leaks a term $t$ (in a
    processing step $s_i \rightarrow s_{i+1}$) to a set of processes
    $P'$} if there exists a message $m$ that is emitted (in $s_i
  \rightarrow s_{i+1}$) by some $p \in P$ and $t$ is derivably
  contained in $m$ for $P'$ in the processing step $s_i \rightarrow
  s_{i+1}$. If we omit $P'$, we define $P' := \bidsystem \setminus
  P$. If $P$ is a set with a single element, we omit the set notation.
\end{definition}

\begin{definition}\label{def:creating}
  We say that an DY process $p$ \emph{created} a message $m$ (at
  some point) in a run if $m$ is derivably contained in a message
  emitted by $p$ in some processing step and if there is no earlier
  processing step where $m$ is derivably contained in a message
  emitted by some DY process $p'$.
\end{definition}

\begin{definition}\label{def:accepting}
  We say that \emph{a browser $b$ accepted} a message (as a response
  to some request) if the browser decrypted the message (if it was an
  HTTPS message) and called the function $\mathsf{PROCESSRESPONSE}$,
  passing the message and the request (see
  Algorithm~\ref{alg:processresponse}).
\end{definition}

\begin{definition}\label{def:knowing}
  We say that an atomic DY process \emph{$p$ knows a term $t$} in some
  state $s=(S,E,N)$ of a run if it can derive the term from its
  knowledge, i.e., $t \in d_{\emptyset}(S(p))$.
\end{definition}

\begin{definition}\label{def:initiating}
  We say that a \emph{script initiated a request $r$} if a browser
  triggered the script (in Line~\ref{line:trigger-script} of
  Algorithm~\ref{alg:runscript}) and the first component of the
  $\mi{command}$ output of the script relation is either $\tHref$, 
  $\tIframe$, $\tForm$, or $\tXMLHTTPRequest$ such that the browser
  issues the request $r$ in the same step as a result.
\end{definition}

The following lemma captures properties of RP when it uses HTTPS. For
example, the lemma says that other parties cannot decrypt messages
encrypted by RP.

\subsection{Proof of Authentication}
\label{sec:proof-property-a}

We here want to show that every OIDC web system is secure
w.r.t.~authentication, and therefore assume that there exists an
OIDC web system that is not secure w.r.t.~authentication. We then
lead this to a contradiction, thereby showing that all OIDC web
systems are secure w.r.t.~authentication. In detail, we assume:

\begin{lemma}[Integrity of Issuer Cache]\label{lemma:issuer-cache-integrity}
  For any run $\rho$
  of an OIDC web system $\oidcwebsystem^n$
  with a network attacker or an OIDC web system $\oidcwebsystem^w$
  with web attackers, every configuration $(S, E, N)$
  in $\rho$,
  every IdP $i$
  that is honest in $S$,
  every identity $\mi{id} \in \mathsf{ID}^i$,
  every relying party $r$
  that is honest in $S$,
  we have that $S(r).\str{issuerCache}[\mi{id}] \equiv \an{}$
  (not set) or
  $S(r).\str{issuerCache}[\mi{id}] \in \mathsf{dom}(i)$.
\end{lemma}

\begin{proof}
  Initially, the issuer cache of an honest relying party is empty (according to Definition~\ref{def:relying-parties}). This issuer cache can only be modified in
  Line~\ref{line:rp-set-issuercache} of
  Algorithm~\ref{alg:rp-oidc-http-response}. There, the value of
  $S^l(r).\str{issuerCache}[\mi{id}']$
  (for some $l<j$)
  is taken from an HTTPS response. The value of $\mi{id}'$
  is taken from session data (Line~\ref{line:rp-lookup-session}) which
  is identified by a session id that is taken from the internal
  reference data of the incoming message. This internal reference data
  must have been created previously in Algorithm~\ref{alg:simple-send}
  ($\mathsf{HTTPS\_SIMPLE\_SEND}$)
  which must have been called in
  Line~\ref{line:rp-create-webfinger-request} of
  Algorithm~\ref{alg:rp-oidc-cont-login-flow} (since this is the only
  place where the reference data for a webfinger request is created).
  In this algorithm, it is easy to see that the request to which the
  request is sent (see Line~\ref{line:rp-webfinger-domain}) is the
  domain part of the identity. We therefore have that a webfinger
  request must have been sent (using HTTPS) to the IdP $i$.
  (Note that an attacker can neither decrypt any information from this
  request, nor spoof a response to this request. The request must
  therefore have been responded to by the honest IdP. \gs{add reference to lemma attacker cannot spoof or decrypt https messages})

  Since the path of this request is $\str{/.wk/webfinger}$,
  the IdP can respond to this request only in
  Lines~\ref{line:idp-response-webfinger}ff.~of
  Algorithm~\ref{alg:idp-oidc}. Since the IdP there chooses an issuer
  value that is one of its own domains (see
  Line~\ref{line:idp-webfinger-choose-issuer}), we finally have that
  $S(r).\str{issuerCache}[\mi{id}] \equiv \an{}$
  (if the response is blocked or the webfinger request was never sent)
  or we have that
  $S(r).\str{issuerCache}[\mi{id}] \in \mathsf{dom}(i)$,
  which proves the lemma.
\end{proof}

\begin{lemma}[Integrity of oidcConfigCache] \label{lemma:oidc-config-cache-integrity}
  For any run $\rho$
  of an OIDC web system $\oidcwebsystem^n$
  with a network attacker or an OIDC web system $\oidcwebsystem^w$
  with web attackers, every configuration $(S, E, N)$
  in $\rho$,
  every IdP $i$
  that is honest in $S$,
  every domain $d \in \mathsf{dom}(i)$,
  every relying party $r$
  that is honest in $S$,
  $l \in \{1,2,3,4\}$
  we have that $S(r).\str{oidcConfigCache}[d] \equiv \an{}$
  (not set) or
  $S(r).\str{oidcConfigCache}[d] \equiv [ \str{issuer}: d,
  \str{auth\_ep}: u_1, \str{token\_ep}: u_2, \str{jwks\_ep}: u_3,
  \str{reg\_ep}: u_4]$ with $u_l$
  being URLs, $u_l.\str{host} \in \mathsf{dom}(i)$,
  and $u_l.\str{protocol} \equiv \https$.
\end{lemma}

\begin{proof}
  This proof proceeds analog to the one for
  Lemma~\ref{lemma:issuer-cache-integrity} with the following changes:
  First, the OIDC configuration cache is filled only in
  Line~\ref{line:rp-set-oidc-config-cache} of
  Algorithm~\ref{alg:rp-oidc-http-response}. It requires a request
  that was created in Line~\ref{line:rp-create-oidc-config-request} of
  Algorithm~\ref{alg:rp-oidc-cont-login-flow}. This request was not
  sent to the domain contained in an ID (as above) but instead to the
  issuer (in this case, $d$).
  The issuer responds to this request in
  Lines~\ref{line:idp-response-configuration}ff.~of
  Algorithm~\ref{alg:idp-oidc}. There, the issuer only choses the
  redirection endpoint URIs such that the host is the domain of the
  incoming request and the protocol is HTTPS ($\https$).
  This proves the lemma.
\end{proof}

\begin{lemma}[Integrity of JWKS Cache]\label{lemma:oidc-jwks-cache-integrity}
  For any run $\rho$ of an OIDC web system $\oidcwebsystem^n$ with a
  network attacker or an OIDC web system $\oidcwebsystem^w$ with web
  attackers, every configuration $(S, E, N)$ in $\rho$, every IdP $i$
  that is honest in $S$, every domain $d \in \mathsf{dom}(i)$, every relying party $r$ that
  is honest in $S$, we have that
  $S(r).\str{jwksCache}[d] \equiv \an{}$ (not set) or
  $S(r).\str{jwksCache}[d] \equiv \mathsf{pub}(S(i).\str{jwks})$.
\end{lemma}

\begin{proof}
  This proof proceeds analog to the one for
  Lemma~\ref{lemma:oidc-config-cache-integrity}. The relevant HTTPS
  request by $r$
  is created in Line~\ref{line:rp-create-jwks-request} of
  Algorithm~\ref{alg:rp-oidc-cont-login-flow}, and responded to by
  the IdP $i$
  in Lines~\ref{line:idp-response-jwks}ff.~of
  Algorithm~\ref{alg:idp-oidc}. There, the IdP chooses its own
  signature verification key to send in the response. This proves the
  lemma.
\end{proof}

\begin{lemma}[Integrity of Client Registration]\label{lemma:client-reg-integrity}
  For any run $\rho$
  of an OIDC web system $\oidcwebsystem^n$
  with a network attacker or an OIDC web system $\oidcwebsystem^w$
  with web attackers, every configuration $(S, E, N)$
  in $\rho$,
  every IdP $i$
  that is honest in $S$,
  every domain $d \in \mathsf{dom}(i)$,
  every relying party $r$
  that is honest in $S$,
  every client id $c$ that has been issued to $r$ by $i$,
  every URL $u \inPairing S(i).\str{clients}[c][\str{redirect\_uris}]$
  we have that $u.\str{host} \in \mathsf{dom}(r)$
  and $u.\str{protocol} \equiv \https$.

\end{lemma}

\begin{proof}
  From Definition~\ref{def:client-id-issued} it follows that an HTTPS
  request must have been sent from $r$
  to $i$
  in Lines~\ref{line:rp-create-registration-request}ff.~of
  Algorithm~\ref{alg:rp-oidc-cont-login-flow}. This request must have
  been processed by $i$
  in Lines~\ref{line:idp-process-reg-request}ff.~of
  Algorithm~\ref{alg:idp-oidc}, and, after receiving the client id
  from some other party (usually the attacker), in
  Algorithm~\ref{alg:idp-oidc-other}. From the latter algorithm it is
  easy to see that the redirection endpoint data must have been taken
  from $r$'s
  initial registration request to create the dictionary stored in
  $S(i).\str{clients}[c]$.
  This data, however, was chosen by $r$
  in Line~\ref{line:rp-choose-redirect-uris} of
  Algorithm~\ref{alg:rp-oidc-cont-login-flow} such that
  $u.\str{host} \in \mathsf{dom}(r)$
  and $u.\str{protocol} \equiv \https$
  for every
  $u \inPairing S(i).\str{clients}[c][\str{redirect\_uris}]$.
\end{proof}

\begin{lemma}[Other parties do not learn passwords]\label{lemma:attacker-does-not-learn-password}
  For any run $\rho$
  of an OIDC web system $\oidcwebsystem^n$
  with a network attacker or an OIDC web system $\oidcwebsystem^w$
  with web attackers, every configuration $(S, E, N)$
  in $\rho$,
  every IdP $i$
  that is honest in $S$,
  every identity $\mi{id} \in \mathsf{ID}^i$,
  every browser $b$
  with $b = \mathsf{ownerOfID}(\mi{id})$
  that is honest in $S$,
  every $p \in \websystem \setminus \{b, i\}$
  we have that
  $\mathsf{secretOfID}(\mi{id}) \not\in d_\emptyset(S^l(p))$.
\end{lemma}

\begin{proof}
  Let $s := \mathsf{secretOfID}(\mi{id})$.
  Initially, in $S^0$,
  $s$
  is only contained in $S^0(b).\str{secrets}[\an{d,\https}]$
  with $d \in \mathsf{dom}(i)$
  and in no other states of any atomic processes (or in any waiting
  events). By the definition of the browser, we can see that only
  scripts loaded from the origins $\an{d,\https}$
  can access $s$.
  We know that $i$
  is an honest IdP. Now, the only script that an honest IdP sends to
  the browser is $\mi{script\_idp\_form}$.
  This scripts sends the form data only to its own origin, which
  means, that the form data is sent over HTTPS and to the honest IdP.
  In this request, the script uses the path $\str{/auth2}$.
  There, identity and password are checked, but not used otherwise.
  Therefore, the form data cannot leak from the honest IdP. It could,
  however, leak from the browser itself. The form data is sent via
  POST, and therefore, not used in any referer headers. A redirection
  response from the server contains the status code 303, which implies
  that the browser does not send the form data again when following
  the redirection. Since there are also no other scripts from the same
  origin running in the browser which could access the form data, the
  password $s$
  cannot leak from the browser either. This proves
  Lemma~\ref{lemma:attacker-does-not-learn-password}. \qed 
\end{proof}

\begin{lemma}[Attacker does not Learn ID Tokens]\label{lemma:attacker-does-not-learn-id-tokens}
  For any run $\rho$
  of an OIDC web system $\oidcwebsystem^n$
  with a network attacker or an OIDC web system $\oidcwebsystem^w$
  with web attackers, every configuration $(S, E, N)$
  in $\rho$,
  every IdP $i$
  that is honest in $S$,
  every domain $d \in \mathsf{dom}(i)$,
  every identity $\mi{id} \in \mathsf{ID}^i$
  with $b = \mathsf{ownerOfID}(\mi{id})$
  being an honest browser (in $S$),
  every relying party $r$
  that is honest in $S$,
  every client id $c$
  that has been issued to $r$
  by $i$,
  every term $y$,
  every id token
  $t = \mathsf{sig}([\str{iss}: d, \str{sub}: \mi{id}, \str{aud}: c,
  \str{nonce}: y ] , k)$ with $k= S(i).\str{jwks}$,
  every attacker process $a$
  we have that $t \not\in d_{\emptyset}(S(a))$.
\end{lemma}

\begin{proof}
  The signing key $k$
  is only known to $i$
  initially and at least up to $S$
  (since $i$
  is honest). Therefore, only $i$
  can create $t$.
  There are two places where an honest IdP can create such a token in
  Algorithm~\ref{alg:idp-oidc}: In
  Line~\ref{line:idp-create-id-token-auth2} (immediately after
  receiving the user credentials) and in
  Lines~\ref{line:idp-create-id-token-from-code}ff.~(after receiving an
  access token).

  We now distinguish between these two cases to show that in either
  case, the attacker cannot get hold of an id token. We start with the
  first case.

  \begin{description}
  \item[ID token was created in
    Line~\ref{line:idp-create-id-token-auth2}.]\strut

To create $t$,
    the IdP $i$
    must have received a request to the path $\str{/auth2}$
    in Lines~\ref{line:idp-receive-auth2}ff.~of
    Algorithm~\ref{alg:idp-oidc}. It is clear that $i$
    sends the response to this request to the sender of the request,
    and, if that sender is honest, the response cannot be read by an
    attacker. The request must contain $\mathsf{secretOfId}(\mi{id})$.
    Only $b$
    and $i$
    know this secret (per
    Lemma~\ref{lemma:attacker-does-not-learn-password}). Since $i$
    does not send requests to itself, the request must
    have been sent from $b$.
    Since the origin header in the request must be a domain of $i$,
    we know that the request was not initiated by a script other than
    $i$'s
    own scripts, in particular, it must have been initiated by
    $\mi{script\_idp\_form}$.

    Now it is easy to see that this script does not use the token $t$
    in any way after the token was returned from $i$,
    since the script uses a form post to transmit the credentials to
    $i$,
    and the window is subsequently navigated away. Instead, $i$
    provides an empty script in its response to $b$.
    This response contains a location redirect header. It is now
    crucial to check that this location redirect does not cause the id
    token to be leaked to the attacker: With
    Lemma~\ref{lemma:client-reg-integrity} we have that the
    redirection URIs that are registered at $i$
    for the client id $c$
    only point to domains of $r$
    (and use HTTPS).

    We therefore know that $b$ will send an HTTPS request (say $m$) containing $t$ to
    $r$. We have to
    check whether $r$ or a script delivered by $r$ to $b$ will leak
    $t$. Algorithm~\ref{alg:rp-oidc-http-request} processes all HTTPS
    requests delivered to $r$.
As $i$ redirected $b$ using the 303
    status code, the request to $r$ must be a GET request. Hence, $r$
    does not process this request in
    Lines~\ref{line:rp-start-login-endpoint}ff.~of Algorithm~\ref{alg:rp-oidc-http-request}.
 Lines~\ref{line:rp-serve-index}ff.~do
    only respond with a script and do not use $t$ in any way. We are
    left with Lines~\ref{line:rp-redir-endpoint}ff.~to be analyzed.

    As in $m$ the id token $t$ is always contained in a dictionary
    under the key $\str{id\_token}$ and this dictionary is either in the
    parameters, the fragment, or the body of $m$, it is now easy to
    see that $r$ does not store or send out $t$ in any way.

    We now have to check if a script delivered by $r$ to $b$ leads to
    $t$ being leaked. First note that $r$ always sets the header
    $\str{ReferrerPolicy}$ to $\str{origin}$ in every HTTP(S) response
    $r$ sends out. Hence, $t$ can never leak using the Referer header.

    There are only two scripts that $r$ may deliver: (1) The script
    $\mi{script\_rp\_index}$ either issues a $\str{FORM}$ command to
    the browser, which does not contain $t$, or this script issues a
    $\str{HREF}$ command to the browser for some URL, which also does
    not contain $t$. (2) The script $\mi{script\_rp\_get\_fragment}$
    takes the fragment of the current URL (which may be a dictionary
    that contains $t$ under the key $\str{id\_token}$) and the
    $\str{iss}$ parameter and issues an HTTPS request to $r$ for the
    path $\str{/redirect\_ep}$, which will be processed by $r$ in
    Lines~\ref{line:rp-redir-endpoint}ff.~of
    Algorithm~\ref{alg:rp-oidc-http-request}. Now, the same reasoning
    as above applies.

  \item[ID token was created in
    Lines~\ref{line:idp-create-id-token-from-code}ff.]\strut

 In this case,
    the id token is created by $i$
    only when an HTTPS request was received by $i$
    that matches the following criteria: (a) it must be for the path
    $\str{/token}$,
    (b) it must contain the client id $c$
    in the body (under the key $\str{client\_id}$),
    and (c) it must contain a authorization code in the body (under
    the key $\str{code}$)
    that occurs in one of $i$'s
    internal records with a matching subject, issuer, and nonce. To be
    more precise, the request must contain a code $\mi{code}$
    such that there is a record $\mi{rec}$
    with $\mi{rec} \inPairing S(i).\str{records}$
    and $\mi{rec}[\str{issuer}] \equiv d$,
    $\mi{rec}[\str{subject}] \equiv \mi{id}$,
    $\mi{rec}[\str{client\_id}] \equiv c$,
    and $\mi{rec}[\str{code}] \equiv c$.
    Such a record can only be created and the
    authorization code $\mi{code}$
    issued under exactly the same circumstances that allow an id token
    (of the above form) to be created in
    Line~\ref{line:idp-create-id-token-from-code}. With exactly the
    same reasoning as above, this time for the code instead of the id
    token, we can follow that $\mi{code}$
    does not leak to the attacker.

  \end{description}
  We have therefore shown that no attacker process can get hold of the
  id token $t$. This proves the lemma.
\end{proof}

\begin{assumption}\label{asn:prop-a}  
  There exists an OIDC web system $\oidcwebsystem^n$
  with a network attacker such that there exists a run $\rho$
  of $\oidcwebsystem^n$,
  a configuration $(S, E, N)$
  in $\rho$,
  some $r\in \fAP{RP}$
  that is honest in $S$,
  some identity $\mi{id} \in \mathsf{ID}$
  with $\mathsf{governor}(\mi{id})$
  being an honest IdP (in $S$)
  and $\mathsf{ownerOfID}(\mi{id})$
  being an honest browser (in $S$),
  some service session identified by some nonce $n$
  for $\mi{id}$
  at $r$,
  and $n$
  is derivable from the attackers knowledge in $S$
  (i.e., $n \not\in d_{\emptyset}(S(\fAP{attacker}))$).
\end{assumption}

\begin{lemma}\label{lemma:authn-contradiction}
  Assumption~\ref{asn:prop-a} is a contradiction.
\end{lemma}

\begin{proof}
  We first recall how the service session identified by some nonce $n$
  for $\mi{id}$
  at $r$
  is defined. It means that there is some session id $x$
  and a domain $d \in \mathsf{dom}(\mathsf{governor}(\mi{id}))$
  with
  $S(r).\str{sessions}[x][\str{loggedInAs}] \equiv \an{d, \mi{id}}$
  and $S(r).\str{sessions}[x][\str{serviceSessionId}] \equiv n$.
  Now the assumption is that $n$
  is derivable from the attacker's knowledge. Since we have that
  $S(r).\str{sessions}[x][\str{serviceSessionId}] \equiv n$,
  we can check where and how, in general, service session ids can be
  created. It is easy to see that this can only happen in
  Algorithm~\ref{alg:rp-check-id-token}, where, in
  Line~\ref{line:rp-choose-service-session-id}, the RP chooses a fresh
  nonce as the value for the service session id, in this case $x$.
  In the line before, it sets the value for
  $S(r).\str{sessions}[x][\str{loggedInAs}]$,
  in this case $\an{d, \mi{id}}$.
  In the Lines~\ref{line:rp-start-id-token-checks}ff., $r$
  performs several checks to ensure the integrity and authenticity of
  the id token.

  The function function $\mathsf{CHECK\_ID\_TOKEN}$
  can be called in either (a)
  Line~\ref{line:rp-call-check-id-token-immediately} of
  Algorithm~\ref{alg:rp-oidc-http-request} or (b) in
  Line~\ref{line:rp-call-check-id-token-after-code} of
  Algorithm~\ref{alg:rp-oidc-http-response}.
  
  We can now distinguish between these two cases.
  \begin{description}
  \item[Case (a).]\strut

 In this case, we can easily see that the same party that finally
    receives the service session id $x$,
    must have provided, in an HTTPS request, an id token (say, $t'$)
    with the following properties (for some $l < j$):
    \[ \mathsf{extractmsg}(t')[\str{iss}] \equiv d \]
    \[ \mathsf{extractmsg}(t')[\str{sub}] \equiv \mi{id} \]
    \[ \mathsf{extractmsg}(t')[\str{aud}] \equiv
      S^l(r).\str{clientCredentialsCache}[d][\str{client\_id}]\]
    \[ \mathsf{checksig}(t', \mathsf{pub}(S^l(i).\str{jwks})) \equiv
      \True\ .\]
    The attacker (and, by extension, any other party except for $i$,
    $b$,
    and $r$),
    however, cannot know such an id token (see
    Lemma~\ref{lemma:attacker-does-not-learn-id-tokens}). Since $r$
    and $i$
    do not send requests to $r$,
    the id token must have been sent by $b$
    to $r$.
    As the service session id $x$
    is only contained in a set-cookie header with the httpOnly and
    secure flags set, $b$
    will only ever send the service session id $x$
    to $r$
    (contained in a cookie header). As $b$
    does not leak $x$
    in any other way and as $r$
    does not leak information sent in cookie headers, the service
    session id $x$ does not leak.
  \item[Case (b).]\strut

 Otherwise, the party that finally receives the
    service session id $x$
    needs to provide a code $c$ 
    such that, when this code is sent to the token endpoint of $i$
    (Algorithm~\ref{alg:rp-token-request}), $i$
    responds with an id token matching the criteria listed in
    Case~(a). This, however, would mean that an attacker, knowing this
    code, could do the same, violating
    Lemma~\ref{lemma:attacker-does-not-learn-id-tokens}. (Note that
    for every run where a client secret is associated with the client
    id there is also a run where the client secret is not used; the
    client secret does not prevent the attacker from requesting an id
    token at the token endpoint for a valid code.) \df{<- kann man noch ausfuehren}
  \end{description}
  We therefore have shown that the attacker cannot know $x$,
  proving the lemma and showing that Assumption~\ref{asn:prop-a} is,
  in fact, a contradiction.
\end{proof}

\subsection{Proof of Authorization}
\label{sec:proof-authorization}

As above, we assume that there exists an OIDC web system that is not
secure w.r.t.~authorization and lead this to a contradiction. 

\begin{assumption}\label{asn:authorization} 
  There exists an OIDC web system with a network attacker
  $\oidcwebsystem^n$,
  a run $\rho$
  of $\oidcwebsystem^n$,
  a state $(S^j, E^j, N^j)$
  in $\rho$,
  a relying party $r\in \fAP{RP}$
  that is honest in $S^j$,
  an identity provider $i\in \fAP{IdP}$
  that is honest in $S^j$,
  a browser $b$
  that is honest in $S^j$,
  an identity $\mi{id} \in \mathsf{ID}^i$
  owned by $b$,
  a nonce $n$,
  a term $x \inPairing S^j(i).\str{records}$
  with $x[\str{subject}] \equiv \mi{id}$,
  $n \inPairing x[\str{access\_tokens}]$,
  and the client id $x[\str{client\_id}]$
  has been issued by $i$
  to $r$,
  and $n$
  is derivable from the attackers knowledge in $S^j$
  (i.e., $n \in d_{\emptyset}(S^j(\fAP{attacker}))$).
\end{assumption}

\begin{lemma}\label{lemma:authz-contradiction}
  Assumption~\ref{asn:authorization} is a contridiction.
\end{lemma}

\begin{proof}
  We have that $n \in d_{\emptyset}(S^j(\fAP{attacker}))$
  and therefore, there must have been a message from a third party to
  $\fAP{attacker}$
  (or any other corrupted party, which could have forwarded $n$
  to the attacker) that contained $n$.
  We can now distinguish between the parties that could have sent $n$
  to the attacker (or to the corrupted party):

  \textbf{The access token $n$ was sent  by the browser $b$:}

  We now track different cases in which the access token $n$ can get
  into $b$'s knowledge. We will omit the cases in which $b$ learns $n$
  from any dishonest party as in such a case there is a different run
  $\rho'$ of $\oidcwebsystem^n$ in which this dishonest party
  immediately sends $n$ to the attacker.

  (I) First, we analyze the case in which $b$ has learned $n$ from an
  honest (in $S^j$) identity provider, say $i'$. In this case, $b$ must have
  received an HTTPS response from $i'$ (honest identity providers do not send out unencrypted HTTP responses). Honest identity providers send
  out HTTPS responses in Lines~\ref{line:idp-send-webfinger},
  \ref{line:idp-send-oidc-config}, \ref{line:idp-send-jwks},
  \ref{line:idp-send-form}, \ref{line:idp-send-auth-resp},
  and~\ref{idp-send-token-ep} of Algorithm~\ref{alg:idp-oidc} and
  Line~\ref{line:idp-send-reg-response} of
  Algorithm~\ref{alg:idp-oidc-other}. It is easy to see that $i'$ does
  not send out $n$ in Lines~\ref{line:idp-send-webfinger},
  \ref{line:idp-send-oidc-config}, \ref{line:idp-send-jwks},
  and~\ref{line:idp-send-form} of Algorithm~\ref{alg:idp-oidc} and
  Line~\ref{line:idp-send-reg-response} of
  Algorithm~\ref{alg:idp-oidc-other} (given that the attacker does not know $n$), leaving
  Lines~\ref{line:idp-send-auth-resp}, and~\ref{idp-send-token-ep} of
  Algorithm~\ref{alg:idp-oidc} to analyze.

  (a) If $i'$ sends out $n$ in Line~\ref{line:idp-send-auth-resp} of
  Algorithm~\ref{alg:idp-oidc}, $b$ must have sent an HTTPS POST
  request bearing an Origin header for one of the domains of $i'$ to
  $i'$. As $i'$ only delivers the script $\mi{script\_idp\_form}$,
  only this script could have caused this request (using a
  $\str{FORM}$ command). Hence, $b$ will navigate the corresponding
  window to the location indicated in the $\str{Location}$ header of
  the HTTPS response assembled in
  Lines~\ref{line:idp-receive-auth2}ff.~of
  Algorithm~\ref{alg:idp-oidc}. The body of this response can consist
  of an authorization code (a fresh nonce), an access token (a fresh
  nonce), and an id token consisting of one domain of $i$, a valid
  user name for $i'$, a client id, and a nonce (say $n'$) from the
  request.

  We now reason why $i'$ must be $i$, and the access token in the
  response must be $n$. In the id token, only the client id and the
  nonce $n'$ could be $n$. As the client id is always set by the
  attacker during registration, the client id cannot be $n$. The nonce
  $n'$ originates from the request sent by $b$ on the command of
  $\mi{script\_idp\_form}$. In this request, the nonce $n'$ must be
  contained in the URL, which is the URL from which the script was
  loaded before. Hence, the browser must have been navigated to this
  URL. As the attacker does not know $n$ at this point, only honest
  scripts or honest web servers could have navigated the browser to
  such an URL (containing $n$). Honest relying parties only populate
  the parameter $\str{nonce}$ (bearing $n'$) in such a redirect with a
  fresh nonce, honest identity providers do not populate such an URL
  parameter by themselves, but could have used this parameter in a
  redirect based on a registered redirect URL. As honest parties never
  register such a redirect URL, $n'$ cannot be $n$. Hence, only the
  access token in the response above can be $n$. As the access token
  is a fresh nonce, we must have that $i'$ is $i$ and that $i$ creates the term
  $x \inPairing S^j(i).\str{records}$ with
  $x[\str{subject}] \equiv \mi{id}$,
  $n \inPairing x[\str{access\_tokens}]$ ($i$ will never create such a term at any other time), and the client id
  $x[\str{client\_id}]$ has been issued by $i$ to $r$. Hence, the
  location redirect issued by $i$ must point to an URL of $r$ with the path $\str{/redirect\_ep}$ (see Lemma~\ref{lemma:client-reg-integrity}) and this URL contains the parameter $\str{iss}$  with a domain of $i$. The access token $n$ is only contained in the fragment of
  this URL under the key $\str{access\_token}$.

  Now, $b$ sends an HTTPS request to $r$. This request does not
  contain $n$ (as it is placed in the fragment part of the URL). The
  relying party $r$ can (regardless of the path) send out only the
  scripts $\mi{script\_rp\_index}$ and
  $\mi{script\_rp\_get\_fragment}$ as a response to such a
  request. The script $\mi{script\_rp\_index}$ ignores the fragment of
  its URL. The script $\mi{script\_rp\_get\_fragement}$ takes the
  fragment of the URL and uses it as the body of a POST request to its
  own origin (which is $r$) with path $\str{/redirect\_uri}$. When $r$
  processes this POST request, $r$ only ever uses $n$ in
  Line~\ref{line:rp-call-use-access-token-1} of
  Algorithm~\ref{alg:rp-oidc-http-request}. There, the access token
  $n$ and the value of the parameter $\str{iss}$ (a domain of $i$) is
  processed by Algorithm~\ref{alg:rp-use-access-token}. From
  Lemma~\ref{lemma:oidc-config-cache-integrity}, we know that $r$ will
  only send $n$ to the token endpoint of $i$ in an HTTPS request. This
  request is then processed by $i$ in
  Lines~\ref{line:idp-token-ep}ff.~of
  Algorithm~\ref{alg:idp-oidc}. There, $i$ only checks $n$, but does
  not send out $n$.

  If $b$ sends out a response, the same reasoning as above applies.
  Hence, we have that $n$ does not leak to the attacker in this case.

  (b) If $i'$ sends out $n$ in Line~\ref{idp-send-token-ep} of
  Algorithm~\ref{alg:idp-oidc}, we have that the response does not
  contain a script or a redirect. The browser would only interpret
  such a response if the request was caused by an
  $\str{XMLHTTPREQUEST}$ command of a script. Honest scripts do not
  issue such a command, leaving only the attacker script as the only
  possible source for such a request. If $i'$ is not $i$, it is easy
  to see that this response cannot contain $n$. The identity provider
  $i$ only sends out $n$ (taken from the subterm $\str{records}$ from
  its state) if the request contains a valid authorization code for
  this access token. With the same reasoning as for the authentication property above,
  the attacker cannot know a valid id token for any user id owned by
  $b$. If the attacker would know a valid authorization code, he
  could retrieve a valid id token (for such a user id) from $i$. Hence, the attacker cannot
  know a valid authorization code. As this reasoning also applies for the
  attacker script, the attacker script could not have caused a request
  to $i$ revealing $n$.

  (II) Now, we analyze the case in which $b$ received $n$ from some
  honest (in $S^j$) relying party, say $r'$.  In this case, $b$ must
  have received an HTTPS response from $r'$ (honest relying parties do
  not send out unencrypted HTTP responses). Honest relying parties
  only send out such HTTPS responses in Lines~\ref{line:rp-send-index}
  and~\ref{line:rp-send-script-get-fragment} of
  Algorithm~\ref{alg:rp-oidc-http-request},
  Line~\ref{line:rp-send-set-service-session} of
  Algorithm~\ref{alg:rp-check-id-token}, and
  Line~\ref{line:rp-send-authorization-redir} of
  Algorithm~\ref{alg:rp-oidc-cont-login-flow}. In the former three
  cases, $r'$ only sends out fixed information and fresh nonces
  (either chosen by $r'$ directly before sending out the message or
  the HTTPS nonce and key chosen by $b$ when creating the request). In
  the latter case, $r'$ (besides the pieces of information as before)
  also adds information from its OpenID Connect configuration cache
  (i.e., client id and authorization URL). From
  Lemma~\ref{lemma:oidc-config-cache-integrity}
  and~\ref{lemma:client-reg-integrity} we know that if $r'$ gathered
  this information from an honest party, this information cannot
  contain $n$. As the attacker does not know $n$ at this point, this
  registration information cannot contain $n$ if $r'$ gathered this
  information from a dishonest party. Hence, $b$ cannot have learned
  $n$ from any $r'$.

  (III) $b$ cannot have learned $n$ from a different honest (in $S^j$) browser as honest browsers do not create messages that can be interpreted by honest browsers.

  (IV) $b$ cannot have learned $n$ from the attacker, as the attacker does not know $n$ at this point.

  \textbf{The access token $n$
    was sent by the IdP $i$:}
  We can see that access tokens are sent by the IdP only after a
  request to the path (endpoints) $\str{/auth2}$ or to the path
  and $\str{/token}$.

  In case of a request to the path $\str{/auth2}$,
  a pair of access tokens is created and the first access token in the
  pair is returned from the endpoint. If the attacker would be able to
  learn $n$
  from this endpoint such that there exists a record
  $x \inPairing S^j(i).\str{records}$
  with $x[\str{subject}] \equiv \mi{id}$,
  then the attacker would need to provide the user's credentials to
  the IdP $i$.
  The attacker cannot know these credentials
  (Lemma~\ref{lemma:attacker-does-not-learn-password}), therefore the
  attacker cannot request $n$ from this endpoint.

  In case of a request to the path $\str{/token}$,
  the attacker would need to provide an authorization code that is
  contained in the same record (in this case $x$)
  as $n$.
  Now, recall that we have that $x[\str{subject}] \equiv \mi{id}$
  and $c := x[\str{client\_id}]$
  has been issued to $r$
  by $i$.
  We can now see that if the attacker would be able to send a request
  to the endpoint $\str{/token}$
  which would cause a response that contains $n$,
  the attacker would also be able to learn an id token of the form
  shown in Lemma~\ref{lemma:attacker-does-not-learn-id-tokens} (the
  issuer is a domain of $i$,
  the subject is $\mi{id}$,
  and the audience is $c$).
  This would be a contradiction to
  Lemma~\ref{lemma:attacker-does-not-learn-id-tokens}.

  We can conclude that the access token $n$ was not sent by the IdP $i$. 

  \textbf{The access token $n$ was sent by the RP $r$:}
  
  The only place where the (honest) RP uses an access token is in
  Algorithm~\ref{alg:rp-use-access-token}. There, the access token is
  sent to the domain of the token endpoint (compare
  Algorithm~\ref{alg:rp-token-request}, where the authorization code
  is sent to that endpoint). We can now see that the access token is
  always sent to $i$:
  If the access token would be sent to the attacker, so would the
  authorization code in Algorithm~\ref{alg:rp-token-request}, and
  Lemma~\ref{lemma:attacker-does-not-learn-id-tokens} would not hold
  true. \df{<- etwas kurz}
    
\end{proof}

\subsection{Proof of Session Integrity}
\label{sec:proof-session-integrity-all}

Before we prove this property, we highlight that in the absence of a
network attacker and with the DNS server as defined for
$\oidcwebsystem^w$,
HTTP(S) requests by (honest) parties can only be answered by the owner
of the domain the request was sent to, and neither the requests nor
the responses can be read or altered by any attacker unless he is the
intended receiver.

We further show the following lemma, which says that an attacker
(under the assumption above) cannot learn a $\mi{state}$
value that is used in a login session between an honest browser, an
honest IdP, and an honest RP.

\begin{lemma}[Third parties do not learn state]\label{lemma:attacker-does-not-learn-state}
  There exists no run $\rho$
  of an OIDC web system with web attackers $\oidcwebsystem^w$,
  no configuration $(S, E, N)$
  of $\rho$,
  no $r\in \fAP{RP}$
  that is honest in $S$,
  no $i \in \fAP{IDP}$
  that is honest in $S$,
  no browser $b$
  that is honest in $S$,
  no nonce $\mi{lsid} \in \nonces$,
  no domain $h \in \mathsf{dom}(r)$
  of $r$,
  no terms $g$,
  $x$,
  $y$,
  $z \in \terms$,
  no cookie $c := \an{\str{sessionId}, \an{\mi{lsid}, x, y, z}}$,
  no atomic DY process $p \in \websystem \setminus \{b,i,r\}$
  such that (1) $S(r).\str{sessions}[\mi{lsid}] \equiv g$,
  (2) $g[\str{state}] \in d_\emptyset(S(p))$,
  (3) $S(r).\str{issuerCache}[g[\str{identity}]] \in \mathsf{dom}(i)$,
  and (4) $c \inPairing S(b).\str{cookies}[h]$.
\end{lemma}

\begin{proof}
  To prove Lemma~\ref{lemma:attacker-does-not-learn-state}, we track
  where the login session identified by $\mi{lsid}$
  is created and used.
  
  Login session ids are only chosen in Line~\ref{line:rp-choose-lsid}
  of Algorithm~\ref{alg:rp-oidc-http-request}. After the session id
  was chosen, its value is sent over the network to the party that
  requested the login (in Line~\ref{line:rp-send-authorization-redir}
  of Algorithm~\ref{alg:rp-oidc-cont-login-flow}). We have that for
  $\mi{lsid}$,
  this party must be $b$
  because only $r$
  can set the cookie $c$
  for the domain $h$
  in the state of $b$\footnote{Note
    that we have only web attackers.} and
  Line~\ref{line:rp-send-authorization-redir} of
  Algorithm~\ref{alg:rp-oidc-cont-login-flow} is actually the only
  place where $r$ does so.

  Since $b$
  is honest, $b$
  follows the location redirect contained in the response sent by $r$.
  This location redirect contains $\mi{state}$
  (as a URL parameter). The redirect points to some domain of $i$.
  (This follows from Lemma~\ref{lemma:oidc-config-cache-integrity}.)
  The browser therefore sends (among others) $\mi{state}$
  in a GET request to $i$.
  Of all the endpoints at $i$
  where the request can be received, the authorization endpoint is the
  only endpoint where $\mi{state}$
  could potentially leak to another party. (For all other endpoints,
  the value is dropped.) If the request is received at the
  authorization endpoint, $\mi{state}$
  is only sent back to $b$
  in the initial scriptstate of $\mi{script\_idp\_form}$.
  In this case, the script sends $\mi{state}$
  back to $i$
  in a POST request to the authorization endpoint. Now, $i$
  redirects the browser $b$
  back to the redirection URI that was passed alongside $\mi{state}$
  from $r$
  via the browser to $i$.
  This redirection URI was chosen in
  Line~\ref{line:rp-choose-redirect-uris} of
  Algorithm~\ref{alg:rp-oidc-cont-login-flow} and therefore points to
  one of $r$'s domains. 
  The value $\mi{state}$
  is appended to this URI (either as a parameter or in the fragment).
  The redirection to the redirection URI is then sent to the browser
  $b$. Therefore, $b$ now sends a GET request to $r$.

  If $\mi{state}$
  is contained in the parameter, then $\mi{state}$
  is immediately sent to $r$
  where it is compared to the stored login session records but neither
  stored nor sent out again. In each case, a script is sent back to
  $b$.
  The scripts that $r$
  can send out are $\mi{script\_rp\_index}$
  and $\mi{script\_rp\_get\_fragment}$,
  none of which cause requests that contain $\mi{state}$
  (recall that we are in the case where $\mi{state}$
  is contained in the URI parameter, not in the fragment). Also, since
  both scripts are always delivered with a restrictive Referrer Policy
  header, any requests that are caused by these scripts (e.g., the
  start of a new login flow) do not contain $\mi{state}$
  in the referer header.\footnote{We note that, as discussed earlier,
    without the Referrer Policy, $\mi{state}$
    could leak to a malicious IdP or other parties.}

  If $\mi{state}$
  is contained in the fragment, then $\mi{state}$
  is not immediately sent to $r$,
  but instead, a request without $\mi{state}$
  is sent to $r$.
  Since this is a GET request, $r$
  either answers 
  with a response that only contains the string $\str{ok}$ but no script
  (Lines~\ref{line:rp-send-set-service-session}ff.~of Algorithm~\ref{alg:rp-check-id-token}),
  a response containing $\mi{script\_rp\_index}$
  (Lines~\ref{line:rp-serve-index}ff.~of
  Algorithm~\ref{alg:rp-oidc-http-request}), or
  a response containing $\mi{script\_rp\_get\_fragment}$
  (Line~\ref{line:rp-send-script-get-fragment}~of
  Algorithm~\ref{alg:rp-oidc-http-request}).
  In case of the $\str{ok}$
  response, $\mi{state}$
  is not used anymore by the browser. In case of
  $\mi{script\_rp\_index}$,
  the fragment is not used. (As above, there is no other way in which
  $\mi{state}$
  can be sent out, also because the fragment part of an URL is
  stripped in the referer header.) In the case of
  $\mi{script\_rp\_get\_fragment}$
  being loaded into the browser, the script sends $\mi{state}$
  in the body of an HTTPS request to $r$
  (using the path $\str{/redirect\_ep}$).
  When $r$
  receives this request, it does not send out $\mi{state}$
  to any party (see Lines~\ref{line:rp-redir-endpoint}ff.\ of
  Algorithm~\ref{alg:rp-oidc-http-request}).

  This shows that $\mi{state}$
  cannot be known to any party except for $b$, $i$, and $r$.
\end{proof}

\subsubsection{Proof of Session Integrity for Authentication} \label{sec:proof-sess-integr-authn}
To prove that every OIDC web system with web attackers is secure
w.r.t.~session integrity for authentication, we assume that there
exists an OIDC web system with web attackers which is not secure
w.r.t.~session integrity for authentication and lead this to a
contradiction.

\begin{assumption}\label{asn:property-si-authn}
  There exists an $\oidcwebsystem^w$
  be an OIDC web system with web attackers, a run $\rho$
  of $\oidcwebsystem^w$, a processing step $Q$ in $\rho$ with
  $$Q = (S, E, N) \xrightarrow[]{}
  (S', E', N')$$  (for some $S$, $S'$, $E$, $E'$, $N$, $N'$), a browser $b$
  that is honest in $S$,
  an IdP $i\in \fAP{IdP}$,
  an identity $u$
  that is owned by $b$,
  an RP $r\in \fAP{RP}$
  that is honest in $S$,
  a nonce $\mi{lsid}$,
  with $\mathsf{loggedIn}_\rho^Q(b, r, u, i, \mi{lsid})$
  and (1) there exists no processing step $Q'$
  in $\rho$
  (before $Q$)
  such that $\mathsf{started}_\rho^{Q'}(b, r, \mi{lsid})$,
  or (2) $i$
  is honest in $S$,
  and there exists no processing step $Q''$
  in $\rho$
  (before $Q$)
  such that
  $\mathsf{authenticated}_\rho^{Q''}(b, r, u, i, \mi{lsid})$.
\end{assumption}

\begin{lemma}\label{lemma:si-authn-contradiction}
  Assumption~\ref{asn:property-si-authn} is a contradiction.
\end{lemma}

\begin{proof}
  \textbf{(1)} We have that
  $\mathsf{loggedIn}_\rho^Q(b, r, u, i, \mi{lsid})$.
  With Definition~\ref{def:user-logged-in} we have that $r$
  sent out the service session id belonging to $\mi{lsid}$
  to $b$.
  (This can only happen when the function $\mathsf{CHECK\_ID\_TOKEN}$
  (Algorithm~\ref{alg:rp-check-id-token}) was called with $\mi{lsid}$
  as the first parameter.) This means that $r$
  must have received a request from $b$
  containing a cookie with the name $\str{sessionId}$
  and the value $\mi{lsid}$:
  The response by $r$
  (which we know was sent to $b$)
  was sent in Line~\ref{line:rp-send-set-service-session} in
  Algorithm~\ref{alg:rp-check-id-token}. There, $r$
  looks up the address of $b$
  using the login session record under the key
  $\str{redirectEpRequest}$.
  This key is only ever created in
  Line~\ref{line:rp-set-redirect-ep-request-record} of
  Algorithm~\ref{alg:rp-oidc-http-request}. This line is only ever
  called when $r$
  receives an HTTPS request from $b$ with the cookie as described.

  We can now track how the cookie was stored in $b$:
  Since the cookie is stored under a domain of $r$
  and we have no network attacker, the cookie must have been set by
  $r$.
  This can only happen in Line~\ref{line:rp-send-authorization-redir}
  in Algorithm~\ref{alg:rp-oidc-cont-login-flow}. Similar to the
  $\str{redirectEpRequest}$
  session entry above, $r$
  sends this cookie as a response to a stored request, in this case,
  using the key $\str{startRequest}$
  in the session data. This key is only ever created in
  Lines~\ref{line:rp-start-login-endpoint}ff.~of
  Algorithm~\ref{alg:rp-oidc-http-request}. Hence, there must have
  been a request from $b$
  to $r$
  containing a POST request for the path $\str{/startLogin}$
  with an origin header for an origin of $r$.
  There are only two scripts which could potentially send such a
  request, $\mi{script\_rp\_index}$
  and $\mi{script\_rp\_get\_fragment}$.
  It is easy to see that only the former send requests of the kind
  described. We therefore have a processing step $Q'$
  that happens before $Q$
  in $\rho$ with $\mathsf{started}_\rho^{Q'}(b, r, \mi{lsid})$.

  \textbf{(2)} Again, we have that
  $\mathsf{loggedIn}_\rho^Q(b, r, u, i, \mi{lsid})$.
  Now, however, $i$ is honest.

  We first highlight that if $r$
  receives an HTTPS request, say $m$,
  which contains $\mi{state}$
  such that
  $S(r).\str{sessions}[\mi{lsid}][\str{state}] \equiv \mi{state}$
  and contains a cookie with the name $\str{sessionId}$
  and the value $\mi{lsid}$
  then this request must have come from the browser $b$
  and be caused by a redirection from $i$
  or a script from $r$.
  From $\mathsf{loggedIn}_\rho^Q(b, r, u, i, \mi{lsid})$
  it follows that there is a term $g$
  such that $S(r).\str{sessions}[\mi{lsid}] \equiv g$,
  and $g[\str{loggedInAs}] \equiv \an{d, u}$
  with $d \in \mathsf{dom}(i)$.
  From the Algorithm~\ref{alg:rp-check-id-token} we have that
  $S(r).\str{issuerCache}[g[\str{identity}]] \equiv d$.
  With Lemma~\ref{lemma:attacker-does-not-learn-state} we have that
  only $b$, $r$, and $i$ know $\mi{state}$.

  We can now show that $m$
  must have been caused by $i$
  by means of a Location redirect that was sent to $b$
  or by the script $\mi{script\_rp\_get\_fragment}$.
  First, neither $r$
  nor $i$
  send requests that contain cookies. The request must therefore have
  originated from $b$.
  Since no attacker knows $\mi{state}$,
  the request cannot have been caused by any attacker scripts or by
  redirects from parties other than $r$
  or $i$
  (otherwise, there would be runs where the attacker learns
  $\mi{state}$).

  Redirects from $r$
  can be excluded, since $r$
  only sends a redirection in
  Line~\ref{line:rp-send-authorization-redir} in
  Algorithm~\ref{alg:rp-oidc-cont-login-flow} but there, a freshly
  chosen state value is used, hence, there is only one processing step
  in which $r$
  uses $\mi{state}$
  for this redirect. This is the processing step where $r$
  adds $\mi{state}$
  to the session data stored under the key $\mi{lsid}$.
  Since this is a session in which the honest IdP $i$
  is used, and with Lemma~\ref{lemma:oidc-config-cache-integrity}, we
  have that $r$ does not redirect to itself (but to $i$ instead).

  The scripts $\mi{script\_rp\_index}$
  and $\mi{script\_idp\_form}$
  do not send requests with the $\mi{state}$
  parameter. Therefore, the remaining causes for the request $m$
  are either the script $\mi{script\_rp\_get\_fragment}$
  or a location redirect from $i$.
  
  If the request $m$
  was caused by $\mi{script\_rp\_get\_fragment}$,
  then it is easy to see from the definition of
  $\mi{script\_rp\_get\_fragment}$
  (Algorithm~\ref{alg:script-rp-get-fragment}) that this script only
  sends data from the fragment part of its own URI (except for the
  $\str{iss}$
  parameter) and it sends this data only to its own origin. This
  script therefore must have been sent to $b$
  by $r$,
  which only sends this script after receiving HTTPS request to the
  redirection endpoint ($\str{/redirect\_ep}$).
  With the same reasoning as above this must have been caused by a
  location redirect from $i$. 

  For clarity, by $m_\text{redir}$
  we denote the response by $i$
  to the browser $b$
  containing this redirection. We now show that for $m_\text{redir}$
  to take place, there must have been a processing step $Q''$
  (before $Q$)
  with $\mathsf{authenticated}_\rho^{Q''}(b, r, u', i, \mi{lsid})$
  for some identity $u'$.
  
  In the honest IdP $i$,
  there is only one place where a redirection happens, namely in
  Line~\ref{line:idp-send-auth-resp} in Algorithm~\ref{alg:idp-oidc}.
  To reach this point, $i$
  must have received the login data for $u'$
  in the HTTPS request corresponding to $m_\text{redir}$.
  This must be a POST request with an origin header containing an
  origin of $i$.
  As $i$
  only uses $\mi{script\_idp\_form}$,
  the request must have been caused by this script. Hence, we have
  $\mathsf{authenticated}_\rho^{Q''}(b, r, u', i, \mi{lsid})$.

  We now only need to show that $u' = u$. 

  With $\mathsf{loggedIn}_\rho^Q(b, r, u, i, \mi{lsid})$,
  we know that $r$
  must have called the function $\mathsf{CHECK\_ID\_TOKEN}$
  (Algorithm~\ref{alg:rp-check-id-token}). We further have that
  $S(r).\str{sessions}[\mi{lsid}][\str{loggedInAs}] \equiv
  \an{d,u}$. We therefore have that $i$
  must have created an id token with the issuer $d$
  and the identity $u$.
  $\mathsf{CHECK\_ID\_TOKEN}$
  can be called in Line~\ref{line:rp-call-check-id-token-after-code}
  in Algorithm~\ref{alg:rp-oidc-http-response} and in
  Line~\ref{line:rp-call-check-id-token-immediately} in
  Algorithm~\ref{alg:rp-oidc-http-request}. We now distinguish between
  these two cases.

  \textit{$\mathsf{CHECK\_ID\_TOKEN}$
    was called in Line~\ref{line:rp-call-check-id-token-after-code} in
    Algorithm~\ref{alg:rp-oidc-http-response}:} When the function is
  called in this line, there must have been an HTTPS request reference
  with the string $\str{TOKEN}$
  (cf. generic web server model, Algorithm~\ref{alg:simple-send}).
  Such a reference is only created in
  Line~\ref{line:rp-start-retrieve-code} of
  Algorithm~\ref{alg:rp-token-request}. With
  Lemma~\ref{lemma:oidc-config-cache-integrity} we know that this
  HTTPS request was sent to the token endpoint (path $\str{/token}$)
  of $i$
  (because the issuer, stored in the login session record, is $i$).
  Since the token endpoint returned an id token of the form described
  above, and $i$
  is honest, there must have been a record in $i$,
  say $v$,
  with $v[\str{subject}] \equiv u$.
  In the request to the token endpoint, $r$
  must have sent a nonce $c$
  such that $v[\str{code}] \equiv c$.
  This request, as already mentioned, must have been sent in
  Line~\ref{line:rp-start-retrieve-code} of
  Algorithm~\ref{alg:rp-token-request}. This means, that there must
  have been an HTTPS request to $i$
  containing the session id $\mi{lsid}$
  as a cookie, $c$,
  and $\mi{state}$.
  Such a request can only be the request $m$
  as shown above, hence there must have been the HTTPS response
  $m_\text{redir}$
  containing the values $c$
  and $\mi{state}$.
  Recall that we have the record $v$
  as shown above in the state of $i$.
  Such a record is only created in $i$
  if $\mathsf{authenticated}_\rho^{Q''}(b, r, u, i, \mi{lsid})$.
  Therefore, $u = u'$ in this case.
  
  \textit{$\mathsf{CHECK\_ID\_TOKEN}$
    was called in Line~\ref{line:rp-call-check-id-token-immediately}
    in Algorithm~\ref{alg:rp-oidc-http-request}:} In this case, the id
  token must have been contained in $m$
  and $m_\text{redir}$
  as above. Such an id token is only sent out in $m_\text{redir}$
  by $i$
  if $\mathsf{authenticated}_\rho^{Q''}(b, r, u, i, \mi{lsid})$.
  Therefore, $u = u'$ in every case.
\end{proof}

\subsubsection{Proof of Session Integrity for Authorization} \label{sec:proof-sess-integr-authz}
To prove that every OIDC web system with web attackers is secure
w.r.t.~session integrity for authorization, we assume that there
exists an OIDC web system with web attackers which is not secure
w.r.t.~session integrity for authorization and lead this to a
contradiction.

\begin{assumption} \label{asn:property-si-authz}
  There is a OIDC web system $\oidcwebsystem^w$
  with web attackers, a run $\rho$
  of $\oidcwebsystem^w$, a processing step $Q$ in $\rho$ with
  $$Q = (S, E, N) \xrightarrow[]{}
  (S', E', N')$$  (for some $S$, $S'$, $E$, $E'$, $N$, $N'$) a browser $b$
  that is honest in $S$,
  an IdP $i\in \fAP{IdP}$,
  an identity $u$
  that is owned by $b$,
  an RP $r\in \fAP{RP}$
  that is honest in $S$,
  a nonce $\mi{lsid}$,
  with (1) $\mathsf{usedAuthorization}_\rho^Q(b, r, i, \mi{lsid})$
  and there exists no processing step $Q'$
  in $\rho$
  (before $Q$)
  such that $\mathsf{started}_\rho^{Q'}(b, r, \mi{lsid})$,
  or (2) $i$
  is honest in $S$
  and $\mathsf{actsOnUsersBehalf}_\rho^Q(b, r, u, i, \mi{lsid})$
  and there exists no processing step $Q''$
  in $\rho$
  (before $Q$)
  such that
  $\mathsf{authenticated}_\rho^{Q''}(b, r, u, i, \mi{lsid})$.
\end{assumption}

\begin{lemma}\label{lemma:si-authz-contradiction}
  Assumption~\ref{asn:property-si-authz} is a contradiction.
\end{lemma}

\begin{proof}
  \textbf{(1)} We have that
  $\mathsf{usedAuthorization}_\rho^Q(b, r, i, \mi{lsid})$.
  With Definition~\ref{def:rp-uses-access-token} we have that $r$
  sent out the access token belonging to $\mi{lsid}$
  to $i$.
  This can only happen when the function $\mathsf{USE\_ACCESS\_TOKEN}$
  (Algorithm~\ref{alg:rp-use-access-token}) was called with
  $\mi{lsid}$.
  This function is called in
  Line~\ref{line:rp-call-use-access-token-1} of
  Algorithm~\ref{alg:rp-oidc-http-request} and in
  Line~\ref{line:rp-call-use-access-token-2} of
  Algorithm~\ref{alg:rp-oidc-http-response}. 

  In both cases, there must have been a request, say $m$,
  to $r$
  containing a cookie with the session id $\mi{lsid}$.
  In the former case, this is the request that is processed in the
  same processing step as calling the function
  $\mathsf{USE\_ACCESS\_TOKEN}$.
  In the latter case, there must have been an HTTPS request reference
  with the string $\str{TOKEN}$
  (cf. generic web server model, Algorithm~\ref{alg:simple-send}).
  Such a reference is only created in
  Line~\ref{line:rp-start-retrieve-code} of
  Algorithm~\ref{alg:rp-token-request}. To get to this point in the
  algorithm, a request as described above must have been received.
  Since we have web attackers (and no network attacker), it is easy to
  see that this request must have been sent by $b$.
  With the same reasoning as in the proof for session integrity for
  authentication, we now have that there exists a processing step $Q'$
  in $\rho$
  (before $Q$)
  such that $\mathsf{started}_\rho^{Q'}(b, r, \mi{lsid})$.

  \textbf{(2)} We have that $i$
  is honest and
  $\mathsf{actsOnUsersBehalf}_\rho^Q(b, r, u, i, \mi{lsid})$.
  From (1) we know that $r$
  must have received a request $m$
  from $b$
  containing a cookie with the session id $\mi{lsid}$.
  Therefore, we know that $m_\text{redir}$
  exists just as in the proof for
  Lemma~\ref{lemma:si-authn-contradiction}~(2). As in that proof, we
  have that
  $\mathsf{authenticated}_\rho^{Q''}(b, r, u', i, \mi{lsid})$
  for some identity $u'$. We therefore need to show that $u = u'$.

  With $\mathsf{actsOnUsersBehalf}_\rho^Q(b, r, u, i, \mi{lsid})$,
  we know that $r$
  must have called the function $\mathsf{USE\_ACCESS\_TOKEN}$
  with some access token $t$
  (Algorithm~\ref{alg:rp-use-access-token}). We further have that
  there is a term $g$
  such that $g \inPairing S(i).\str{records}$
  with $t \inPairing g[\str{access\_tokens}]$
  and $g[\str{subject}] \equiv u$.

  $\mathsf{USE\_ACCESS\_TOKEN}$
  can be called in Line~\ref{line:rp-call-use-access-token-2} in
  Algorithm~\ref{alg:rp-oidc-http-response} and in
  Line~\ref{line:rp-call-use-access-token-1} in
  Algorithm~\ref{alg:rp-oidc-http-request}. We now distinguish between
  these two cases.

  \textit{$\mathsf{USE\_ACCESS\_TOKEN}$
    was called in Line~\ref{line:rp-call-use-access-token-2} in
    Algorithm~\ref{alg:rp-oidc-http-response}:} When the function is
  called in this line, there must have been an HTTPS request reference
  with the string $\str{TOKEN}$
  (cf. generic web server model, Algorithm~\ref{alg:simple-send}).
  Such a reference is only created in
  Line~\ref{line:rp-start-retrieve-code} of
  Algorithm~\ref{alg:rp-token-request}. With
  Lemma~\ref{lemma:oidc-config-cache-integrity} we know that this
  HTTPS request was sent to the token endpoint (path $\str{/token}$)
  of $i$
  (because the issuer, stored in the login session record, is $i$).
  Since the token endpoint returned the access token $t$, and $i$
  is honest, there must have been a record in $i$,
  say $v$,
  with $v[\str{subject}] \equiv u$.
  In the request to the token endpoint, $r$
  must have sent a nonce $c$
  such that $v[\str{code}] \equiv c$.
  This request, as already mentioned, must have been sent in
  Line~\ref{line:rp-start-retrieve-code} of
  Algorithm~\ref{alg:rp-token-request}. This means, that there must
  have been an HTTPS request to $i$
  containing the session id $\mi{lsid}$
  as a cookie, $c$,
  and $\mi{state}$.
  Such a request can only be the request $m$
  as shown above, hence there must have been the HTTPS response
  $m_\text{redir}$
  containing the values $c$
  and $\mi{state}$.
  Recall that we have the record $v$
  as shown above in the state of $i$.
  Such a record is only created in $i$
  if $\mathsf{authenticated}_\rho^{Q''}(b, r, u, i, \mi{lsid})$.
  Therefore, $u = u'$ in this case.
  
  \textit{$\mathsf{USE\_ACCESS\_TOKEN}$
    was called in Line~\ref{line:rp-call-use-access-token-1} in
    Algorithm~\ref{alg:rp-oidc-http-request}:} In this case, the
  access token $t$
  must have been contained in $m$
  and $m_\text{redir}$
  as above. This access token is only sent out in $m_\text{redir}$
  by $i$
  if $\mathsf{authenticated}_\rho^{Q''}(b, r, u, i, \mi{lsid})$.
  Therefore, $u = u'$ in every case.
  
\end{proof}

\subsection{Proof of Theorem~\ref{thm:theorem-1}}
With
Lemmas~\ref{lemma:authn-contradiction},~\ref{lemma:authz-contradiction},~\ref{lemma:si-authn-contradiction}, and~\ref{lemma:si-authz-contradiction},
Theorem~\ref{thm:theorem-1} follows immediately.

\end{document}
